\documentclass[reqno,11pt]{amsart}
\usepackage{amsmath, latexsym, amsfonts, amssymb, amsthm, amscd,stmaryrd}
\usepackage{graphics,epsf,psfrag}
\numberwithin{equation}{section}
\DeclareMathOperator*{\argmin}{\mathrm{argmin}}

\newtheorem{theorem}{Theorem}[section]
\newtheorem{teorema}{Theorem}

\newtheorem{lemma}[theorem]{Lemma}
\newtheorem{proposition}[theorem]{Proposition}

\newtheorem{rem}[theorem]{Remark}

\newcommand{\lint}{\llbracket}
\newcommand{\rint}{\rrbracket}

\DeclareMathOperator{\sign}{\mathrm{sign}}

\newcommand{\dd}{\mathrm{d}}
\newcommand{\ind}{\mathbf{1}}

\renewcommand{\tilde}{\widetilde}
\renewcommand{\hat}{\widehat}
\newcommand{\cc}{\complement }

\newcommand{\cB}{{\ensuremath{\mathcal B}} }
\newcommand{\cA}{{\ensuremath{\mathcal A}} }
\newcommand{\cF}{{\ensuremath{\mathcal F}} }

\newcommand{\cE}{{\ensuremath{\mathcal E}} }
\newcommand{\cH}{{\ensuremath{\mathcal H}} }
\newcommand{\cC}{{\ensuremath{\mathcal C}} }
\newcommand{\cN}{{\ensuremath{\mathcal N}} }

\newcommand{\cD}{{\ensuremath{\mathcal D}} }

\newcommand{\cM}{{\ensuremath{\mathcal M}} }
\newcommand{\ubgb}{\overline{\gb}}

\newcommand{\bP}{{\ensuremath{\mathbf P}} }
\newcommand{\bE}{{\ensuremath{\mathbf E}} }

\newcommand{\grad}{{\ensuremath{\nabla}} }

\DeclareMathSymbol{\leqslant}{\mathalpha}{AMSa}{"36} 
\DeclareMathSymbol{\geqslant}{\mathalpha}{AMSa}{"3E} 
\DeclareMathSymbol{\eset}{\mathalpha}{AMSb}{"3F}     

\newcommand{\Var}{\mathrm{Var}}        
\newcommand{\sumtwo}[2]{\sum_{\substack{#1 \\ #2}}} 

\newcommand{\bbE}{{\ensuremath{\mathbb E}} }

\newcommand{\bbN}{{\ensuremath{\mathbb N}} }

\newcommand{\bbP}{{\ensuremath{\mathbb P}} }

\newcommand{\bbR}{{\ensuremath{\mathbb R}} }

\newcommand{\bbZ}{{\ensuremath{\mathbb Z}} }

\newcommand{\gb}{\beta}
\newcommand{\gd}{\delta}
\newcommand{\gep}{\varepsilon}       

\newcommand{\gvr}{\varrho}

\newcommand{\gG}{\Gamma}

\newcommand{\gD}{\Delta}

\newcommand{\go}{\omega}

\newcommand{\gO}{\Omega}
\newcommand{\gl}{\lambda}
\newcommand{\gL}{\Lambda}

\makeatletter
\def\captionfont@{\footnotesize}
\def\captionheadfont@{\scshape}

\long\def\@makecaption#1#2{%
  \vspace{2mm}
  \setbox\@tempboxa\vbox{\color@setgroup
    \advance\hsize-6pc\noindent
    \captionfont@\captionheadfont@#1\@xp\@ifnotempty\@xp
        {\@cdr#2\@nil}{.\captionfont@\upshape\enspace#2}%
    \unskip\kern-6pc\par
    \global\setbox\@ne\lastbox\color@endgroup}%
  \ifhbox\@ne 
    \setbox\@ne\hbox{\unhbox\@ne\unskip\unskip\unpenalty\unkern}%
  \fi
  \ifdim\wd\@tempboxa=\z@ 
    \setbox\@ne\hbox to\columnwidth{\hss\kern-6pc\box\@ne\hss}%
  \else 
    \setbox\@ne\vbox{\unvbox\@tempboxa\parskip\z@skip
        \noindent\unhbox\@ne\advance\hsize-6pc\par}%
\fi
  \ifnum\@tempcnta<64 
    \addvspace\abovecaptionskip
    \moveright 3pc\box\@ne
  \else 
    \moveright 3pc\box\@ne
    \nobreak
    \vskip\belowcaptionskip
  \fi
\relax
}
\makeatother
\def\writefig#1 #2 #3 {\rlap{\kern #1 truecm
\raise #2 truecm \hbox{#3}}}

\newcommand{\tf}{\textsc{f}}

\newcommand{\be}{{\mathbf e}}
\begin{document}

\title[Pinning for lattice free field]{Pinning and disorder relevance for the lattice Gaussian Free Field II: 
the two dimensional case}

\address{IMPA - Instituto Nacional de Matem\'atica Pura e Aplicada,
Estrada Dona Castorina 110,
Rio de Janeiro / Brasil 22460-320}
\email{lacoin@impa.br}
\author{Hubert Lacoin}

\begin{abstract}
This paper continues a study initiated in \cite{cf:GL}, on the localization transition of a lattice free field on $\bbZ^d$
interacting with a quenched disordered substrate that acts on the interface when its height is close to zero.
The substrate has the tendency to localize or repel the interface at different sites.
A transition takes place when the average pinning potential $h$ goes past a threshold $h_c$: from a delocalized phase $h<h_c$, 
where the field is macroscopically repelled by the substrate 
to a localized one $h>h_c$ where the field sticks to the substrate.
Our goal is to investigate the effect of the presence of disorder on this phase transition.
We focus on the two dimensional case $(d=2)$ for which we had obtained so far only limited results.
We prove that the value of $h_c(\gb)$ is the same as for the annealed model, for all values of $\gb$ and that in a neighborhood of $h_c$.
Moreover we prove that, in contrast with the case $d\ge 3$ where the free energy has a quadratic behavior near the critical point,  
the phase transition is of infinite order
$$\lim_{u\to 0+} \frac{ \log \tf(\gb,h_c(\gb)+u)}{(\log u)}= \infty.$$
\\[10pt]
  \\[10pt]
\end{abstract}

\maketitle

\tableofcontents

\newpage

\section{Introduction}

The aim of statistical mechanics is  to obtain a qualitative understanding of natural phenomena of phase transitions by the study of simplified models, 
often built on a lattices. 
In general the Hamiltonian of a model of statistical mechanics is left invariant by the lattice symmetries: 
a prototypical example being the Ising model describing a ferromagnet. 

\medskip

However, one might argue that materials which are found in nature are usually not completely homogeneous and for this reason, physicists where led to considering 
systems in which the interaction terms, for example the potentials between nearest neighbor spins, are chosen by sampling a random
field -- which we call {\sl disorder} -- with good ergodic properties, often even a field of independent identically distributed random variables.
An important question which arises is thus whether the results concerning the phase transition 
obtained for a model with homogeneous interactions referred to as \textit{the pure system} 
(e.g.\ the Onsager solution of the two dimensional Ising Model) remain valid when a system where randomness of a very small amplitude is introduced.

\medskip

In  \cite{cf:Hcrit} A.~B.~Harris, gave a strikingly simple heuristical argument, based on renormalization theory consideration, 
to predict the effect of the introduction of a small amount of the system: in substance Harris' criterion predict that if the phase transition of the pure system 
is sufficiently smooth, it will not be affected by small perturbation (disorder is then said to be \textit{irrelevant}), 
while in the other cases the behavior of the system 
is affected by an arbitrary small addition of randomness (disorder is \textit{relevant}).
To be complete, let us mention also the existence of a boundary case for which the criterion yields no prediction (the \textit{marginal disorder} case).
The criterion however does not give a precise prediction concerning the nature of the phase transition when the disorder is relevant. 

\medskip

The mathematical verification of the Harris criterion is a very challenging task in general. In the first place, it can only be considered for the few special models of 
statistical mechanics for which we have a rigorous understanding of the critical properties of the pure system.
In the 
last twenty years this question has been addressed, first by theoretical physicists (see e.g. \cite{cf:DHV} and references therein) and then by mathematicians
\cite{cf:KZ, cf:AZnew, cf:Ken, cf:QH2, cf:DGLT, cf:GLT, cf:GLT2, cf:GT_cmp,  cf:L, cf:Trep}
(see also \cite{cf:GB,cf:G} for reviews), 
for a simple model of a 1-dimensional 
interface interacting with a substrate: for this model the interface is given by the graph of a random walk which takes random energy 
rewards when it touches a defect line. In this case, the pure system has the remarkable quality of being what physicists call
{\sl exactly solvable}, meaning that there exists an explicit expression for the  free energy \cite{cf:Fisher}.

\medskip

This model under consideration in the present paper can be seen as a high-dimension generalization of the RW pinning model.
The random walk is replaced by a random field $\bbZ^d\to \bbR$, and the random energies are collected when 
the graph of the field is close to the hyper-plane $\bbZ^d \times \{0\}$. 
While the pure  model is not exactly solvable in that case, it has been studied in details
and the nature of the phase transition is well known \cite{cf:BB, cf:BDZ1, cf:BV, cf:CV,  cf:Vel}.

\medskip

On the other hand, the study of the disordered version of the model is much more recent
\cite{cf:CM1, cf:CM2, cf:GL}. 

 \medskip
 
In \cite{cf:GL}, we gave a close to complete description of the free energy diagram of the 
disordered model when $d\ge 3$:
\begin{itemize}
\item We identified the value of the disordered critical point, which is shown to coincide with that of the associated annealed model, regardless of the amplitude of disorder.
\item We proved that for Gaussian disorder, the behavior of the free energy close to $h_c$ is quadratic, in contrasts with the annealed model for which the transitition is of first order.
\item In case of general disorder, we proved that the quadratic upper-bound still holds, and found a polynomial lower bound with a different exponent.
\end{itemize}
Let us stress that the heuristic of our proof strongly suggests that the behavior of the free energy 
should be quadratic for a suitable large class of environments (those who satisfy a second moment assumption similar to \eqref{eq:assume-gl}).

\medskip

In the present paper, we choose to attack the case $d=2$, for which only limited results were obtained so far. We have seen in the proof of the main result \cite{cf:GL} that the critical behavior of 
the model is very much related to the extremal process of the field. The quadratic behavior of the free-energy in \cite[Theorem 2.2]{cf:GL} comes 
from the fact that high level 
sets of the Gaussian free field for $d\ge 3$ look like a uniformly random set with a fixed density (see \cite{cf:CCH}).
In dimension $2$ however, the behavior of the extremal process is much more intricate, with a phenomenon of clustering in the level sets 
(see \cite{cf:BL, cf:DZ2, cf:Dav} or also
\cite{cf:ABK} for a similar phenomenon for branching Brownian Motion). 
This yields results of a very different nature.

\section{Model and results}

Given $\Lambda$ be a finite subset of $\bbZ^d$, we
 let $\partial \gL$ denote the internal boundary of $\Lambda$, 
$\mathring{\gL}$ the set of interior points of $\gL$, and $\partial^-\gL$ the set of point which are adjacent to the boundary,
\begin{equation}\label{boundary}\begin{split}
\partial \gL&:=\{x \in \gL : \,  \exists y\notin \gL, \ x\sim y  \},\\
\mathring{\gL}&:=\gL \setminus \partial \gL,\\
\partial^- \gL&:=\{x \in \mathring{\gL} : \,  \exists y\in \partial \gL, \ x\sim y  \}.
\end{split}\end{equation}
In general some of these sets could be empty, but throughout this work $\gL$ is going to be a large 
square. Given $\hat \phi: \bbZ^d \to \bbR$, we define
$\bP^{\hat \phi}_{\gL}$ to be the law of the lattice Gaussian free field $\phi=(\phi_x)_{x\in \Lambda}$ with boundary condition $\hat \phi$
 on $\partial \gL$. The field $\phi$ is a random function from $\gL$ to $\bbR$. 
It is satisfies 
\begin{equation}
 \phi_x\,:=\, \hat \phi_x \quad  \text{ for every } x \in  \partial \gL,
\end{equation}
and the distribution of $(\phi_x)_{x\in \mathring \gL}$ is given by
\begin{equation}
\label{density}
\bP^{\hat \phi}_\gL(\dd \phi)=\frac{1}{\mathcal Z^{\hat \phi}_{\gL}}
  \exp\left(-\frac 1 2 \sumtwo{(x,y)\in (\gL)^2 \setminus (\partial \gL)^2 }{x\sim y}\frac{ (\phi_x-\phi_y)^2 }{2} \right)
\prod_{x\in \mathring{\gL}} \dd \phi_x \, ,
\end{equation} 
where $\prod_{x\in \mathring{\gL}} \dd \phi_x$ denotes the Lebesgue measure on $\bbR^{\mathring{\gL}}$ and
\begin{equation}\label{eq:defcalz}
\mathcal Z^{\hat \phi}_{\gL}:= \int_{\bbR^{\mathring{\gL}}} 
\exp\left(-
 \frac 1 2 \sumtwo{(x,y)\in (\gL)^2 \setminus (\partial \gL)^2 }{x\sim y}
\frac{ (\phi_x-\phi_y)^2 }{2} \right) \prod_{x\in\mathring{\gL}} \dd \phi_x \, ,
\end{equation}
(one of the two $(1/2)$ factors is present to compensate that the edges are counted twice in the sum, the other one being the one usually present for Gaussian densities).
In what follows we consider the case
$$\Lambda=\Lambda_N:=\{0,\dots,N\}^d,$$ for some $N\in \bbN$.
Note that we have  $$\mathring{\gL}_N:= \{1,\dots, N-1\}^d.$$
We also introduce the notation $\tilde \gL_N:= \{1,\dots, N\}^d$, and 
we simply write $\bP^{\hat \phi}_N$ for $\bP^{\hat \phi}_{\gL_N}$.
We drop $\hat \phi$ from our notation in the case where we consider zero boundary condition $\hat \phi \equiv 0$.

\medskip

We let $\go=\{\go_x\}_{x \in \bbZ^d}$ be the realization of a family of  IID square integrable centered random variables (of law $\bbP$). We
assume that they have finite exponential moments, or more precisely, that there exist constants $\gb_0, \ubgb\in (0, \infty]$ such that
\begin{equation}
\label{eq:assume-gl}
 \gl(\gb)\, :=\, \log \bbE[e^{\gb \go_x}]\,  < \,  \infty\ \text{ for  every } \gb \in  (-\gb_0\, 2\ubgb]\, .
 \end{equation}
For $x\in \gL_N$  set  $\delta_x:= \ind_{[-1,1]}(\phi_x)$.
For $\gb>0$ and $h\in \bbR$, we define  a modified measure $\bP_{N, h}^{\gb,\go,\hat \phi}$ via the density
\begin{equation}
\label{eq:modmeas}
\frac{\dd \bP^{\gb,\go, \hat \phi}_{N,h}}{\dd \bP^{\hat \phi}_N}(\phi)=\frac{1}{Z^{\gb,\go,\hat \phi}_{N,h}}\exp\left( \sum_{x\in  \tilde \gL_N} 
(\gb \go_x-\gl(\gb)+h)\delta_x\right)\, ,
\end{equation}
 where
\begin{equation}
\label{eq:modZ}
Z^{\gb,\go,\hat \phi}_{N,h}:=\bE_N\left[ \exp\left( \sum_{x\in  \tilde \gL_N} (\gb \go_x-\gl(\gb)+h)\delta_x\right)\right].
\end{equation}
Note that in the definition of $\bP^{\gb,\go,\hat \phi}_{N,h}$, the sum $\left(\sum_{x\in  \tilde \gL_N}\right)$ can be replaced by either 
$\left(\sum_{x\in  \gL_N}\right)$ or $\left(\sum_{x\in  \mathring{\gL}_N}\right)$
as these changes affect only the partition function. In the case where 
$\hat \phi \equiv 0$, we drop the corresponding superscript it from the notation.
In the special case where $\gb=0$, we simply write $\bP^{\hat \phi}_{N,h}$ and $Z^{\hat \phi}_{N,h}$ 
for the pinning measure and partition function (as they do not depend on $\go$) respectively.
This case is referred to as the \textsl{pure} (or homogeneous) model. When $\gb>0$, \eqref{eq:modmeas} 
defines the pinning model with \textsl{quenched} disorder.

\subsection{The free energy}

The important properties of the system are given by the asymptotic behavior of the partition function, or more precisely 
by the free energy.
The existence of quenched free energy for the disordered model has been proved in \cite[Theorem 2.1]{cf:CM1}. 
We recall this result here together with some basic properties

\medskip

\begin{proposition}\label{freen}
The free energy 
\begin{equation}
\label{eq:freen}
\tf(\gb,h):=\lim_{N\to \infty} \frac{1}{N^d}\bbE \left[\log Z^{\gb,\go}_{N,h}\right]
\stackrel{\bbP(\dd \go)-a.s.}{=}  \lim_{N\to \infty} \frac{1}{N^d}\log Z^{\gb,\go}_{N,h} \, ,
\end{equation}
exists (and is self-averaging).  It is a convex, nonnegative, nondecreasing function of $h$.
Moreover there exists a $h_c(\gb)\in (0,\infty)$ which is such that 
\begin{equation}
 \tf(\gb,h)\begin{cases} =0 \text{ for } h\le h_c(\gb),\\ 
            >0 \text{ for } h> h_c(\gb).
           \end{cases}
\end{equation}
\end{proposition}

Let us briefly explain why $h_c(\gb)$ marks a transition on the large scale behavior of $\phi$ under $\bP^{\gb,\go}_{N,h}$.
A simple computation gives
\begin{equation}
\partial_h\left(  \frac{1}{N^d}  \log Z^{\gb,\go}_{N,h}\right)= \frac{1}{N^d}\sum_{x\in \tilde \gL_N} \bE^{\gb,\go}_{N,h} \left[ \delta_x \right].
\end{equation}
Hence by convexity, we have 
\begin{equation}
 \partial_h \tf(\gb,h)=\lim_{N\to \infty}\frac{1}{N^d}\sum_{x\in \tilde \gL_N} \bE^{\gb,\go}_{N,h} \left[ \delta_x \right],
\end{equation}
for the $h$ for which $\tf(\gb,h)$ differentiable (for the hypothetical countable set where  $\partial_h \tf(\gb,h)$ may not exist, we can replace $\lim$ by $\liminf$ resp. $\limsup$, $=$ by $\le$ resp. $\ge$ and consider the
left- resp. right-derivative in the above equation).

\medskip

For $h>h_c(\gb)$, we have $\partial_h \tf(\gb,h)>0$ by convexity and thus the expected number of point in contact with the substrate is asymptotically of order 
$N^d$. On the contrary when $h<h_c(\gb)$, the asymptotic expected contact fraction vanishes when $N$ tends to infinity.

\medskip

Note that the whole model is perfectly defined for all $d\ge 1$. However, the case $d=1$, which is a variant of the random walk pinning model which as mentioned in the introduction
was the object of numerous studies in the literature. However, 
the effect of disorder in dimension $1$ being quite different, in the remainder of the paper, 
we prove results for the case $d=2$ and discuss how they compare with those obtained in the more  related case $d\ge 3$ \cite{cf:GL}.

\subsection{The pure model}

In the case $\gb=0$, we simply write $\tf(h)$ for $\tf(0,h)$. 
In that case 
the behavior of the free energy is known in details (see \cite[Fact 2.4]{cf:CM1} and also \cite[Section 2.3 and Remark 7.10]{cf:GL} for a full proof for $d\ge 3$).
We summarize it below.

\begin{proposition}\label{propure}
For all $d\ge 1$, we have $h_c(0)=0$ and moreover 
\begin{itemize}
 \item [(i)] For $d=2$
 \begin{equation}\label{pure}
\tf(h)\stackrel{h\to 0+}{\sim} \frac{\sqrt{2} h}{\sqrt{|\log h|}},
\end{equation}
 \item [(ii)] For $d\ge 3$
 \begin{equation}\label{pure3}
\tf(h)\stackrel{h\to 0+}{\sim} c_d h,
\end{equation}
where $c_d:= \bP[\sigma_d \cN \in [-1,1]]$ and $\sigma_d$ is the standard deviation for the infinite volume free field in $\bbZ^d$.
\end{itemize}
\end{proposition}

To be more precise $\sigma_d:=\sqrt{G^0(x,x)}$ where $G^0$ is the Green function defined in \eqref{greenff}.
The result in dimension $2$ is well known folklore to people in the fields, but as 
to our knowledge, no proof of it is available in the literature. For this reason we present a short one in Appendix \ref{secpropure}.

\subsection{The quenched/annealed free energy comparison}

Using Jensen's inequality, we can for every $\gb\ge 0$, compare the free energy to that of the annealed system, which is the one associated to the 
averaged partition function $\bbE \left[ Z^{\gb,\go}_{N,h}\right]$,
\begin{equation}\label{anealed}
\tf(\gb,h)=\lim_{N\to \infty} \frac{1}{N^d}\bbE \left[ \log Z^{\gb,\go}_{N,h}\right] \le \lim_{N\to \infty} \frac{1}{N^d}\log \bbE \left[ Z^{\gb,\go}_{N,h}\right].
\end{equation}
Our choice of parametrization implies 
\begin{equation}
\bbE \left[ Z^{\gb,\go}_{N,h}\right]=\bE_N\left[ \bbE\left[e^{\sum_{x\in \tilde \gL_N}( \go_x-\gl(\gb)+h)\delta_x}\right] \right]=
\bE_N\left[ e^{\sum_{x\in \tilde \gL_N}h\delta_x}\right]=
Z_{N,h},
\end{equation}
and thus for this reason we have
\begin{equation} \label{annehilde}
\tf(\gb,h)\le \tf(h) \quad \text{and} \quad
 h_c(\gb)\, \ge\, 0.
\end{equation}
It is known that the inequality \eqref{anealed} is strict: for $h>0$, we have $\tf(\gb,h)< \tf(h)$ in all dimensions (cf.\ \cite{cf:CM1}).
However we can ask ourselves if the behavior of the model with quenched disorder is similar to that of the annealed one in several other ways
\begin{itemize}
 \item [(a)] Is the critical point of the quenched model equal to that of the annealed model (i.e.\ is $h_c(\gb)=0$)?
 \item [(b)] Do we have a critical exponent for the free energy transition: do we have $$\tf(\gb,h_c(\gb)+u)\stackrel{u\to 0+} \sim u^{\nu+o(1)},$$ 
 and is $\nu$ equal to one, like for the annealed model (cf. Proposition \ref{propure})?
\end{itemize}
This question has been almost fully solved in the case $d\ge 3$. Let us display the result here 
\begin{teorema}[{\cite[Theorem 2.2]{cf:GL}}]
  For $d\ge 3$, for every $\gb\in [0,\bar\beta]$ we have 
 \begin{itemize}
  \item [(i)] $h_c(\gb)=0$ for all values of $\gb>0$.
  \item [(ii)] If $\go$ is Gaussian, there exist positive constants  $c_1(\gb)<c_2(\gb)$ such that for all $h\in (0,1)$.
  \begin{equation}
   c_1(\gb) h^2 \le \tf(\gb,h)\le c_2(\gb) h^2.
  \end{equation}
 \item[(iii)] In the case of general $\go$, for all there exist positive constants  $c_1(\gb)<c_2(\gb)$ such that for all $h\in(0,1)$
   \begin{equation}
   c_1(\gb) h^{66d}\le \tf(\gb,h)\le c_2(\gb) h^2.
  \end{equation}
 \end{itemize}
\end{teorema}

\begin{rem}
We strongly believe that the quadratic behavior holds for every $\go$ as soon as $\gl(2\gb)<\infty$, and the Gaussian assumption is mostly technical.
However, if $\gl(2\gb)=\infty$, we believe that the model is in a different universality class and the critical exponent depends on the tail of the distribution of the variable
$\xi:=e^{\gb \go_0}$.
\end{rem}

The aim of the paper is to provide answers in the case of dimension $2$. 

\subsection{The main result}

We present now the main achievement of this paper.
We prove that similarly to the $d\ge 3$ case, the critical point $h_c(\gb)$ coincides with the annealed one for every value of $\gb$
(which is in contrast with the case $d=1$ where the critical points differs for every $\gb>0$ \cite{cf:GLT}).
However, we are able to prove also that the critical behavior of the free energy is not quadratic, $\tf(\gb,h)$ is becomes smaller than 
any power of $h$ in a (positive) neighborhood of $h=0$. This indicates that the phase transition is of infinite order. 

\begin{theorem}\label{mainres}
When $d=2$, for every $\gb \in[0,\bar \gb]$ the following holds
\begin{itemize}
 \item [(i)]  We have $h_c(\gb)=0$. 
 
 \item[(ii)]   We have 
 \begin{equation}
 \lim_{h\to 0+} \frac{\log \tf(\gb,h)}{\log h}=\infty.
 \end{equation}
 More precisely, there exists $h_0(\gb)$  such that for all $h\in (0,h_0(\gb))$
 \begin{equation}\label{breaks}
\exp\left(- h^{-20} \right)\le   \tf(\gb,h)  \le \exp\left( - | \log h|^{3/2}  \right).
 \end{equation}
\end{itemize}
\end{theorem}

\begin{rem}
We do not believe that either bound in \eqref{breaks} is sharp. However it seems to us that the strategy used for the 
lower-bound is closer to capture the behavior of the field.
We believe that the true behavior of the free energy might be given by 
$$\tf(\gb,h) \approx \exp( h^{-1+o(1)}).$$
While a lower bound  of this type might be achieved by optimizing the proof presented in the present paper (but this would require some significant technical work), 
we do not know how to obtain a significant improvement on the upper-bound.
\end{rem}

\subsection{Co-membrame models in two dimension}

Like in \cite{cf:GL}, it worthwhile to notice that the proof of the results of the present paper can be adapted to 
a model for with a different localization mechanism.
It is the analog of the model of a copolymer in the proximity of the interface between selective solvents, see \cite{cf:Bcoprev,cf:coprev} and references therein. 
For this model given  a realization of $\go$ and two fixed parameters $\gvr, h>0$, the measure is defined via the following density
\begin{equation}
\label{eq:modmeascop}
\frac{\dd \check\bP^{\go,\gvr}_{N,h}}{\dd \bP_N}\, \propto\, \exp\left( \gvr \sum_{x\in  \tilde \gL_N} ( \go_x+h)\sign \left(\phi_x \right)\right)\, ,
\end{equation}
where we assume  $\sign(0)=+1$.
A natural interpretation of the model is that the graph of $(\phi_x)_{x\in \gL_N}$ models a membrane lying between two solvents $A$ and $B$ which fill
the upper and lower half-space respectively:
for each point of the graph, the quantity $\go_x+h$ describes the energetic preference 
for one solvent of the corresponding portion of the membrane (A if $\go_x+h>0$ and B if $\go_x+h<0$).
As $h$ is positive and $\go_x$ is centered, there is, on average, a preference for solvent A (by symetry this causes no loss of generality).

\medskip

If $\bbP[(\go_x<-h)>0]$,
there is a non-trivial competition between energy and entropy: the interaction with the solvent gives an incentive for the field $\phi$ to stay close to the interface so that its
sign can match as much as possible that of $\go+h$, but such a strategy might be valid only if the energetic rewards it brings is superior to the entropic cost of the 
localization.

\medskip

A more evident analogy with the pinning measure 
\eqref{eq:modmeas} can be made by observing that we can write
\begin{equation}
\label{eq:modcop}
\frac{\dd \check\bP^{\go,\gvr}_{N,h}}{\dd \bP_N}\, =\, \frac1{\check Z^{\go,\gvr}_{N,h}} \exp\left( -2\gvr \sum_{x\in  \tilde \gL_N} ( \go_x+h)\gD_x
\right)\, ,
\end{equation}
where $\gD_x:= (1- \sign(\phi_x))/2$, that is $\gD_x$ is the indicator function that $\phi_x$ is in the lower half plane.
It is probably worth stressing that from \eqref{eq:modmeascop} to \eqref{eq:modcop} there is a non-trivial (but rather simple) change in energy.
 And in the form \eqref{eq:modcop}.
In particular, the strict analog of Proposition~\ref{freen} holds -- the free 
energy in this case is denoted by $\check\tf(\gvr, h)$ -- and, precisely like for the pinning case, one sees that
$\check\tf(\gvr, h)\ge 0$. We then set $\check h_c(\gvr):= \inf\{h>0:\, \check\tf(\gvr, h)=0\}$.
Adapting the proof for the lower-bound in \eqref{breaks} we can identify the value of  $\check h_c(\gvr)$.


\medskip

\begin{theorem}
\label{th:cop}
For $d= 2$, for any $\gvr\in(0, \bar \gb/2)$ we have
\begin{equation}
\label{eq:cop}
\check h_c(\gvr)\, =\, \frac1{2\gvr}  \gl (-2\gvr)\, .
\end{equation}
Moreover \eqref{breaks}, with $\tf(\gb,h)$ replaced by 
$\check\tf (\gvr, h_c(\gvr)-h)$, holds true.
\end{theorem}

Note that while pure co-membrane model (i.e.\ with no disorder) displays a first order phase transition in $h$, the above result underlines that
the transition becomes of infinite order in the presence of an arbitrary small quantity of disorder.
Note that this result differs both from the one obtained in dimension $d\ge 3$ (for which the transition is shown to be quadratic at least for Gaussian environment
\cite[Theorem 2.5]{cf:GL}), and that in dimension $1$: for the copolymer model based on renewals presented in \cite{cf:Bcoprev}, 
$\frac1{2\gvr}  \gl (-2\gvr)$ is in most cases a strict upper-bound on 
$\check h_c(\gvr)$ 
(see e.g.\ the results in \cite{cf:Tcopo}).

\medskip

The proof of Theorem \ref{th:cop} is not given in the paper but it can be obtained with straightforward modification, from that of Theorem \ref{mainres}.

\subsection{Organization of the paper}

The proof of the upper-bound and of the lower-bound on the free energy presented in Equation \eqref{breaks}  are largely independent.
However some general technical results concerning the covariance structure of the free field are useful in both proofs, 
and we present these in Section \ref{toolbox}. Most of the proofs for results presented in this section are in Appendix \ref{appendix}.

The proof of the upper-bound is developed in Section \ref{seclower}.
The proof of the lower-bound is  spreads from Section \ref{finicrit} to \ref{intelinside}.
In Section \ref{finicrit} we present an estimate on the free energy in terms of 
a finite system with ``stationary'' boundary condition. In Section \ref{decompo}, we give a detailed sketch of the proof of the lower-bound
based on this finite volume criterion, divided into several steps. The details of these steps are covered in Section \ref{liminouze} and \ref{intelinside}.

\medskip

For the proof of both the upper and the  lower-bound, we need fine results on the structure of the free field. 
Although these results or their proof cannot directly be extracted from the existing literature, our proof (especially the techniques developed in Section \ref{intelinside})
is largely based on tools that were developed in the numerous
study on extrema and extremal processes  of the two dimensional free field \cite{cf:BDG, cf:BDZ, cf:Dav, cf:DZ2} and other $\log$-correlated Gaussian processes
\cite{cf:A, cf:AS, cf:ABK,  cf:brams, cf:Mad} (the list of references being far from being complete). In particular for the lower bound,  we present 
an \textit{ad-hoc} decomposition of the field
in Section \ref{decompo} and then exploit decomposition to apply a conditioned second moment technique, similarly to what is done e.g.\ 
in \cite{cf:AS}.

\medskip

For the upper-bound, we also make use a change of measure machinery inspired by a similar techniques developed in the study of disordered pinning model
\cite{cf:BL, cf:DGLT, cf:GLT, cf:GLT2} and adapted successfully  to the study of other models \cite{ cf:BT, cf:BS1,  cf:BS2, cf:L2, cf:L3, cf:YZ}.

\medskip

\section{A toolbox}\label{toolbox}

\subsection{Notation and convention}

Throughout the paper, to avoid a painful enumeration, 
we use $C$ to denote an arbitrary constant which is not allowed to depend on the value of $h$ or $N$ nor on the realization of $\go$.
Its value may change from one equation to another. For the sake of clarity, we try to write $C(\gb)$ when the constant may depend on $\gb$.
When a constant has to be chosen small enough rather than large enough, we may use $c$ instead of $C$.

\medskip

For $x=(x_1,x_2)\in \bbZ^2$ we let $|x|$ denote its $l_1$ norm.
\begin{equation}
 |x|:= |x_1|+|x_2|.
\end{equation}
The notation $|\cdot|$ is also used to denote the cardinal of a finite set as this should yield no confusion.

\medskip

\noindent If $A\subset \bbZ^2$ and $x\in \bbZ^2$ we set 
\begin{equation}\label{distA}
 d(x,A):=\min_{y\in A} |x-y|.
\end{equation}
We use double brackets to denote interval of integers,
that for $i<j$ in $\bbZ$
\begin{equation}
 \lint i,j\rint := [i,j]\cap \bbZ=\{i,i+1,\dots,j\}.
\end{equation}
If $(A_i)_{i=1}^k$ is a finite family of events, we refer to the following inequality as 
\textit{the union bound}.
\begin{equation}
 \bbP(\cup_{i=1}^k A_i)\le \sum_{i=1}^k \bbP(A_i).
\end{equation}
We let $(X_t)_{t\ge 0}$ denote continuous time simple random walk on $\bbZ^d$ whose generator $\gD$ is the lattice Laplacian defined by 
\begin{equation}\label{laplace}
 \gD f(x):= \sum_{y\sim x} \big( f(y)-f(x) \big)
\end{equation}
and we let $P^x$ denote its law starting from $x\in \bbZ^d$.
We let $P_t$ denote the associated heat-kernel
\begin{equation}\label{heat}
 P_t(x,y)=P^x(X_t=y).
\end{equation}
If $\mu$ denote a probability measure on a space $\gO$, and $f$ a measurable function on $\gO$ we denote the expectation of $f$ by 
\begin{equation}
 \mu(f) = \int_{\gO}  f(\go) \mu( \dd \go),
\end{equation}
with an exception where the probability measure is denoted by the letter $P$, in that case we use $E$ for the expectation.

\medskip

\noindent If $\cN(\sigma)$ is a Gaussian of standard deviation $\sigma$, it is well known that we have 
\begin{equation}\label{gtail}
 P\left[ \cN(\sigma) \ge u \right] \le \frac{\sigma}{ \sqrt{2\pi}u}e^{-\frac{u^2}{2\sigma^2}}.
\end{equation}
We refer to the Gaussian tail bound when we use this inequality.

\subsection{The massive free field}\label{secmass}

In this section we quickly recall the the definition and some basic properties of the massive free field.
Given $m>0$,  and a  set $\gL\subset \bbZ^d$ and a function $\hat \phi$, we define the law $\bP^{m,\hat \phi}_{\gL}$ of the massive free field on $\gL$ with boundary condition $\hat \phi$ and mass $m$ 
as follows: it is absolutely continuous w.r.t $\bP^{\hat \phi}_{\gL}$ and
 \begin{equation}\label{massivedensity}
\frac{\dd \bP^{m,\hat \phi}_{\gL} }{\dd \bP^{\hat \phi}_{\gL}}(\phi):= \frac{1}{\bE^{\hat \phi}_{\gL}
\left[  \exp\left(-m^2 \sum_{x\in \mathring \gL} \phi_x^2\right) \right] } \exp\left(-m^2 \sum_{x\in \mathring \gL} \phi_x^2\right).
\end{equation}
We let $\bP_N^{m,\hat \phi}$ denote the law of the massive field on $\gL_N$.
(in the special case $\hat \phi_x\equiv 0$,  $\hat \phi$ is omitted in the notation).

\medskip

We let  $\bP^m$ denote the law of the centered infinite volume massive free field $\bbZ^d$, which is the limit of $\bP^{m}_{\gL}$ when $\gL \to \bbZ^d$ (see 
Section \ref{hkernel} for a proper definition with the covariance function).
We will in some cases have to choose the boundary condition $\hat \phi$ itself to be random 
and distributed like  an infinite volume centered massive free field (independent $\phi$), in which case we denote its law by $\hat \bP^m$ instead of $\bP^m$.

\medskip

Note that the free field and its massive version satisfy a Markov spatial property. In particular the law of 
$(\phi)_{x\in \gL_N}$ under $\hat \bP^{m} \times\bP^{m,\hat \phi}_N$
is the same as under the infinite volume measure  $\bP^{m}$.

\subsection{Getting rid of the boundary condition}\label{grbc}

Even if the definition of the free energy given in Proposition \ref{freen} 
is made in terms of the partition function with $\hat \phi\equiv 0$ it turns out that our methods to obtain 
upper and lower bounds involve considering non-trivial boundary conditions (cf. Proposition \ref{scorpiorizing} and Proposition \ref{th:finitevol}).

\medskip

However, it turns out to be more practical to work with a fixed law for the field and not one that depends on $\hat \phi$. 
Fortunately, given a boundary condition $\hat \phi$ the law of
$\bP^{m,\hat \phi}_N$ can simply be obtained by translating the field with $0$ boundary condition by a function that depends only on $\hat \phi$.
This is a classical property of the free field but let us state it in details.
As the covariance function of $\phi$ under 
$\bP^{m,\hat \phi}_N$ and $\bP^{m}_N$ are the same, we have 
we have 
\begin{equation}\label{transkix}
\bP^{m,\hat \phi}_N[\phi \in \cdot \ ]=\bP^{m}_N[\phi + H^{m,\hat \phi}_N \in \cdot \ ],
\end{equation}
where
\begin{equation}
 H^{m,\hat \phi}_N(x):=\bE^{m,\hat \phi}_N[\phi(x)].
\end{equation}
It is not difficult to check that  $H^{m,\hat \phi}_N$ must be a solution of the system (recall \eqref{laplace})
\begin{equation}\label{defH}\begin{cases}
H(x):= \hat \phi(x),  & x \in \partial \gL_N,\\
\gD H(x) =m^2  H^{m,\hat \phi}_N,  \quad & x \in \mathring \gL_N.
 \end{cases}
\end{equation}
We simply write $H^{\hat \phi}_N(x)$ when $m=0$.
The solution of \eqref{defH} is unique and $H^{m,\hat \phi}_N$ has the following representation:
consider $X_t$  the simple random walk on $\bbZ^d$ and for $A\subset \bbZ^d$ let $\tau_A$ denotes the first 
hitting of $A$. We have 
\begin{equation}\label{RWrepresent}
H^{\hat \phi}_N(x):= E_x\left[ e^{-m^2 \tau_{\partial \gL_N}}\hat \phi\left(X_{\tau_{\partial \gL_N}}\right)\right].
\end{equation}
Given $\hat \phi$ and $x\in \tilde \gL_N$, we introduce the notation
\begin{equation}\label{deltaf}
\delta^{\hat \phi}_x:= \ind_{[-1,1]}(\phi_x+ H^{\hat \phi}_N(x)).
\end{equation}
In view of \eqref{transkix} an alternative way of writing the partition function is 
\begin{equation}\label{migrate}
Z^{\gb,\go,\hat \phi}_{N,h}= \bE_N \left[ e^{\sum_{x\in \tilde \gL_N} (\gb\go_x-\gl(\gb)+h) \delta^{\hat \phi}_x}  \right].
\end{equation}
In some situation the above expression turns our to be handier than the definition \eqref{eq:modZ}.

\subsection{Some estimates on Green functions and heat Kernels}\label{hkernel}

In this section we present some estimates on the covariance function of the free field and massive free field in dimension $2$, which will be useful in the course of the proof.
These are not new results,  but rather variants of existing estimates in the literature (see e.g \cite[Lemma 2.1]{cf:Dav}). 

\medskip

The covariance kernel of the infinite volume free field with mass $m>0$ in $\bbZ^2$ or $m\ge 0$  in $\gL_N$ is given by the Green function $G^m$ which is the inverse of $\gD-m^2$
 (this can in fact be taken as the definition of the infinite volume free field, requiring in addition that it is centered).
 The covariance function of the field under the measure $\bP_N$ is $G^{m,*}$ which is the inverse of $\gD-m^2$  
 with Dirichlet boundary condition on $\partial \gL_N$. 
Both of these functions can be represented as integral of the heat kernel \eqref{heat},
we have  
\begin{equation}\label{greenff}
\begin{split}
 \bE^m[\phi(x)\phi(y)]&=\int_{0}^{\infty}e^{-m^2t}P_t(x,y)\dd t=:G^m(x,y),\\
 \bE^m_N[\phi(x)\phi(y)]&=\int_{0}^{\infty}e^{-m^2t}P^*_t(x,y)\dd t=:G^{m,*}(x,y),
\end{split}
 \end{equation}
 where $P^*_t$ is the heat kernel on $\gL_N$ with Dirichlet boundary condition on $\partial \gL_N$,
 \begin{equation}
  P^*_t(x,y):= P_x\left[ X_t=y \ ; \ \tau_{\partial \gL_N}<t\right].
 \end{equation}
We simply write $G^*$ in the case $m=0$.

\medskip

Note that, because of the spatial Markov property (Section \ref{secmass}) and of \eqref{transkix},
when $\hat \phi$ has law $\hat \bP^m$ and $\phi$ has law $\bP_N$, 
$(H^{m,\hat\phi}_N(x)+ \phi_x)_{x\in \gL_N}$ has the same law as the (marginal in $\gL_N$ of the) infinite volume field.
Hence as a consequence
\begin{equation}\label{covH}
\hat\bE^m[H^{\hat\phi}_N(x)H^{\hat\phi}_N(y)]=G^m(x,y)-G^{m,*}(x,y)=\int_{0}^{\infty}e^{-m^2t}(P_t(x,y)-P^*_t(x,y))\dd t.
\end{equation}

Before giving more involved estimates, let us mention first  a quantitative version of the Local Central Limit Theorem \cite[Theorem 2.1.1]{cf:LL} 
for the heat kernel which we use as an essential building brick to obtain them. There exists a constant $C$ such that for all $t\ge 1$,
\begin{equation}\label{lclt}
 \left|P_t(x,x)-\frac{1}{4\pi t}\right|\le \frac{C}{t^{3/2}},
\end{equation}
Let us recall the notation \eqref{distA} for the distance between a set and a point.
The following two lemmas are proved in Appendix \ref{appendix}.
 
\begin{lemma}\label{Greenesteem}
There exists a constant such that $C$  
\begin{itemize}
 \item [(i)] For all $m\le 1$, for any $x\in \bbZ^2$
\begin{equation}\label{eq:variance}
\left| G^m(x,x)+\frac{1}{2\pi} \log m\right|\le C
\end{equation}
\item[(ii)] For all $m\le 1$, for any $x\in \gL_N$
\begin{equation}\label{eq:stimagreen}
\left| G^{m,*}_N(x,x)-\frac{1}{2\pi} \log \min(m^{-1},d(x,\partial \gL_N))\right|\le C.
\end{equation}
\end{itemize}
\end{lemma}

 \begin{lemma}\label{lem:kerestimate}
The following assertions hold
\begin{itemize}
 \item  [(i)] There exists a constant $C$ such that for all  $t\ge 1$, $|x-y| \le \sqrt{t}$, we have 
 \begin{equation}\label{gradientas}\begin{split}
 \left( P_t(x,x)- P_t(x,y)\right)&\le \frac{ C |x-y|^2 }{t^2},\\
 \left( P^*_t(x,x)+P^*_t(y,y)-2P^*_t(x,y)\right)&\le \frac{ C |x-y|^2 }{t^2}.
 \end{split}\end{equation}
 \item [(ii)] There exist a constant $C$ such that for all $t\ge 1$ and $x,y$ satifying $ |x-y| \le t$ 
 we have 
 \begin{equation}\label{croco}
  P_t(x,y)\le \frac{C}{t}e^{-\frac{|x-y|^2}{Ct}}.
 \end{equation}
 and as a consequence 
\begin{equation}\label{greensum}
\sum_{y\in \bbZ^2} G^{m}(x,y)\le Cm^{-2}.
\end{equation}
\item[(iii)] We have for all $x$
\begin{equation}\label{ltest}
 \frac{P^*_t(x,x)}{P_t(x,x)}\le C  \frac{\left[d(x,\partial \gL_N)\right]^2}{t}.
\end{equation}

\item[(iv)] We have for all $x$ 
\begin{equation}
\begin{cases}\label{kilcompare}
 P_t(x,x)-P^*_t(x,x)\le \frac{C}{t}e^{-\frac{d(x,\partial \gL_N)^2}{C t}}, \quad  &\text{ for } t\ge d(x,\partial \gL_N),\\
  P_t(x,x)-P^*_t(x,x)\le \frac{C}{t}e^{-\frac{1}{C}d(x,\partial \gL_N) \log \left(\frac{ d(x,\partial \gL_N)}{t} \right) }, \quad  &\text{ for } t\le d(x,\partial \gL_N).
\end{cases}
 \end{equation}
\end{itemize}
 \end{lemma}
 
  \subsection{Cost of positivity constraints for Gaussian random walks}
  
  Finally we conclude this preliminary section with an estimate for the probability to remain above a line for Gaussian random walks.
  The statement is not optimal and the term $(\log k)$ could be replaced by $1$ but as the rougher estimate is sufficient for our purpose we prefer to keep the proof simpler.
We include the proof in the Appendix \ref{appendix} for the sake of completeness.

  \begin{lemma}\label{lem:bridge}
Let $(X_i)_{i=1}^k$ be arandom walk with independent centered Gaussian increments, each of which with variance bounded above by $2$
and such that the total variance satisfies $\Var(X_k)\ge k/2$.
Then we have for all $x\ge 0$
\begin{equation}
1-e^{-\frac{x^2}{k}}\le \bP\left[ \max_{i}X_i\le x \ | \ X_k=0 \right]\le \frac{ C(x+(\log k))^2}{k}.
\end{equation}
 \end{lemma}

 \section{The upper-bound on the free energy}\label{seclower}
 
 Let us briefly discuss the structure of the proof before going into more details.
 The main idea is presented in Section \ref{changeofme}: we introduce a function
 which penalizes some environments $\go$  which are too favorable, and use it to get a bet annealed bound which penalizes 
 the trajectories with clustered contact points in a small region (Proposition \ref{nonrandom}).  
 \medskip

However, to perform the coarse-graining step of the proof, we need some kind of control on $\phi$.
For this reason, in Section \ref{restrictou} we start the proof by showing that restricting the partition function to a set of uniformly bounded trajectory does not 
affect a lot the free energy. 
 
 \subsection{Restricting the partition function}\label{restrictou}
 
 In this section, we show that restricting the partition function by limiting the maximal height of the field $\phi$ does not affect too much the free energy.
 This statement is to be used to control the boundary condition of each cell when performing a coarse-graining argument in Proposition \ref{scorpiorizing}.
 Let us set  
\begin{equation}
   \cA^h_N  := \left\{ \forall x \in  \gL_N, \ |\phi_x|\le |\log h|^2 \right\},
\end{equation}
and  write 
\begin{equation}
 Z^{\gb,\go}_{N,h}(\cA^h_N):=\bE_N \left[ \exp \left(\sum_{x\in \tilde \gL_N}(h+\gb \go_x-\gl(\gb))\delta_x\right)\ind_{\cA^h_N} \right].
\end{equation}

\begin{proposition}\label{boundtheprob}
There exists a constant $c$ such for any  $h\in(0,h_0)$ and $\gb>0$ we have  
\begin{equation}
 \liminf_{N\to \infty} \frac{1}{N^2} \bbE \log \bP^{\gb,\go}_{N,h}\left[ \cA^h_N \right]\ge  -\exp\left(-c |\log h|^2\right).
\end{equation}

As a consequence, we have 
\begin{equation}
 \tf(\gb,h)\le \liminf_{N\to \infty} \frac{1}{N^2} \bbE \log Z^{\gb,\go}_{N,h}(\cA^h_N)+\exp\left(-c |\log h|^2\right).
\end{equation}
\end{proposition}

\begin{proof}
For practical purposes we introduce the two following events 
\begin{equation}\begin{split}\label{evvents}
 \cA^{h,1}_N&:= \left\{\forall x \in \mathring \gL_N,\  \phi_x\le |\log h|^2\right\},\\
 \cA^{h,2}_N&:= \left\{\forall x \in \mathring \gL_N,\  \phi_x\ge -|\log h|^2\right\},\\
 \cB_N&:= \left\{\forall x \in \mathring \gL_N,\  \phi_x\ge 1\right\},\\
 \cC_N&:= \left\{\forall x \in \mathring \gL_N,\  \phi_x\le -1 \right\}.
      \end{split}
\end{equation}
We have $ \cA^{h}_N= \cA^{h,1}_N\cap  \cA^{h,2}_N$.
In order to obtain a bound on the probability of  $\cA^{h}_N$ we need to use the FKG inequality for the Gaussian free field which we present briefly 
(we refer to \cite[Section B.1]{cf:Notes} for more details).
We denote by $\le$ the natural order on the set of functions $\{\phi, \  \gL_N \to \bbZ^d\}$ defined by
\begin{equation}
\left\{  \ \phi\le \phi' \ \right\} \quad \Leftrightarrow \quad \left\{\  \forall x \in \gL_N, \ \phi_x\le \phi'_x \ \right\}.
\end{equation}
An event $A$ is said to be increasing if for $\phi\in A$ we have $$\phi'\ge \phi \Rightarrow \phi'\in A$$
and decreasing if its complement is increasing.
Let us remark that all the events described in \eqref{evvents} are either decreasing or increasing.
A probability measure $\mu$ is said to satisfy the FKG inequality if 
for any pair of increasing events $A, B$ we have $\mu(A \cap B)\le \mu(A)\mu( B)$. 
Note that this yields automatically similar inequalities for any pairs of monotonic events which we also call FKG inequalities.

\medskip

It is well know that $\bP_N$ satisfies the FKG inequality: it is sufficient to check that Holley's criterion \cite{cf:FKG, cf:Holley} is satisfied  by
the Hamiltonian in \eqref{density}. The same argument yields that  $\bP^{\gb,\go}_{N,h}$ as well as the conditionned measures 
$\bP^{\gb,\go}_{N,h}\left( \cdot \ | \ \cA^{h,1}_N\right)$ and $\bP^{\gb,\go}_{N,h}\left( \cdot \ | \ \cB^{h,1}_N\right)$ also satisfy the FKG inequality.
Hence using the FKG inequality for  $\bP^{\gb,\go}_{N,h}$, we have  
\begin{equation}
 \bP^{\gb,\go}_{N,h}(\cA^{h,1}_N) \ge  \bP^{\gb,\go}_{N,h}(\cA^{h,1}_N \  | \ \cB_N ) =\bP_{N}(\cA^{h,1}_N \  | \ \cB_N ).
\end{equation}
Then, using the FKG inequality for  $\bP^{\gb,\go}_{N,h}( \ \cdot \  |  \ \cA^{h,1}_N)$  and we have 
\begin{equation}
  \bP^{\gb,\go}_{N,h}\left( \cA^{h,2}_N \ | \  \cA^{h,1}_N\right)\ge \bP^{\gb,\go}_{N,h}\left(\cA^{h,2}_N \  | \ \cC_N \right)
  \ge \bP_{N}\left(\cA^{h,2}_N \  | \ \cC_N \right)=\bP_{N}(\cA^{h,1}_N \  | \ \cB_N ),
\end{equation}
where we used symmetry to get the last equality.
Then we can conclude that 
\begin{equation}
  \bP^{\gb,\go}_{N,h}(\cA^{h}_N)= \bP^{\gb,\go}_{N,h}(\cA^{h,1}_N \cap \cA^{h,2}_N) 
  \ge \left[\bP_{N}(\cA^{h,1}_N \  | \ \cB_N )\right]^2
  \ge  \left[\bP_{N}(\cA^{h,1}_N \cap \cB_N )\right]^2.
\end{equation}
We are left with estimating the last term.
Note that changing the boundary condition by a constant amount does not affect the leading order of the asymptotic thus to conclude it is sufficient to bound 
asymptotically the probability of the event
\begin{equation}
\cA^{h,3}_N:= \left\{  \max_{x \in \gL_N} |\phi_x| \le \frac{|\log h|^2-1}{2} \right\},
\end{equation}
which is a translated version of $\cA^{h,1}_N \cap \cB_N$. 
More precisely we have for an adequate constant $K_h$
\begin{equation}
   \bP^{\gb,\go}_{N,h}(\cA^{h}_N)\ge \exp(-K_h N) \left[\bP_{N}(\cA^{h,3}_N) \right]^2
\end{equation}
To bound the probability of $\cA^{h,3}_N$ we use the following result, whose proof is postponed to the end of the Section.

\begin{lemma}\label{controlgrid}
 There exists a constant $C$ such that for any $N$, and for any set $\gG\subset \mathring\gL_N$ which is such that $\gG \cup \partial \gL_N$ is connected, we have 
\begin{equation}
   \bP_{N}\left[ \max_{x \in \gG} |\phi_x|\le 1 \right] \ge \exp(-C |\gG|).
\end{equation}
\end{lemma}
We divide $\gL_N$ in cells of side-length $$N_0:=\exp(c|\log h|^2)$$ 
for some small constant $c$.
We set $$\gL(y,N_0):= yN_0+ \gL_{N_0}.$$
We apply Lemma \ref{controlgrid} for the following set 
\begin{equation}
\gG_N=\left(\bigcup_{y\in \bbZ^2} \partial\gL(y,N_0)\right)\cap \mathring\gL_N,
\end{equation}
which is is a grid which splits $\gL_N$ in cells of side-length $N_0$.
We obtain that 
\begin{equation}
 \frac{1}{N^2} \log\bP_{N}\left[ \max_{x \in \gG_N} |\phi_x|\le 1 \right] \ge \frac{2C}{N_0}= 2C \exp(-c|\log h|^2),
\end{equation}
where we used the inequality $|\gG_N|\le 2 N^2/N_0$ valid for all $N$. 
To conclude we need to show that
\begin{equation}\label{fdfe}
 \frac{1}{N^2} \log \bP_{N}\left[  \max_{x \in \gL_N} |\phi_x| \le \frac{|\log h|^2-1}{2} \ \Big| \  \max_{x \in \gG_N} |\phi_x|\le 1 \right] \ge - (N_0)^{-2}.
\end{equation}
To prove \eqref{fdfe} it is sufficient to remark that conditioned to $(\phi_x)_{x\in \gG_N}$,
the variance of the field $(\phi_x)_{x\in\gL(y,N_0)}$ 
is uniformly bounded by  $\frac{1}{2\pi} \log N_0+C$ (cf. \eqref{eq:stimagreen} for $m=0$).
Thus, for any realization of $\phi$ satisfying 
$\max_{x \in \gG_N} |\phi_x|\le 1$, for any $z\in \gL_N \setminus \gG_N$, using the Gaussian tail bound \eqref{gtail}
we have for $h$ sufficiently small 
\begin{equation}
\bP_N \left[  |\phi_z| \ge  \frac{|\log h|^2-1}{2}  \ \Big| \  (\phi_x)_{x\in \gG_N} \right] \le \exp\left( -\frac{ \pi |\log h|^4}{ 4 \log N_0} \right)\le
\exp\left(-\frac{\pi}{4c} |\log h|^2\right).
\end{equation}
Now with this in mind we can apply union bound in $\gL(y,N_0)$ and obtain
\begin{multline}\label{ddsad}
  \bP_{N}\left[  \max_{z \in \gL(y,N_0)} |\phi_z| \le \frac{|\log h|^2-1}{2} \ \Big| \  (\phi_x)_{x\in \gG_N} \right]\\
  \ge  1- (N_0-1)^2\exp\left(-\frac{\pi}{4c} |\log h|^2\right)
  \ge e^{-1/2}.
\end{multline}
where the last inequality is valid provided the constant $c$ is chosen sufficiently small.
As, conditioned to  the realization of $(\phi_x)_{x\in \gG_N}$, the fields  $(\phi_x)_{x\in\gL(y,N_0)}$ are independent for different values of $y$, 
we prove that the inequality \eqref{fdfe} holds by multiplying \eqref{ddsad}
for all distinct $\gL(y,N_0)$ which fit (at least partially) in $\gL_N$ (there are at most $(N/N_0)^2$ full boxes, to which one must add at most $2N/N_0+1$ 
uncompleted  boxes), and taking the expectation with 
respect to $(\phi_x)_{x\in \gG_N}$ conditioned on the event  $\max_{x \in \gG_N} |\phi_x|\le 1$. This ends the proof of
Proposition \ref{boundtheprob}.
\end{proof}

\begin{proof}[Proof of Lemma \ref{controlgrid}]
 We can prove it by induction on the cardinality of $\gG$.
 Assume that the result is valid for $\gG$ and let us prove it for $\gG\cup\{z\}$.
 \begin{equation}
    \bP_{N}\left[  \phi_z\in [-1,1] \ | \ \max_{x\in \gG} |\phi_x|\le 1 \right]\ge \exp(-C).
 \end{equation}
 Note that conditioned to $(\phi_x)_{x\in \gG}$, $\phi_z$ is a Gaussian variable.
 Its variance is given by
 \begin{equation}
  E_x\left[\int^{\tau_{\partial \gL_N \cup \gG}}_0 \ind_{\{X_t=z\}} \dd t \right]\le 1.
 \end{equation}
 The reason being that as by assumption $\partial \gL_N \cup \gG\cup \{z\}$, the walk $X$ is killed with rate one while it lies on $z$.
In addition, if $\max_{x\in \gG} |\phi_x|\le 1$, then necessarily 
 
 \begin{equation}
  \bE_N\big[ \phi_z  \ |  \ (\phi_x)_{x\in \gG} \big]\in[-1,1].
 \end{equation}
 For this reason, the above inequality is valid if one chooses
 \begin{equation}
  C:=-\max_{u\in[-1,1]} \log P\left (\cN \in [-1+u,1+u] \right)= -\log P\left(\cN \in [0,2] \right),
 \end{equation}
where $\cN$ is a standard normal.
\end{proof}

\subsection{Change of measure}\label{changeofme}

To bound the expectation of $\bbE [ \log Z^{\gb,\go}_{N,h} (\cA^h_N)]$ we use a ``change of measure'' argument .
The underlying idea is that the annealed bound obtained by Jensen's inequality \eqref{anealed} is not sharp because some very atypical  $\go$'s 
(a set of $\go$ of small probability) give the most important contribution to the annealed partition function. 
Hence our idea is to identify these bad environments and to introduce a function $f(\go)$ that penalizes them.
This idea originates from \cite{cf:GLT} where it was used to prove the non-coincidence of critical point for a hierarchical variant of the pinning model
and was then improved many times in the context of pinning \cite{ cf:BL, cf:DGLT, cf:GLT2} and found application for other models like random-walk pinning, directed polymers,  random walk in a random environment or
self-avoiding walk in a random environment \cite{cf:BT, cf:BS1,  cf:BS2,  cf:L2, cf:L3, cf:YZ}.
\medskip

In \cite{cf:BL, cf:DGLT, cf:GLT2}, we used the detailed knowledge that we have on the structure of the set of contact points, 
(which is simply a renewal process) in order to find the right penalization function $f(\go)$. 

\medskip

Here we have a much less precise knowledge on the structure $(\delta_x)_{x\in \gL_N}$ under $\bP_N$ (especially because we have to consider possibly very wild boundart condition),
but we know that one typical feature of the two-dimensional free field is that the
level sets tend to have a clustered structure. We want to perform a change of measure that has the consequence of penalizing these clusters of contact points: we do so by looking at 
the empirical mean of $\go$ in some small regions and by giving a penalty when it takes an atypically high value.

\medskip

Let us be more precise about what we mean by penalizing with a function $f(\go)$.
Using Jensen inequality, we remark that 
\begin{equation}
  \bbE\left[ \log  Z^{\gb,\go}_{N,h}(   \cA^h_N )\right]=    2\bbE\left[ \log \sqrt{ Z^{\gb,\go}_{N,h}(   \cA^h_N )}\right]\le  2 \log \bbE \left[ \sqrt{Z^{\gb,\go}_{N,h}(   \cA^h_N )} \right] 
\end{equation}
If we let $f(\go)$ be an arbitrary positive function of $(\go_x)_{x\in \tilde \gL_N}$, we have by Cauchy-Schwartz inequality
\begin{equation}
 \bbE \left[ \sqrt{Z^{\gb,\go}_{N,h}(   \cA^h_N )} \right]^2 \le \bbE[f(\go)^{-1}]  \bbE \left[ f(\go)Z^{\gb,\go}_{N,h}(   \cA^h_N )\right] ,
\end{equation}
and hence
\begin{equation}\label{Cauchyschwartz}
  \frac{1}{N^2}\bbE\left[ \log Z^{\gb,\go}_{N,h} (   \cA^h_N )\right]\le  \frac{1}{N^2}\log \bbE[f(\go)^{-1}]+ 
  \frac{1}{N^2}\log \bbE \left[ f(\go)Z^{\gb,\go}_{N,h}(   \cA^h_N )\right] .
\end{equation}
Let us now present our choice of $f(\go)$. 
Our idea is to perform some kind of coarse-graining argument: we divide $\gL_N$ into cells of fixed side-length $N_1$ 
\begin{equation}\label{cellsize}
N_1(h):=h^{-1/4},
\end{equation}
and perform a change of measure inside of each cell.
We assume that $N_1$ is an even integer (the free energy being monotone this causes no loss of generality), and that
$N=k N_1$ is a sufficiently large multiple of $N_1$.
Given $y\in \bbZ^2$, we let $\tilde \gL_{N_1}(y)$ denote the translation of the box $\tilde \gL_{N_1}$ which is (approximately) centered at $yN_1$ 
(see Figure \ref{fig:structure})
$$\tilde \gL_{N_1}(y):= N_1\left[y-\left(\frac{1}{2},\frac{1}{2}\right)\right]+\tilde \gL_{N_1}.$$
In the case $y=(1,1)$ we simply write $\tilde \gL'_{N_1}$ (note that it is not identical to  $\tilde\gL_{N_1}$).
We define the event
\begin{equation}
  \cE_{N_1}(y):=
  \left\{ \exists x\in \tilde \gL_{N_1}(y), \!\!\!\!\!\!  \sum_{\{ z\in \tilde \gL_{N_1}(y) \ : \ |z-x|\le (\log N_1)^2 \}} \!\!\!\!\!\! \!\!\!\!\!\! \go_z \quad \ge \frac{\gl'(\gb) (\log N_1)^3}{2} \right\}.
\end{equation}
which is simply denoted by $\cE_{N_1}$ in the case when $y=(1,1)$. Here $\gl'(\gb)$ denotes the derivative of $\gl$ defined in \eqref{eq:assume-gl}.
Finally we set
\begin{equation}
f(\go):= \exp\left(- 2\sum_{y\in \lint 1, k-1\rint} \ind_{\cE_{N_1}(y)}\right).
\end{equation} 
The effect of $f(\go)$ is to give a penalty (multiplication by $e^{-2}$) for each cell in which one can find a region of $\go$ with diameter $(\log N_1)^2$ and atypically high empirical mean.

\medskip

\noindent Combining Proposition \ref{boundtheprob} and \eqref{Cauchyschwartz}, we have (provided that the limit exists)
\begin{equation}\label{uppbn}
 \tf(\gb,h)\le e^{-c(\log h)^{2}}+ \lim_{k\to \infty} \frac{1}{N^2}\log \bbE[(f(\go))^{-1}]+ 
 \liminf_{k\to \infty} \frac{1}{N}\log \bbE \left[ f(\go)Z^{\gb,\go}_{N,h}(\cA^h_N)\right].
\end{equation}

We can conclude the proof with
the two following results, which evaluate respectively the cost and the benefit of our change of measure procedure.

\begin{proposition}\label{cost}
There exists positive constants $c(\gb)$ and $h_0(\gb)$ and  such that for all $h\in (0,h_0(\gb))$ sufficiently small, for all $k$
 \begin{equation}
\log \bbE\left[ (f(\go))^{-1}\right]\le (k-1)^2 e^{-c(\gb)(\log h)^2}.
\end{equation}
As a consequence we have
\begin{equation}
 \frac{1}{N^2}\log \bbE[(f(\go))^{-1}]\le  e^{-c(\gb)(\log h)^2}.
\end{equation}
\end{proposition}

\begin{proposition}\label{benefit}
 There exists $h_0(\gb)>0$ such that for all $h\in (0,h_0(\gb))$
 \begin{equation}
\limsup_{k\to \infty} \frac{1}{N^2}\log \bbE \left[ f(\go)Z^{\gb,\go}_{N,h}(\cA^h_N)\right]\le  e^{-2|\log h|^{3/2}}.
\end{equation}
\end{proposition}

As a consequence of \eqref{uppbn} and of the two propositions above, we obtain that for $h\in(0, h_0(\gb)),$ we have 
\begin{equation}
 \tf(\gb,h)\le e^{-|\log h|^{3/2}}.
\end{equation}
The proof of Proposition \ref{cost} is simple and short and is presented below. The proof Proposition \ref{benefit} requires a significant amount of work.
We decompose it in important steps in the next subsection.

\begin{proof}[Proof of Proposition \ref{cost}]
 Because of the product structure, we have 
\begin{equation}\label{cgt}
 \bbE\left[ (f(\go))^{-1}\right]= \big( \bbE \left[ \exp\left( 2\cE_{N_1}\right) \right]\big)^{(k-1)^2}.
\end{equation}
Hence it is sufficient to obtain a bound on
\begin{equation}\label{fds}
\log \bbE \left[ \exp\left( 2\ind_{\cE_{N_1}}\right) \right]\le  (e^2-1) \bbP[\cE_{N_1}(0)].
\end{equation}
As an easy consequence of the proof of Cram\'ers Theorem (see e.g.\ \cite[Chapter 2]{cf:DZ}),
there exists a constant $c(\gb)$ that any $x\in \tilde \gL_{N_1}$
\begin{equation}
\bbP \left[ \sum_{\{ z\in \tilde \gL'_{N_1} \ : \ |z-x|\le (\log N_1)^2 \}}  \!\!\!\!\!\! \!\!\!\!\!\!  \go_z \ge \,  \frac{\gl'(\gb) (\log N_1)^3}{2} \  \right]
\le e^{-c(\gb)(\log N_1)^2},
\end{equation}
and by union bound we obtain that  $\bbP[\cE_{N_1}]\le N^2_0\exp(-c(\log N_1)^2)$,
which in view of \eqref{fds} and \eqref{cgt} is sufficient to conclude
\end{proof}

\begin{figure}[h]
\begin{center}
\leavevmode
\epsfxsize =10.5 cm
\psfragscanon
\psfrag{0}[c][l]{\small $0$}
\psfrag{N1/2}[c][l]{ \small $N_1/2$}
\psfrag{2N1/3}[c][l]{ \small $3N_1/2$}
\psfrag{N}[c][l]{ \small $N=kN_1$}
\psfrag{NN}[c][l]{ \small $N$}
\psfrag{GL14}[c][l]{\small $\tilde \gL_{N_1}(1,4)$}
\psfrag{GGL14}[c][l]{\small $\gL_{2N_1}(1,4)$}
\psfrag{GL42}[c][l]{\small $\tilde\gL_{2N_1}(4,2)$}
\epsfbox{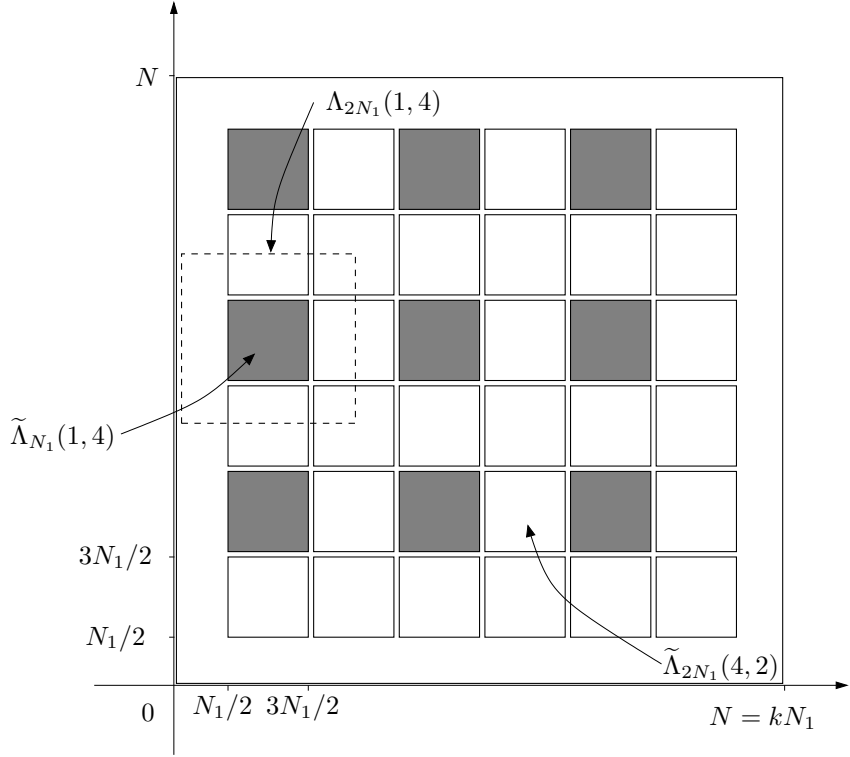}
\end{center}
\label{fig:structure} 
\caption{A schematic representation of our coarse graining procedure. We have chosen $k=7$.
The small squares of side length $N_1$ represent the disjoint boxes $\tilde \gL_{N_1}(y)$, $y\in \lint 1,k-1\rint^2$. 
The dark squares represent the boxes $\tilde \gL_{N_1}(y)$ for which $y\in\Xi(2)$. The dotted square corresponds to the box $\gL_{2N_1}(1,4)$.}
\end{figure}

\subsection{Decomposing the proof of Proposition \ref{benefit}}

The proof is split in three steps, whose details are performed in Section \ref{proufin}, \ref{proufeux} and \ref{proufoix} respectively.
In the first one we show that our averaged partition function 
$\bbE\left[ f(\go) Z^{\gb,\go}_{N,h}(\cA^h_N)\right]$, can be bounded from above by the partition of an homogenous system where 
an extra term is added in the Hamiltonian to penalize the presence of clustered contact in a small region (here a region of diameter $(\log N_1)^2$).
We introduce the event $\cC_{N_1}(y)$ which indicates the presence of such a cluster in $\tilde\gL_{N_1}(y)$,
\begin{equation}
 \cC_{N_1}(y):= \left \{ \exists x\in \tilde \gL_{N_1}(y), \!\!\!\!\!\! \!\! \sum_{\{ z\in \tilde \gL_{N_1}(y) \ : \ |z-x|\le (\log N_1)^2 \}} 
 \!\!\!\!\!\! \!\!\!\!\!\!  \gd_z  \ge \  (\log N_1)^3 \right\}.
\end{equation}
We simply write  $\cC_{N_1}$ for the case $y=(1,1)$.

\begin{proposition}\label{nonrandom}
 We have 
\begin{multline}\label{gteraf}
  \bbE\left[ f(\go) Z^{\gb,\go}_{N,h}(\cA^h_N)\right]\\
  \le 
  \bE_N\left[ \exp\left(  h\sum_{x \in \tilde \gL_{N}} \delta_x - \sum_{y \in \lint 1, k-1 \rint^2} \ind_{\cC_{N_1}(y)} \right)\ind_{\cA^h_N} \right]=:
  \hat Z(N,N_1,h).
\end{multline}
\end{proposition}

In the second step, we perform a factorization in order to reduce the estimate of $\hat Z(N,N_1,h)$ to that of 
similar system with only one cell. 
Let us set (see Figure \ref{fig:structure})
\begin{equation}
 \gL_{2N_1}(y):= N_1\big( y-(1,1)\big)+\gL_{2N_1}.
\end{equation}
Note  that for every for  $y\in  \lint 1,k-1\rint^2$ we have $\gL_{2N_1}(y)\subset \gL_N$  and that  $\gL_{2N_1}((1,1))=\gL_{2N_1}$.

\begin{proposition}\label{scorpiorizing}
 We have 
 
\begin{equation}
  \hat Z(N,N_1,h)\le e^{2 N_1 N h}\left(\max_{\{ \hat \phi \ : \ |\hat \phi|_{\infty}\le |\log h|^2\} }
 \bE^{\hat \phi}_{2N_1}\left[e^{4h\sum\limits_{x \in \tilde \gL'_{N_1}}
  \delta_x -4\ind_{\cC_{N_1}}} \right] \right)^{\frac{(k-1)^2}{4}}
\end{equation}

\end{proposition}

Let us notice two important features in our factorization which are present to reduce possible nasty boundary effects:
\begin{itemize}
 \item There is a restriction on the boundary condition  $|\hat \phi|_{\infty}\le |\log h|^2$, which forbids wild behavior of the field.
 This restriction is directly inherited from the restriction to $\cA^h_N$ in the partition function  
 and brings some light on the role of Proposition \ref{boundtheprob} in our proof.
 \item The Hamiltonian  
 $$4h\sum\limits_{x \in \tilde \gL'_{N_1}}
   \delta_x- 4\ind_{\cC_{N_1}},$$ is a functional of $(\phi_x)_{x\in \tilde\gL'_{N_1}}$ i.e.\ of the field restricted to a region which 
   is distant from the boundary of the box $\partial \gL_{2N_1}$. 
\end{itemize}

The final step of the proof consists in evaluating the contribution of one single cell to the partition function.

\begin{proposition}\label{onecell}

There exists a constant $c$ such that for all $h$ sufficiently small 
for all $\hat \phi$ satisfying $|\hat \phi|_{\infty}\le |\log h|^2$ we have 
\begin{equation}
\log \bE^{\hat \phi}_{2N_1}\left[e^{ 4h\sum\limits_{x \in \tilde \gL_{N_1}}
   \delta_x- 4\ind_{\cC_{N_1}} } \right]\le e^{-2(\log h)^{3/2}}.
\end{equation}

\end{proposition}
\noindent Combining the three results presented above, we have 
\begin{equation}
   \bbE\left[ f(\go) Z^{\gb,\go}_{N,h}(\cA^h_N)\right]\le 2 N_1 N h+\frac{(k-1)^2}{4}e^{-2(\log h)^{3/2}},
\end{equation}
and this is sufficient to conclude the proof of Proposition \ref{benefit}.

\subsection{Proof of Proposition \ref{nonrandom}}\label{proufin}
Given a realization $\phi$, we let $\bbP^{\phi}$ be a probability law which 
is absolutely continuous with respect to $\bbP$ and whose the density is given by
\begin{equation}
 \frac{\dd  \bbP^{ \phi}}{\dd \bbP}(\go):= \exp\left(\sum_{x\in \tilde \gL_N} \left(\gb\go_x-\gl(\gb)\right)\delta_x \right).
\end{equation}
Under $\bbP^{ \phi}$, the variables $(\go_x)_{x\in \bbZ^d}$ are still independent but they are not IID, as the law of the $\go_x$s for which 
$\delta_x=1$ have been tilted.
In particular it satisfies  
\begin{equation}
 \bbE^{\phi}[\go_x]= \gl'(\gb)\delta_x \quad \text{ and }  \Var_{{\bbP}^{\phi}}[\go_x]=1+ (\gl''(\gb)-1)\delta_x 
\end{equation}
where $\gl'(\gb)$ and $\gl''(\gb)$ denote the two first derivatives of $\gl$ the function defined in \eqref{eq:assume-gl}.
This notation gives us another way of writing the quantity that we must estimate
\begin{equation}\label{rewrite}
 \bbE\left[ f(\go)Z^{\gb,\go}_{N,h}(\cA^h_N)\right]= \bE_N\left[ \bbE^{\phi}[f(\go)] e^{h \sum_{x \in \tilde \gL_N }\delta_x}\ind_{\cA^h_N}\right]. 
\end{equation} 
To conclude it is sufficient to prove that 
\begin{equation}
\bbE^{\phi}[f(\go)] \le  \exp \left( -\sum_{y \in \lint 0, k-1 \rint} \ind_{\cC_{N_1}(y)} \right).
\end{equation}
Note that because both $\bbE^{\phi}$ and $f(\go)$  have a product structure, it is in fact sufficient to prove that for any $y\in \lint 0, k-1 \rint^2$
we have 
\begin{equation}\label{unecase}
 \bbE^{\phi}\left[e^{-2 \ind _{\cE_{N_1}(y)}}\right]\le e^{-\ind_{\cC_{N_1}(y)}}.
\end{equation}
With no loss of generality we assume that $y=(1,1)$.
The result is obvious when  $\phi \notin \cC_{N_1}$ hence we can also assume $\phi \in \cC_{N_1}$.
Let $x_0\in \tilde \gL'_{N_1}$ be a vertex satisfying
$$ \sum_{\{z\in \gL'_{N_1} \ : \ |z-x_0|\le (\log N_1)^2\}} \delta_x \ge (\log N_1)^3,$$
(e.g.\ the smallest one for the lexicographical order).
We have 
\begin{equation}\begin{split}
\bbE^{\phi}\left[ \sum_{\{ z\in  \tilde \gL'_{N_1} \  : \ |z-x_0|\le (\log N_1)^2\}}  \!\!\!\!\!\! \!\! \go_z \ \right]&= \gl'(\gb) \!\!\!\!\!\! \!\!\!\!\!\!
\sum_{\{ z\in  \tilde \gL'_{N_1} \  : \ |z-x_0|\le (\log N_1)^2\}}  \!\!\!\!\!\! \!\!\!\!\!\! \gd_z  \, \ge  \, \gl'(\gb)(\log N_1)^3,\\
\Var_{\bbP^{\phi}}\left[ \sum_{\{ z\in  \tilde \gL'_{N_1} \  : \ |z-x_0|\le (\log N_1)^2\}}  \!\!\!\!\!\! \!\! \go_z \ \right]
&\le \left[2(\log N_1)^2+1\right]^2\max(\gl''(\gb),1).
\end{split}\end{equation}
Hence in particular if $N_1$ is sufficiently large, Chebychev's inequality gives 
\begin{equation}
 \bbP^{\phi} \left[\cE_{N_1}\right]\le e^{-1}-e^{-2},
\end{equation}
which implies \eqref{unecase}.

\qed

\subsection{Proof of Proposition \ref{scorpiorizing}}\label{proufeux}

We start by taking care of the contribution of the contact points located near the boundary $\partial \gL_N$, as they are not included in any  $\tilde \gL_{N_1}(y)$.
Assuming that all these points are contact points we obtain the following crude bound
\begin{equation}
\sum_{x \in \tilde \gL_N } \delta_x\le  \left[N^2-(k-1)^2N_1^2\right] +\sum_{y \in \lint 1, k-1 \rint^2} \sum_{x \in \tilde \gL_{N_1}(y)} \delta_x.
\end{equation}
and the first term is smaller than  $2NN_1$.
Hence  we have

\begin{equation}\label{gteraf2}
\hat Z(N,N_1,h)
\le   e^{2N_1 N h}  \bE_N\left[  e^{\sum_{y \in \lint 1,k-1\rint^2} \left( h\sum_{x \in \tilde \gL_{N_1}(y)} \delta_x -  
\ind_{\cC_{N_1}(y)}\right)} \ind_{\cA^h_N}\right].
\end{equation}
We partition the set of indices $\lint 1, k-1 \rint^2$ into $4$ subsets, according to the parity of the of the coordinates.
If we let $\alpha_1(i)$ and $\alpha_2(i)$ denote the first and second diadic digits of $i-1$. We set 
\begin{equation}
 \Xi(i):= \big\{ y=(y_1,y_2) \in \lint 1,k-1\rint^2 \ : \  \forall j \in \{1,2\}, \, y_j \stackrel{(\text{mod } 2)}{=}\alpha_j(i)\big\}.
\end{equation}
Using H\"older's inequality we have 
\begin{multline}
 \bE_N\left[  e^{\sum_{y \in \lint 1,k-1\rint^2} \left( h\sum_{x \in \tilde \gL_{N_1}(y)} \delta_x -  
\ind_{\cC_{N_1}(y)}\right)} \ind_{\cA^h_N}\right]^{4}
\\
 \le \prod_{i=1}^4 
  \bE_N\left[  e^{4\sum_{y \in \Xi(i)} \left( h\sum_{x \in \tilde \gL_{N_1}(y)} \delta_x -  
\ind_{\cC_{N_1}(y)}\right)} \ind_{\cA^h_N}\right].
\end{multline}
For a fixed $i\in \lint 1, 4 \rint$,
the interiors of the boxes  $\gL_{2N_1}(y)$, $y\in \Xi(i)$ are disjoint (neighboring boxes overlap only on their boundary, we refer to Figure \ref{fig:structure}).
This gives us a way to factorize the exponential: let us condition the expectation to the realization of $(\phi_x)_{x\in \gG(i)}$
where 
\begin{equation}
 \gG(i):= \bigcup_{y \in \Xi(i)} \partial \gL_{2N_1}(y).
\end{equation}
The spatial Markov property implies that conditionally on  $(\phi_x)_{x\in \gG(i)}$, the restrictions 
$\left[(\phi_x)_{x\in \gL_{2N_1}(y)}\right]_{y \in \Xi(i)}$ are independent. 
Hence we can factorize the expectation and get 
\begin{multline}\label{yipendence}
    \bE_N\left[  e^{4\sum_{y \in \Xi(i)} \left( h\sum_{x \in \tilde \gL_{N_1}(y)} \delta_x -  
\ind_{\cC_{N_1}(y)}\right)} \ | \ (\phi_x)_{x\in \gG(i)} \right]\\
    \le 
    \prod_{y\in \Xi(i)}     \bE_N\left[  e^{4 \left( h\sum_{x \in \tilde \gL_{N_1}(y)} \delta_x -  
\ind_{\cC_{N_1}(y)}\right)} \ | \ (\phi_x)_{x\in \gG(i)} \right].
\end{multline}
On the event 
$$\cA^h(i):= \{\max_{x\in \gG(i)} |\phi_x|\le |\log h|^2 \},$$
we have for any  $y\in \Xi(i)$, by translation invariance, 
\begin{multline}
 \bE_N\left[  e^{4 \left( h\sum_{x \in \tilde \gL_{N_1}(y)} \delta_x -  
\ind_{\cC_{N_1}(y)}\right)} \ | \ (\phi_x)_{x\in \gG(i)} \right]\\
 \le 
    \max_{\{ \hat \phi \ : \ \|\hat \phi\|_{\infty}\le |\log h|^2 \}} 
     \bE^{\hat \phi}_{2N}\left[e^{4h\left( \sum_{x \in \tilde \gL'_{N_1}}
   \delta_x\right) - 4\ind_{\cC_{N_1}}} 
   \right].
 \end{multline}
and hence we can conclude by taking the expectation of \eqref{yipendence} restricted to the event $\cA^h(i)$ (which includes $\cA^h_N$).

\qed

\subsection{Proof of Proposition \ref{onecell}}\label{proufoix}

Note that because of our choice of $N_1=h^{-1/4}$ 
we always have
\begin{equation}
 h \sum_{x \in \tilde \gL'_{N_1}} \delta_x \le h N^2_1\le h^{1/2},
\end{equation}
which is small.
Hence for that reason, if $h$ is sufficiently small, the Taylor expansion of the exponential gives
\begin{multline}\label{xfiles}
\log \bE^{\hat \phi}_{2N_1}\left[e^{4h\sum_{x \in \tilde \gL'_{N_1}}
   \delta_x- 4\ind_{\cC_{N_1}}}\right]
   \le \log \bE^{\hat \phi}_{2N_1}\Big[1+ 5h\sum_{x \in \tilde \gL'_{N_1}}
   \delta_x- \frac{1}{2}\ind_{\cC_{N_1}}  \Big]
  \\ \le 5h \bE^{\hat \phi}_{2N_1}\Big[ \sum_{x \in \tilde \gL'_{N_1}}
   \delta_x  \Big]- \frac{1}{2}\bP^{\hat \phi}_{2N_1}[\cC_{N_1}]\\
   \le 5 N_1^{-2} \max_{x\in \tilde \gL'_{N_1}}\bP^{\hat \phi}_{2N_1}[\phi_x\in[-1,1]]-\frac{1}{2}\bP^{\hat \phi}_{2N_1}[\cC_{N_1}].
\end{multline}
We have to prove that the r.h.s.\ is small. 
Before going into technical details let us quickly expose the main idea of the proof. For the r.h.s.\ of \eqref{xfiles} to be positive, we need
\begin{equation}\label{hpuissance}
 \frac{\max_{x\in \tilde \gL'_{N_1}}\bP^{\hat \phi}_{2N_1}\big(\phi_x\in[-1,1]\big)}{\bP^{\hat \phi}_{2N_1}[\cC_{N_1}]}\ge \frac{N_1^2}{10}.
\end{equation}
What we are going to show is that for this ratio to be large we need the boundary condition $\hat \phi$ to be very high above the substrate  
(or below by symmetry), but that in that case the quantity $\left(\max_{x\in \tilde \gL'_{N_1}}\bP^{\hat \phi}_{2N_1}[\phi_x\in[-1,1]]\right)$ itself has to be very small and this should allow ourselves to conclude.

\medskip

To understand the phenomenon better we need to introduce quantitative estimates.
Let  $G^*$ denote the Green function \eqref{greenff} in the box $\gL_{2N_1}$ with $0$ boundary condition, and set 
\begin{equation}
 V_{N_1}:= \max_{x\in  \tilde \gL'_{N_1}} G^*(x,x).
\end{equation}
We have from Lemma \ref{Greenesteem}
\begin{equation}
 |V_{N_1}- \frac{1}{2\pi} \log N_1 | \le C.
\end{equation}
Recall that from \eqref{transkix} we have 
\begin{equation}
\bP^{\hat \phi}_{2N_1}\big( \ \phi_x\in[-1,1] \ \big)\le \bP_{2N_1}\left(\phi_x\in \left[-1-H^{\hat \phi}_{2N_1}(x),1-H^{\hat \phi}_{2N_1}(x)\right] \ \right).
\end{equation}
With this in mind we fix
\begin{equation}
u=u(\hat \phi, N_1):=\min_{x\in \tilde \gL_N} |H^{\hat \phi}_{2N_1}(x)|.
\end{equation}
Hence using basic properties of the Gaussian distribution, we obtain (provided that $h$ is sufficiently small)
\begin{equation}\label{krix}
 \max_{x\in \tilde \gL'_{N_1}}\bP^{\hat \phi}_{2N_1}\big(\phi_x\in[-1,1]\big) \le e^{-\frac{(u-1)^2}{2V_{N_1}}}.
\end{equation}
It requires  a bit more work to
 obtain a good lower bound for $\bP^{\hat \phi}_{2N_1}[\cC_{N_1}]$ which is valid for all values of $u$. 
 Fortunately we only need a rough estimate as the factor $N^2_1$ in \eqref{hpuissance} gives us a significant margin in the computation.

\medskip

Recall that $P_t^*$ denotes the two-dimensional heat-kernel with zero boundary condition on $\partial \gL_{2N_1}$. Let us set 
\begin{equation}
V'_{N_1}:= \min_{x \in \tilde \gL'_{N_1}} \int_{(\log N_1)^8}^{\infty} P^*_t(x,x) \dd t.
\end{equation}
From the estimates in Lemma  \ref{lem:kerestimate}, we can deduce that
 \begin{equation}\label{boundvn}
\left|V'_{N_1}-  \frac{1}{2\pi}\left(\log N_1 - 4 \log \log N_1\right)\right|- C.
\end{equation}
For instance we have 
\begin{equation}\label{okless}
\left| \int_{0}^{(\log N_1)^8} P^*_t(x,x) \dd t - \frac{2}{\pi} \log \log N_1 \right|\le \frac{C}{2}.
\end{equation}
for some appropriate $C$ (the estimate is obtained using \eqref{kilcompare} and \eqref{lclt}) so that the result can be deduced from the 
estimate in the Green-function \eqref{eq:stimagreen}.

\begin{proposition}\label{boundforcluster}
For all $h$ sufficiently small, for all
 $\hat \phi$ satisfying $|\hat \phi|_{\infty}<|\log h|^2$, and all $u\in (0, (2\log N_1)^2)$ we have 
 \begin{equation}
  \bP^{\hat \phi}_{2N_1}[\cC_{N_1}]\ge c(\log N_1)^{-1} e^{-\frac{u^2}{2V'_{N_1}}},
 \end{equation}
\end{proposition}

Combining the above result with \eqref{krix} and \eqref{xfiles}
We have 

\begin{multline}\label{zorrib}
  \max_{\{ \hat \phi \ : \ |\hat \phi|\le |\log h|^2 \}} \log 
 \bE^{\hat \phi}_{2N_1}\left[\exp\left( 4h\sum_{x \in \tilde \gL_{N_1}}
   \delta_x- 4\ind_{\cC_{N_1}} \right) \right]\\\le 
   \sup_{u>0}\left( 5N_1^{-2} e^{-\frac{(u-1)^2}{2V_{N_1}}}-  c( 2\log N_1)^{-1}e^{-\frac{u^2}{2V'_{N_1}}}\ind_{\{ u \le (\log N_1)^2\}} \right)
  \\ = \sup_{u>0} \frac{5e^{-\frac{(u-1)^2}{2V_{N_1}}}}{N_1^2}\left[1- \frac{cN_1^2}{10(\log N_1)} e^{-
  \frac{u^2(V_{N_1}-V'_{N_1})}{2V'_{N_1}V_{N_1}}-\frac{2u-1}{2V_{N_1}}} \ind_{\{ u \le (\log N_1)^2\}}\right].
\end{multline}
Now note that for the second factor to be positive, we need one of the terms in the exponential to be at least of order $\log N_1$ in absolute value.
Using the estimates we have for $V'_{N_1}$ and $V_{N_1}$, we realize that the exponential term is larger than
$$ c\exp\left(-\frac{cu^2(\log \log N_1)}{\log N_1}\right),$$
and hence the expression is negative if 
$u^2\le c (\log N_1)^3 (\log \log N_1)^{-1},$
for some small $c$.
For the other values of $u$ we can just consider the first factor which already gives a satisfying bound, and we can conclude that the l.h.s.\ of \eqref{zorrib}
is smaller than
 \begin{equation}
 e^{-\frac{c (\log N_1)^2}{\log \log N_1}}
   \le e^{-c |\log h|^{3/2}}.
\end{equation}

\subsection{Decomposing the proof of Proposition \ref{boundforcluster}}

We show here how to split the proof the proposition into three lemmas which we prove in the next subsection.
Set 
\begin{equation}
 x_{\min}:=\argmin_{x\in \tilde \gL'_N} |H^{\hat \phi}_{2N_1}(x)|,
\end{equation}
(it is not necessarily unique but in the case it is not we choose one minimizer in a deterministic manner)
and 
\begin{equation}
\hat \gL:= \{ z \in \tilde\gL'_{N_1} \ : \  |x_{\min}-z|\le (\log N_1)^2 \}
\end{equation}
We bound from below the probability of $\cC_{N_1}$ by only examining the possibility of having a cluster of contact around $x_{\min}$.
Using \eqref{greenff} we have
\begin{multline}\label{kxmin}
 \bP^{\hat \phi}_{2N_1}[ \cC_{N_1}]\le  \bP^{\hat \phi}_{2N_1}\left[ \sum_{z\in \hat \gL} \delta_z \ge 
 (\log N_1)^3 \right]\\
 =  \bP_{2N_1}\left[ \sum_{z \in \hat \gL} \ind_{[-1,1]}\left(\phi_z+H^{\hat \phi}_{2N_1}(x)\right)\ge 
 (\log N_1)^3 \right].
\end{multline}

To estimate the last probability, we first remark that for $x \in \hat \gL$, $H^{\hat \phi}_{2N_1}(x)$ is very close to $H^{\hat \phi}_{2N_1}(x_{\min})$
which we assume to be equal to $-u$ for the rest of the proof (the case $H^{\hat \phi}_{2N_1}(x_{\min})=+u$ is exactly similar).
The factor $\log N_1$ in the estimate is not necessary, but it yields a much simpler proof.

\begin{lemma}\label{regular}
We have for all $x, y \in \tilde \gL'_{N_1}$
\begin{equation}
\left| H^{\hat \phi}_{2N_1}(x)-H^{\hat \phi}_{2N_1}(y)  \ \right| \le \frac{C |\hat \phi|_{\infty}(\log N_1)|x-y|}{N_1}. 
\end{equation}
In particular if $h$ is sufficiently small, $|x-y|\le (\log N_1)^2$ and $|\hat \phi|_{\infty}\le |\log h|^2$,
we have 
\begin{equation}
 |H^{\hat \phi}_{2N_1}(x)-H^{\hat \phi}_{2N_1}(y)|\le 1/4.
\end{equation}
\end{lemma}

\medskip

Then to estimate the probability for $\phi$ to form a cluster of point close to height $u$, we decompose the field 
$(\phi_x)_{x\in \hat \gL}$ into a rough field $\phi_1$ which is 
almost constant on the scale $(\log N_1)^2$ and an independent field  $\phi_2$ which accounts for the local variations of $\phi$.
We set 
\begin{equation}
\begin{split}
Q^1(x,y):=  \int_{(\log N)^8}^{\infty} P^*_t(x,y) \dd t,\\
Q^2(x,y):=  \int_0^{(\log N)^8} P^*_t(x,y) \dd t.
\end{split}
\end{equation}
We let $(\phi_1(x))_{x\in \gL_{2N_1}}$ and $(\phi_2(x))_{x\in \gL_{2N_1}}$ denote two independent centered fields with respective covariance function
$Q^1$ and $Q^2$.
By construction the law of $\phi_1+\phi_2$ has a law given by $\bP_{2N_1}$, and thus we set for the remainder  of the proof
\begin{equation}
\phi:= \phi_1+\phi_2,
\end{equation}
and use $\bP_{2N_1}$ to denote the law of $(\phi_1, \phi_2)$. We have by standard properties of Gaussian variables that for every $u>0$, and for $h$ sufficiently small
\begin{equation}\label{onyva}
 \bP_{2N_1}\left[\phi_1(x_{\min})\in [u-1/4,u+1/4]\right]\ge \frac{1}{4 \sqrt{ 2\pi V_{N_1} }}e^{-\frac{u^2}{2V'_{N_1}}} \ge \frac{1}{5 \sqrt{\log N_1}}e^{-\frac{u^2}{2V'_{N_1}}}.
\end{equation}
Now we have to check that the field  $\phi_1$ remains around level $u$ on the whole box $\hat \gL$.

\medskip

\begin{lemma}\label{locareg}
There exists a constant $c$ such that for all $h$ sufficiently small we have
\begin{equation}
\bP_{2N_1}\left[  \exists y\in \hat \gL, \ |\phi_1(y)-\phi_1(x_{\min})| >1/4 \right] \le  e^{-c(\log N_1)^4}.
\end{equation}

\end{lemma}

\medskip

Finally we show that it is rather likely for $\phi_2$ to have a lot of points around level zero.

\medskip

\begin{lemma}\label{lastep}
There exists a constant $c$ such that for all $h$ sufficiently small we have
\begin{equation}\label{hooop}
 \bP_N\left[ \sum_{z\in \hat \gL} \ind_{\{ |\phi_2(z)|\le 1/4\} } \ge (\log N_1)^3 \right]\ge c(\log\log  N_1)^{-1/2}.
\end{equation}
\end{lemma}

\medskip

\noindent We can now combine all these ingredient into a proof 

\medskip

\begin{proof}[Proof of Proposition \ref{boundforcluster}]

According to Lemma \ref{regular}, if $|x_{\min}-z|\le (\log N_1)^2$ we have 
\begin{multline}
\left\{ (\phi_z+H^{\hat \phi}_{2N_1}(z))\in [-1,1] \right\} \supset \left\{ \phi_z \in [-3/4+u,3/4+u] \right\} \\
\supset 
\left\{ \ \left |\phi_1(x_{\min})-u  \right| \le 1/4  \right\}\cap \left\{ |\phi_1(x_{\min})-\phi_1(z)| \le 1/4 \right\} \cap \left\{ |\phi_2(z)| \le 1/4 \right\}.
\end{multline}
Thus we obtain as a consequence 
\begin{multline}
\left\{ \sum_{z \in \hat \gL} \ind_{[-1,1]}(\phi_z+H^{\hat \phi}_{2N_1}(x)) \ge
 (\log N_1)^3\right\} \\
 \subset   \left\{  \ \left|\phi_1(x_{\min})-\phi_1(z) \right| \le 1/4 \right\} \cap \left\{  \forall z\in \hat \gL, \  |\phi_1(x_{\min})-\phi_1(z)| \le 1/4 \right\}\\
 \cap  \Big\{  \sum_{z\in \hat \gL} \ind_{\{ |\phi_2(z)|\le (1/4)\} } \ge (\log N_1)^3 \Big\}.
\end{multline}
Using \eqref{onyva} combined with Lemmas \ref{locareg} and \ref{lastep} and the independence of $\phi_1$ and $\phi_2$ we conclude that 
\begin{multline}
\bP_{2N_1} \left[ \sum_{z \in \hat \gL} \ind_{[-1,1]}(\phi_z+H^{\hat \phi}_{2N_1}(x) \ge 
 (\log N_1)^3 \right] \\
 \ge \left[  \frac{c}{\sqrt{ \log N_1 }}e^{-\frac{u^2}{2V'_{N_1}}} - e^{-c(\log N)^4} \right]c(\log N_1)^{-1/2} \ge \frac{c'}{(\log N_1)}e^{-\frac{u^2}{2V'_{N_1}}}.
\end{multline}
where the last inequality is holds if $u\le (\log N_1)^2$ and $h$ is sufficiently small.
We can thus conclude using \eqref{kxmin}.
\end{proof}

\subsection{Proof of the technical lemmas}

\begin{proof}[Proof of Lemma \ref{regular}]
Given $x, y \in \tilde \gL'_{N_1}$,
let $X^x$ and $X^y$ be two simple random walk starting from $x=(x_1,x_2)$ and $y=(y_1,y_2)$, and coupled as follows: the coupling is made as the product of two
one-dimensional couplings, along each coordinate  the walk are independent until the coordinate match, then they move together. 
Let $\tau_{x,y}$ be the time where the two walks meet and  $\tau^x_{\partial\gL_{2N_1}}$ be the time
when $X^x$ hits the boundary.
Recalling \eqref{RWrepresent} we have
\begin{equation}
\left|H^{\hat\phi}_{2N_1}(y)- H^{\hat\phi}_{2N_1}(x)\right| \le |\hat \phi|_{\infty} P\left[ \tau_{x,y}< \tau^x_{\partial \gL_{2N_1}}\right].
\end{equation}
We conclude by showing that
\begin{equation}
  P\left[ \tau_{x,y}< \tau^x_{\partial \gL_{2N_1}} \right]< \frac{C|x-y| (\log N_1) }{N_1}
\end{equation}
By union bound, we can reduce to the one dimensional case. 
Let $Y^x$ and $Y^y$ denote the first coordinates of $X^x$ and $X^y$. Until the collision time, they are  two independent one dimensional random walk in 
$\lint 0 , 2N_1 \rint$ 
with initial condition $x_1$ and $y_1$  in 
 $\lint N_1/2 , 3N_1/2 \rint$. Let $T_{x,y}$ and $T'$ denote respectively their collision time and the first hitting time of $\{0,2N_1\}$ for $Y^x$.
We are going to show that 
\begin{equation}
  P\left[ T'<T_{x,y} \right]< C \frac{|x_1-y_1|(\log N_1)}{N_1}
\end{equation} 
Note that before collision, $Y^x-Y^y$ is a nearest neighbor random-walk with jump rate equal to $2$ and for that reason we have for any $t$
\begin{equation}
 P[ T_{x,y}>t]\le C |x_1-y_1| t^{-1/2}.
\end{equation}
On the other hand, we have for any $t\ge N_1^2$
\begin{equation}\label{sectr}
 P[T\le t] \le  2 \exp\left(- \frac{cN_1^2}{t} \right).
\end{equation}
We can conclude choosing $t=N^2_1 (\log N_1)^2$.
\end{proof}

\bigskip

\begin{proof}[Proof of Lemma \ref{locareg}]

We obtain the result simply by performing a union bound on $y\in \hat \gL$. Hence we only need to prove a bound on the variance 

\begin{multline}
\bE_{2N_1} \left[ \left(\phi_1(y)-\phi_1(x_{\min})\right)^2 \right] \\
 \le \int_{(\log N_1)^8}^{\infty} \left[P^*_t(x_{\min},x_{\min})- P^*_t(x_{\min},y)-P^*(y,y) \right]\dd t 
\end{multline}
Using \eqref{gradientas}, we obtain that for any $y\in \hat \gL$
\begin{equation}
\bE_{2N_1} \left[ \left(\phi_1(y)-\phi_1(x_{\min})\right)^2 \right]\le C(\log N_1)^{-4},
\end{equation}
and thus that 
\begin{equation}
\bP_{2N_1} \big[ \ |\phi_1(y)-\phi_1(x_{\min})|\ge 1/4 \big]\le |\hat \gL| e^{- c(\log N_1)^4 },
\end{equation}
which allows to conclude 

\end{proof}

\begin{proof}[Proof of Lemma \ref{lastep}]

We set 
$$J:= \sum_{\{z \ : \  |x_{\min}-z|\le (\log N_1)^2 \}} \ind_{\{\phi_2(z)\in [-1/4,1/4]\}}.$$
Using the fact that the sum is deterministically bounded by $C(\log N_1)^4$,
we have
\begin{equation}\label{chubby}
\bP_{2N_1} \left[  J \ge \frac{\bE_N[J]}{2} \right]\ge \frac{\bE_N[J]}{2C (\log N_1)^4}.
\end{equation}
From \eqref{okless}, we have for small $h$,
\begin{equation}
\Var(\phi_2(x))= Q^2(x,x)\le \log \log N_1.
\end{equation}
Then as $\phi_2(x)$ are centered Gaussians, we have
\begin{equation}
\bE_{2N_1} \left[  \sum_{\{z \ : \  |z-x_{\min}|\le (\log N_1)^2\}} \ind_{\{\phi_2(z)\in [-1/4,1/4]\}} \right]\ge c (\log N_1)^{4} (\log \log N_1)^{-1/2},
\end{equation}
which combined with \eqref{chubby} allows to conclude.

\end{proof}

\section{Finite volume criteria: adding mass and changing the boundary condition}\label{finicrit}

Let us remark that it seems technically easier to get a lower bound 
for $N^{-2}\bbE \left[\log Z^{\gb,\go}_{N,h}\right]$ for a given $N$ than to prove one directly for the limit.
However as there is no obvious sub-additivity property which allows to compare the two.

\medskip

In \cite{cf:GL}, for $d\ge 3$ we introduced the idea of replacing the boundary condition by an infinite volume free field in order to recover sub-additivity . 
In dimension $2$, the infinite volume free field does not exists as the variance diverges with the distance to the boundary of the domain.
A way to bypass the problem it to artificially introduce mass and then to find a comparison between the free energy of the system with massive free field
and the original one. This is the method that we adopted in our previous paper
(see \cite[Proposition 7.1 and Lemma 7.2]{cf:GL}). However our previous results turn out
out to be a bit two rough for our proof.
We present here an improvement of it (Proposition \ref{th:finitevol}) on which we build the proof of Theorem \ref{mainres}. 

\medskip

\subsection{A first finite volume criterion}\label{nonoptfinit}

Let us recall the comparison used in \cite{cf:GL}. Even 
it is not sufficient for our purpose in this paper, it will help us to explain the improvement presented in Section \ref{optfinit}.
Given $u>0$ and $m>0$,
we introduce the notation $$\delta^u_x:=\ind_{[u-1,u+1]}(\phi_x)$$ 
and set 
\begin{equation}
 Z^{\gb,\go,m}_{N,h,u}:= \bE^m_N\left[\exp\left( \sum_{x\in  \tilde \gL_N} (\gb \go_x-\gl(\gb)+h)\delta^u_x\right)\right].
\end{equation}
and 
\begin{equation}\label{massfree}
 \tf(\gb,h,m,u):=\lim_{N\to \infty}\frac{1}{N^d} \log Z^{\gb,\go,m}_{N,h,u}.
\end{equation}
The existence of the above limit is proved in \cite{cf:GL}.
We can compare this free-energy to the original one using the following result.

\begin{proposition}\label{massivecompa}
 We have for every $u$ and $m$
 \begin{equation}\label{massivecompar}
\tf(\gb,h,m,u)\le \tf(\gb,h)+f(m).
\end{equation}
where
\begin{equation}
f(m):= \frac 1 2\int_{[0,1]^2}  \log \left( 1+ \frac{m^2}{4 \left[\sin^2(\pi x/2)+ \sin^2(\pi y/ 2)\right]}\right) \dd x \dd y.
\end{equation}
There exists $C>0$ such that for every $m\le 1$ we have 
\begin{equation}\label{asymf}
 \left| f(m)-\frac{1}{4\pi}m^2|\log m| \right| \le Cm^2.
\end{equation}
Moreover for all $N$ we have 
\begin{equation}\label{subbadd}
 \tf(\gb,h,m,u)\ge \frac{1}{N^2} \hat \bE^m \bbE \left[ \log \bE^{m,\hat \phi}_N
 \left[ \exp\left( \sum_{x\in  \tilde \gL_N} (\gb \go_x-\gl(\gb)+h)\delta^u_x\right)\right]\right].
\end{equation}
\end{proposition}

\begin{proof}[Sketch of proof]
The result is proved in \cite{cf:GL} (as Proposition 7.1 and Lemma 7.2) but let us recall briefly how it is done.
For the first point, we have to remark that changing  the height of the substrate (i.e.\ replacing $\delta_x$ by $\delta^u_x$ in \eqref{eq:modZ}) 
for the original model  does not change the value of the free energy, that is , $$\tf(\gb,h,0,u)=\tf(\gb,h,0,0), \quad \text{ for all values of } u.$$ 
Heuristically this is because the free field Hamiltonian is translation invariant but a proof is necessary to show that the boundary effect are indeed negligible 
(see \cite[Proposition 4.1]{cf:GL}). Note that for the massive free field, the limit \eqref{massfree} really depends on $u$ because adding an harmonic confinement 
breaks this translation invariance.

\medskip

Then we can compare the partition of the two free fields by noticing that
the density of the massive field with respect to the original one (recall \eqref{massivedensity}) satisfies 
 \begin{equation}\label{reldensity}
\frac{\dd \bP^m_N}{\dd \bP_N}(\phi):=\frac{\exp\left( -\frac{m^2}{2}\sum_{x \in \mathring\gL_N} \phi^2_x\right)}{\bE_N\left[\exp\left( -\frac{m^2}{2}\sum_{x \in \mathring\gL_N} \phi^2_x\right)\right]}
\le \frac{1}{\bE_N\left[\exp\left( -\frac{m^2}{2}\sum_{x \in \mathring\gL_N} \phi^2_x\right)\right]},
\end{equation}
and that 
\begin{equation}\label{lmimw}
 \lim_{N\to \infty}\frac{1}{N^2}\log \bE_N\left[\exp\left( -\frac{m^2}{2}\sum_{x \in \mathring\gL_N} \phi^2_x\right)\right]=: \lim_{N\to \infty}\frac{1}{N^2}\log W^m_N =-f(m).
\end{equation}
Equation \eqref{subbadd} then  follows from of a sub-additive argument (see the proof of Proposition 4.2.\ in \cite{cf:GL} or that of \eqref{subaddt} below).
\end{proof}
\medskip

\begin{rem}
Note that Proposition \ref{massivecompa} gives a bound on $\tf(\gb,h)$ which depends only on the partition function of a finite system .
\begin{equation}\label{finitevol1}
\tf(\gb,h)\ge \frac{1}{N^2} \hat \bE^m \bbE\left[ \log \bE^{m,\hat \phi}_N
 \left[ \exp\left( \sum_{x\in  \tilde \gL_N} (\gb \go_x-\gl(\gb)+h)\delta^u_x\right)\right]\right]-f(m).
\end{equation}
In particular we can prove Theorem \ref{mainres}, if for any $h>0$, $\gb\in(0,\ubgb)$ we can find values for $u$ and $m$ and $N$ such that the l.h.s.\ is positive.
 However it turns out that  with our techniques, we cannot prove that the l.h.s is positive for very small $h$.
This is 
mostly because of the presence of a  $|\log m|$ factor in the asymptotic behavior of $f(m)$ around $0$. Therefore we need a better criterion in which the subtracted term is proportional to $m^2$.
\end{rem}

\subsection{A finer comparison}\label{optfinit}

To obtain a more efficient criterion, we want to restrict the
partition function to a set of $\phi$ where $(\dd \bP^m_N/\dd \bP_N)(\phi)$ is much smaller than $\exp(N^2f(m))$.
We define $\cD^0_N$ as a set where the density  $(\dd \bP^{m}_N/\dd \bP_N)(\phi)$ takes ``typical'' values (see Proposition \ref{thednproba}).
For some constant $K>0$, we set
\begin{equation}\label{dnzero}
 \cD^0_N:=\left\{  \sum_{x\in \gL_N} \phi(x)^2 \ge N^2\left( \frac{f(m)}{m^2}- K\right) \right\}.
 \end{equation}
 Recall that $\hat\bP^{m}$ denotes the law of the infinite volume massive free field (see Section \ref{secmass}) for the boundary condition $\hat \phi$.
 
 \medskip
 
\begin{proposition}\label{th:finitevol}
For any value of $N$, and $K$ and $m$ we have 
\begin{equation}\label{finitesds}
 \tf(\gb,h) \ge \frac{1}{N^2}\hat\bE^{m} \bbE\left[ \log \bE^{m,\hat \phi}_N \left[ \exp\left(\sum_{x\in \tilde \gL_N} \left( \gb\go_x-\gl(\gb)+h \right)\delta^{u}_x\right)\ind_{\cD^0_N}\right] \right]
-Km^2.
\end{equation}
\end{proposition}

\medskip

\noindent With the idea of working with a measure that does not depend on the boundary condition, we set similarly to \eqref{deltaf}
\begin{equation}
 \delta^{\hat\phi,u}_x:= \ind_{[u-1,u+1]}(\phi(x)+H^{m,\hat \phi}_N(x)),
\end{equation}
and 
\begin{equation}\label{defDn}
\cD_N:= \left\{ \phi \ : \ \sum_{x\in \tilde \gL_N} (\phi_x+H^{m,\hat \phi}_N(x))^2 \ge N^2\left( \frac{2f(m)}{m^2}- K\right)\right\}.
\end{equation}
With this notation and in view of the considerations of Section \ref{grbc} the expected value in the l.h.s.\ in \eqref{finitesds} is equal to 
\begin{equation}\label{otherexpr}
\hat\bE^m \bbE\left[ \log \bE^{m}_N \left[ \exp\left(\sum_{x\in \tilde \gL_N} \left( \gb\go_x-\gl(\gb)+h \right)\delta^{u,\hat\phi}_x\right)\ind_{\cD_N} \right]\right].
\end{equation}

\subsection{Using the criterion}
 Before giving a proof of Proposition \ref{th:finitevol} let us show how we are going to use it to prove our lower bound on the free energy \eqref{breaks}.
for the remainder of the proof we set
 \begin{equation}\label{parameters}\begin{split}
 N_h&:= \exp(h^{-20}),\\ 
 m_h&:= N^{-1}_h (\log N_h)^{1/4},\\
 u_h&:=\sqrt{\frac{2}{\pi}}\log N_h-\frac{2+\alpha}{2\sqrt{2\pi}}\log \log N_h,
\end{split}
\end{equation}
where $\alpha=3/4$ (we find that the computations are easier to follows with the letter $\alpha$ instead of a specific number, 
in fact any value in the interval $(11/20,1)$ would also work).
With Proposition \ref{th:finitevol}, the proof of the lower bound in \eqref{breaks}
is reduced to the following statement, whose proof will be detailed in the next three sections.
\begin{proposition}\label{mainproposition}
For any $\gb\le \ubgb$, there exists $h_0(\gb)$ such that for any $h\in (0,h_0(\gb))$  
 
 \begin{equation}
 \hat\bE^m 
 \bbE \left[ \log \bE^{m_h,\hat\phi}_{N_h} \left[ \exp\left(\sum_{x\in \tilde \gL_{N_h}} \left( \gb\go_x-\gl(\gb)+h \right)\delta^{u_h}_x\right)
 \ind_{\cD^0_{N_h}} \right]\right]
-K(m_hN_h)^2\ge   1.
 \end{equation}
 \end{proposition}
 
 \medskip
 
 \noindent Indeed the result directly implies that 
 \begin{equation}
  \tf(\gb,h)\ge (N_h)^{-1}.
 \end{equation}

\subsection{Proof of Proposition \ref{th:finitevol}}\label{homogg}
Let us start by setting 
\begin{multline}
Z'_N(\hat \phi)=Z'_N:= \bE^m_N \left[ \exp\left(\sum_{x\in \tilde \gL_N} \left( \gb\go_x-\gl(\gb)+h \right)\delta^{\hat \phi,u}_x\right)\ind_{\cD_N} \right]\\
=\bE^{m,\hat \phi}_N \left[ \exp\left(\sum_{x\in \tilde \gL_N} \left( \gb\go_x-\gl(\gb)+h \right)\delta^{u}_x\right)\ind_{\cD^0_N} \right].
\end{multline}
A simple computation (see below) is sufficient to show that for any $k\ge 0$ we have 
\begin{equation}\label{subaddt}
\hat\bE^m \bbE \left[ \log Z'_{2^k N}\right]\ge 4^k \hat\bE^m \bbE \left[\log Z'_{N}\right].
\end{equation}
Hence that it is sufficient to prove \eqref{finitesds} with $N$ replaced by $2^k N$ for an arbitrary integer $k$, or by the limit when 
$k$ tends to infinity.

\medskip

\noindent Let us prove \eqref{subaddt}.
We divide the box $\gL_{2N}$ into $4$ boxes, $\gL^i_N$, $i=1,\dots,4$. Set 
\begin{equation}\begin{split}
\gL^i_N&:=\gL_N+(\alpha_1(i),\alpha_2(i))N \\
\tilde \gL^i_N&:=\tilde \gL_N+(\alpha_1(i),\alpha_2(i))N,
\end{split}\end{equation}
where $\alpha_j(i)\in\{0,1\}$ is the $j$-th digit of the dyadic development of $i-1$.
Set  
 \begin{equation}
\cD^{0,i}_N:= \left\{ \sum_{x\in \tilde \gL^i_N} \phi(x)^2\ge \left( \frac{2f(m)}{m^2}-K\right)N^2 \right\}.
\end{equation}
We notice that 
\begin{equation}\label{dainclusion}
\bigcap_{i=1}^4 \cD^{0,i}_{N}\subset\cD^0_{2N}.
\end{equation}
We define 
\begin{equation}
\gG_{N}:=\left( \bigcup_{i=1}^{4} \partial \gL^i_N \right) \setminus \partial\gL_{2N}.
\end{equation}
If we condition on the realization on $\phi$ in $\gG_{N}$, 
the partition functions of the system of size $2N$
factorizes into $4$ partition functions of systems of size $N$, whose boundary conditions are determined 
by $\hat \phi$ and $\phi|_{\gG_N}$, and we obtain
\begin{multline}
\bE^{m,\hat \phi}_{2N} \left[ \exp\left( \sum_{x\in  \tilde\gL_{2N}} (\gb \go_x-\gl(\gb)+h)\delta^u_x\right)\ind_{\bigcap_{i=1}^4 \cD^{0,i}_{N}} \ \Bigg| \ \phi|_{\gG_N} \right]
\\ = \prod_{i=1}^{4}\bE^{m,\hat \phi}_{2N} 
\left[ \exp\left( \sum_{x\in  \tilde\gL^i_{N}} (\gb \go_x-\gl(\gb)+h)\delta^u_x\right)\ind_{\cD^{0,i}_{N}} \ \Bigg| \ \phi|_{\gG_N} \right]
=: \prod_{i=1}^{4} \tilde Z^i(\hat \phi , \phi|_{\gG_N} , \go  ).
\end{multline}
By the spatial Markov property for the infinite volume field, each $\tilde Z^i(\hat \phi , \phi|_{\gG_N} , \go )$ 
has the same distribution as $Z'_N$ (if $\hat \phi$ and $\phi|_{\gG_N}$ have distribution
 $\hat \bE^m$ and $\bE^{m,\hat \phi}_{2N}$ respectively and the $\go_x$s are IID).
Using \eqref{dainclusion} and Jensen's inequality for $ \bE^{m,\hat \phi}_{2N} \left[ \ \cdot \ | \ \phi|_{\gG_N} \right]$
we have
\begin{equation}\begin{split}
\bbE \hat \bE^m \left[  \log Z'_{2N}\right]\ge 
 \sum_{i=1}^{4}   \bbE \hat \bE^m  \bE^{m,\hat \phi}_{2N}\left[ \log \tilde Z^i(\hat \phi , \phi|_{\gG_N} \ ) \right]
 =4 \bbE \hat \bE^m\left[  \log Z'_{N}\right],
 \end{split}\end{equation} 
which ends the proof of \eqref{subaddt}.

\medskip

\noindent Now we set  $M:=2^kN$ with $k$ large. In the computation, we write sometimes $H$ for $H^{m,\hat \phi}_M$ for simplicity. We remark that for $\phi\in \cD_M$ we have  
\begin{multline}
\log \left( \frac{\dd P^{m}_M}{\dd P_M}(\phi) \right)\\
= \frac{m^2}{2} \left(  
\sum_{x\in \gL_M} H^2(x)-\sum_{x\in \gL_M} (\phi_x-H(x))^2-2\sum_{x\in \gL_M} \phi_x H(x)\right) -\log W_M\\
\le\left[  M^2 \left(\frac{ m^2 K}{2}-f(m)\right)-\log W_M \right] +  m^2 \left[ \sum_{x\in \gL_M} \phi_x H(x)+ \frac{1}{2} \sum_{x\in \gL_M} H^2(x)\right]\\
\le m^2 M^2 K+  m^2 \left[ \sum_{x\in \gL_M} \phi_x H(x)+ \frac{1}{2} \sum_{x\in \gL_M} H^2(x)\right],
\end{multline}
where the first inequality follows from the definition of $\cD_M$ \eqref{defDn} and the last one from \eqref{lmimw} and is valid provided $K$ is sufficiently large.
From this inequality we deduce that 
\begin{multline}\label{rhss}
Z'_{M}\le e^{m^2K M^2} \bE_M\left[ e^{\sum_{x\in \tilde \gL_M} \left( \gb\go_x-\gl(\gb)+h \right)\delta^{\hat \phi,u}_x} e^{m^2  \sum_{x\in \gL_M} \left[H(x)\phi(x)
+\frac{H(x)^2}{2}\right]} \right]\\
:=e^{m^2K M^2} Z''_M.
\end{multline}
To conclude the proof, we must show that the r.h.s.\ is not affected, in the limit, by the presence of $H$
(which produces the two last terms and enters in the definition of $\delta^{\hat \phi,u}_x$)
i.e.\ that 
\begin{equation}\label{zsecond}
\lim_{M\to \infty} \frac{1}{M^2} \bbE \hat \bE^m \left[ \log Z''_M \right]=\tf(\gb,h).
\end{equation}
We can replace $\delta^{\hat \phi,u}_x$ by $\delta^u_x$ at the cost of a Girsanov-type term in the density.
For computations, it is pratical to define
\begin{equation}
H^0(x):= H(x)\ind_{\{x\in \mathring{\gL}_N\}}.
\end{equation}
The distribution $\phi+H^0$ under $\bP_M$ is absolutely continuous 
with respect to that of $\phi$. The density of its distribution $\tilde \bP_M$ with respect to $\bP_M$ is given by
\begin{multline}\label{densityxx}
 \frac{\dd \tilde \bP_M}{\dd \bP^0_M}(\phi)=\exp\bigg(\frac 1 2 \sum_{\gL_M} (\grad \phi)^2-(\grad \phi -\grad H^0)^2 \bigg)\\= 
 \exp\bigg(-\frac 1 2 \sum_{\gL_M} (\grad H_0)^2+  \sum_{\gL_M}\grad \phi\grad H^0 \bigg)\\
 =\exp\bigg( -m^2\sum_{x\in \gL_M} \left( H(x)\phi(x)- \frac{H^0(x)^2}{2} \right)\\
 + \sum_{x\in \partial\gL_M}\sumtwo{y\in  \partial^- \gL_M }{y\sim x}\left( H(x)\phi(y)-\frac{H(x)H(y)}{2} \right)\bigg),
\end{multline}
where we used the notation
\begin{equation}
\sum_{\gL_M} \grad R \grad T:= \frac{1}{2}\sumtwo{x,y \in \gL_M}{x\sim y} (R(x)-R(y))(T(x)-T(y)).
\end{equation}
To obtain the second line in \eqref{densityxx} we have used 
the summation by part formula (which is valid without adding boundary terms since the functions we are integrating have zero boundary condition) and \eqref{defH} to obtain
\begin{equation}\begin{split}
                \sum_{\gL_M} \grad H \grad \phi&=- \sum_{x\in \mathring{\gL}_M} \gD H(x)\phi(x)= -m^2 \sum_{x\in \gL_M} H(x)\phi(x),\\
                \sum_{\gL_M} \grad H \grad H^0&=- \sum_{x\in \mathring{\gL}_M} \gD H(x)H^0(x)= -m^2 \sum_{x\in \gL_M} H^0(x)^2.
                 \end{split}
\end{equation}
The substitution of $H$ by $H^0$ produces the second term (boundary effects).
Hence the expectation in \eqref{rhss} is equal to  (assume $u>1$)
\begin{multline}\label{Zseccc}
 \exp\bigg(\sum_{x\in \partial \gL_M \cap \tilde \gL_M}\left( \gb\go_x-\gl(\gb)+h \right)\ind_{[u-1,u+1]}(\hat \phi(x))\\ +m^2\sum_{x\in \tilde \gL_M} 
 \frac{H(x)^2+H^0(x)^2}{2}
 -\sum_{x\in \partial\gL_M}\sumtwo{y\in  \partial^- \gL_M}{y\sim x} \frac{H(x)H(y)}{2}   \bigg)\\
 \times
 \bE^0_M\left[ \exp\left(\sum_{x\in \tilde \gL_M} \left( \gb\go_x-\gl(\gb)+h \right)\delta^{u}_x- 
 \sumtwo{x\in \partial\gL_M, y\in \partial^- \gL_M}{y\sim x} H(x)\phi(y) \right)\right].
\end{multline}
Let us show first that the exponential term in front of the expectation in \eqref{Zseccc} does not affect the limit of $M^{-2}\log Z''_M$.
We have 
\begin{multline}\label{limitain}
\lim_{M\to \infty} \bbE \hat\bE^m \left|\frac{1}{M^2} \sum_{x\in \partial \gL_M \cap \tilde \gL_M}\left( \gb\go_x-\gl(\gb)+h \right)\ind_{[u-1,u+1]}(\hat \phi(x))\right|
\\ \le  \lim_{M\to \infty} \frac{1}{M^2} \bbE \sum_{x\in \partial \gL_M \cap \tilde \gL_M}| \gb\go_x-\gl(\gb)+h |=0.
\end{multline}
For the other terms, set
$$\mathcal M_M:= \max_{x \in \partial \gL_M}\hat\phi(x).$$
Being a maximum over $4M$ Gaussian variables of finite variance, it is not difficult to check that for all $M$ sufficiently large,
\begin{equation}
\hat\bE[\cM^2_M]\le (\log M)^2.
\end{equation}
Moreover from the definition of $H^{\hat \phi,N}_M$ we gave for any $x\in \mathring \gL_{M}$ we have 
\begin{equation}
 |H^{\hat \phi,N}_M(x)|= \frac{1}{2d+m^2}\left|\sum_{y\sim x}H^{\hat \phi,N}_M(y)\right|\le \frac{2d}{2d+m^2} \max_{y\sim x} |H^{\hat \phi,N}_M(y)|.
\end{equation}
This implies that the maximum of $H$ is attained on the boundary and that 
\begin{equation}
 |H^{\hat \phi,N}_M(x)|\le \cM_M  \left(\frac{2d}{2d+m^2}\right)^{d(x,\partial \gL_M)}.
\end{equation}
This implies that
\begin{equation}
\left|m^2\sum_{x\in \tilde \gL_M} 
 \frac{H(x)^2+H^0(x)^2}{2}\\
 -\sum_{x\in \partial\gL_M}\sumtwo{y\in  \partial^- \gL_M}{y\sim x} \frac{H(x)H(y)}{2}   \bigg)\right|\le C_m M \cM^2_M.
\end{equation}
In particular we have
\begin{equation}\label{limitdeux}
\lim_{M\to \infty} \frac{1}{M^2} \hat \bE^m \left |m^2\sum_{x\in \tilde \gL_M} 
 \frac{H(x)^2+H^0(x)^2}{2}
 -\sum_{x\in \partial\gL_M}\sumtwo{y\in  \partial^- \gL_M}{y\sim x} \frac{H(x)H(y)}{2}   \right|=0.
\end{equation}
Hence from \eqref{Zseccc}, \eqref{limitain} and \eqref{limitdeux}, Equation \eqref{zsecond} holds provided we can show that
\begin{equation}\label{statz}
 \lim_{M\to \infty } \frac{1}{M^2}\bbE \hat \bE^m\log 
 \bE_M\left[ \exp\left(\sum_{x\in \tilde \gL_M} \left( \gb\go_x-\gl(\gb)+h \right)\delta^{u}_x+ T(\hat \phi, \phi)\right)
\right]
 =\tf(\gb,h),
\end{equation}
where we have used the notation
\begin{equation}
T(\hat \phi, \phi):=\sum_{x\in \partial\gL_M}\sumtwo{y\in \partial^-{\gL}_M}{y\sim x}\hat\phi(x)\phi(y).
\end{equation}
This is extremely similar to the proof of \cite[Proposition 4.2]{cf:GL} but we include the main line of the computation for the sake of completeness.
First we note that because of uniform integrability, \eqref{statz} holds if we prove the convergence in probability,
\begin{equation}\label{statz3}
\lim_{M\to \infty}\frac{1}{M^2}\log 
 \bE_M\left[ \exp\left(\sum_{x\in \tilde \gL_M} \left( \gb\go_x-\gl(\gb)+h \right)\delta^{u}_x+ T(\hat \phi, \phi)\right)
\right]=\tf(\gb,h).
\end{equation}
\medskip

Note that conditioned to $\hat \phi$, $T(\hat \phi, \phi)$ is a centered Gaussian random variable.
We show in fact $\hat \bP^m \otimes \bbP$ almost sure convergence for rather than convergence of the expectation of \eqref{statz}, but since 
\begin{multline}
\left|M^{-2}\log 
 \bE_M\left[ \exp\left(\sum_{x\in \tilde \gL_M} \left( \gb\go_x-\gl(\gb)+h \right)\delta^{u}_x+ T(\hat \phi, \phi)\right)\right] \right|\\ 
 \le 
M^{-2}\sum_{x\in \tilde \gL_M} | \gb\go_x-\gl(\gb)+h|+ M^{-2} \log \bE_M\left[e^{T(\hat \phi,\phi)}\right]\\
=M^{-2}\sum_{x\in \tilde \gL_M} | \gb\go_x-\gl(\gb)+h|+ M^{-2}\Var_{\bP_M} \left(T(\hat \phi,\phi) \right),
\end{multline}
 and the sequence is uniformly integrable (cf. \eqref{variartz}),
 almost sure convergence implies convergence in $L_1$. 
 
 \medskip

Now to prove \eqref{statz3}, we set 
$$\cM_M(\hat \phi):= \max_{x\in \partial \gL_M} |\hat \phi_x|.$$ 
As the covariance function of $\phi$ is positive, we have 
\begin{equation}\label{variartz}
\bE_M\left[ T(\hat \phi, \phi)^2\right]\le  \cM_M ^2 \bE_M \left[ \left(\sum_{x\in \partial\gL_M}\sumtwo{y\in \partial^- \gL_M}{y\sim x}\phi(y)\right)^2\right]
=4(M-1)\cM_M ^2.
\end{equation}
We define
$$A_M:=\big\{\ |T(\hat \phi, \phi)|\le M^{7/4}\cM_M \big\} $$
Combining our bound on the variance and standard Gaussian estimates, we obtain
\begin{equation}\label{deviats}
\begin{split}
\bP_M\left[ A_M  \right]&\le e^{-c M^{5/2}},\\ 
\bE_M\left[ e^{T(\hat \phi, \phi)} \ind_{A_M}\right]&\le e^{-c M^{5/2}}.
\end{split}
\end{equation}
Combining the second line of \eqref{deviats} with an annealed bound we obtain that
\begin{equation}
 \lim_{M\to \infty} \frac{1}{M^2} \log  \bE_M\left[ e^{\sum_{x\in \tilde \gL_M} \left( \gb\go_x-\gl(\gb)+h \right)\delta^{u}_x+ 
T(\hat \phi,\phi) } \ind_{A^{\complement}_M}\right]=-\infty,
\end{equation}
and hence \eqref{statz3} is equivalent to
\begin{equation}\label{statz2}
 \lim_{M\to \infty} \frac{1}{M^2} \log  \bE_M\left[ e^{\sum_{x\in \tilde \gL_M} \left( \gb\go_x-\gl(\gb)+h \right)\delta^{u}_x+ 
 T(\hat \phi,\phi)} \ind_{A_M} \right]=\tf(\gb,h).
\end{equation}
To prove \eqref{statz2}, we first note using the first line of \eqref{deviats} that \eqref{eq:freen} implies that
\begin{equation}
 \lim_{M\to \infty} \frac{1}{M^2} \log  \bE_M\left[ \exp\left(\sum_{x\in \tilde \gL_M} \left( \gb\go_x-\gl(\gb)+h \right)\delta^{u}_x\right) \ind_{A_M} \right]=\tf(\gb,h).
\end{equation}
By definition of $A_M$ we have
\begin{equation}
 \frac{1}{M^2} 
 \left| \log \frac{ \bE_M\left[ e^{\sum_{x\in \tilde \gL_M} \left( \gb\go_x-\gl(\gb)+h \right)\delta^{u}_x+T(\hat \phi,\phi)} 
 \ind_{A_M} \right]}{\bE_M\left[  e^{\sum_{x\in \tilde \gL_M} \left( \gb\go_x-\gl(\gb)+h \right)\delta^{u}_x} \ind_{A_M} \right]} \right|
 \le M^{-1/4} \cM_M.
\end{equation}
Hence to conclude we just need to show that 
\begin{equation}
 \lim_{M\to \infty} M^{-1/4} \cM_M=0.
\end{equation}
This follows from the definition of $\cM_M$ and Borel-Cantelli's Lemma.
\qed.

\section{Decomposition of the proof of Proposition \ref{mainproposition}}\label{decompo}

The overall idea for the proof is to restrict the partition function to a set of typical trajectories $\phi$ and to control 
the first two moments of the restricted partition function to get a good estimate for the expected $\log$.
However the implementation of this simple idea requires a lot of care.
We decompose the proof in three steps.

\medskip

In Section \ref{sketch}, we briefly present these steps  and combine them to obtain the proof and in Section \ref{propnineone} we perform the first step of the proof,
which is the simpler one. The two other steps need some detailed preparatory work which is only introduced in Section \ref{liminouze}.

\subsection{Sketch of proof}\label{sketch}

The first step is to show that $\cD_N$ is a typical event in order to ensure that our restriction 
to $\cD_N$ in the partition function does not cost much.

\begin{proposition}\label{thednproba}
 We can choose $K$ in a way that 
for all $m\le 1$ sufficiently large,  for all $N\ge m^{-1}|\log m|^{1/4}$, and for all realization of $\hat \phi$
 \begin{equation}
  \bP_N[\cD^{\complement}_N]\le C(\log N)^{-1/2}.
 \end{equation}
\end{proposition}
The result is not used directly in the proof of Proposition \ref{mainproposition} but is a crucial input for the proof of 
Proposition \ref{prop:boundary} below.

\medskip

The aim of the second step is to show that at a moderate cost one can restrict the zone of the interaction to a sub-box $\gL'_N$ defined by
\begin{equation}\label{defglprim}
   \gL'_N:= \bbZ^2 \cap [N(\log N)^{-1/8},N(1-(\log N)^{-1/8})]^2.
\end{equation}
The reason for which we want to make that restriction is that it is difficult to control the effect of the boundary condition (i.e.\ of $H^{m,\hat \phi}_N$) in $\partial \gL_N \setminus    \gL'_N$.
Inside  $\gL'_N$ however, due to the choice of the relative values of $m$ and $N$ in \eqref{parameters}, $H^{m,\hat \phi}_N$ is very small and has almost no effect.

\begin{proposition}\label{prop:boundary}
There exists an event $\cC_N\subset \cD_N$ satisfying 
\begin{equation} \label{lecnepetit} 
\bP[\cC^{\cc}_N]\le C (\log N)^{-1/16}
\end{equation}
 and a constant $C(\gb)$ such that
 \begin{multline}\label{primieq}
\bE^m_N \left[ \exp\left(\sum_{x\in \tilde \gL_N} \left( \gb\go_x-\gl(\gb)+h \right)\delta^{\hat \phi,u}_x\right)\ind_{\cD_N} \right]\\
\ge \bE^m_N \left[ \exp\left(\sum_{x\in \gL'_N} \left( \gb\go_x-\gl(\gb)+h \right)\delta^{\hat \phi,u}_x\right)\ind_{\cC_N} \right]-
C(\gb) (\log \log N)^4(\log N)^{\alpha-1/16}.
 \end{multline}
\end{proposition}

Finally we have to show that the expected $\log$ of the restricted partition function  in the r.h.s.\ of  \eqref{primieq} is indeed sufficiently large to compensate for the second term.
We actually only prove that this is the case for the set of good boundary conditions $\hat \phi$ which have no significant influence in the bulk of the box
\begin{equation}
\hat \cA_N:= \{ \forall  x\in \gL'_N, |H^{m,\hat\phi}_N(x)|\le 1\},
\end{equation}
and show that the contribution of bad boundary condition is irrelevant.

\medskip

We have chosen $u_h$ in a way such that the density of expected density contact is very scarce 
(the total expected number of contact in the box is a power of $\log N$, see \eqref{grelot} below),
but the unlikely event that $\phi$ has a lot of contact is sufficient to make the second moment of the partition very large.
Hence for our analysis to work, it is necessary  to restrict the partition function to trajectories which have few contacts.
We set 
\begin{equation}\begin{split}\label{defbn}
L_N&:= \sum_{x\in  \gL'_N} \delta^{\hat \phi,u}_x,\\
\cB_N&:=\cC_N\cap \left\{ L_N\le (\log N)^{\frac{\alpha+1}{2}}\right\}.
\end{split}\end{equation}

\begin{proposition}\label{prop:inside}
We have 
\begin{itemize}
 \item [(i)]  For $N$ sufficiently large 
 \begin{equation}\label{bcondition}
  \hat\bP^m[ \hat \cA^\cc_N] \le N^{-4}.
 \end{equation}
  \item [(ii)] For any $\hat \phi \notin \hat \cA_N$ 
  \begin{equation}\label{binfluence}
  \bbE \bE^m_N \left[ \exp\left(\sum_{x\in \gL'_N} \left( \gb\go_x-\gl(\gb)+h \right)\delta^{\hat \phi,u}_x\right)\ind_{\cC_N} \right]\ge -N^2 \gl(\gb)-\log 2.
 \end{equation}
\item [(iii)] There exists a constant $c>0$ such that for any $\hat \phi \in \hat \cA_N$ 
 \begin{equation}\label{gcondition}
  \bbE \log \bE^m_N \left[ \exp\left(\sum_{x\in \gL'_N} \left( \gb\go_x-\gl(\gb)+h \right)\delta^{\hat \phi,u}_x\right)\ind_{\cB_N} \right]\ge  c h (\log N)^{\alpha}-\log 2.
 \end{equation}

\end{itemize}
\end{proposition}

\begin{proof}[Proof of Proposition \ref{mainproposition}.]

Using Proposition \ref{prop:inside}, we have
\begin{multline}
  \bbE \hat \bE^{m}\log \bE^m_N \left[ \exp\left(\sum_{x\in \gL'_N} \left( \gb\go_x-\gl(\gb)+h \right)\delta^{\hat \phi,u}_x\right)\ind_{\cC_N} \right]
  \\ 
  \ge -\hat\bP^m[ \hat \cA^{\cc}_N]\left(N^2 \gl(\gb)+\log 2\right)\\
  + \bbE \hat \bE^{m}\left[ \log   \bE^m_N \left[ \exp\left(\sum_{x\in \gL'_N} \left( \gb\go_x-\gl(\gb)+h \right)\delta^{\hat \phi,u}_x\right)\ind_{\cB_N}\right] 
  \ind_{\hat \cA_N}\right]
 \\ \ge c h (\log N)^{\alpha}-1.
  \end{multline}
  Using Proposition \ref{prop:boundary} and recalling our choice of parameters  \eqref{parameters}, we have, for $h$ sufficiently small
  \begin{multline}
\bE^m_N \left[ \exp\left(\sum_{x\in \tilde \gL_N} \left( \gb\go_x-\gl(\gb)+h \right)\delta^{\hat \phi,u}_x\right)\ind_{\cD_N} \right]- K(mN)^2\\
 \ge 
c h (\log N)^{\alpha}-C(\gb)(\log \log N)^4 (\log N)^{\alpha-\frac{1}{16}}- K (\log N)^{1/2}-1\\
\ge 
(c/2) (\log N)^{\alpha-\frac{1}{20}},
\end{multline}
where in the last line we used that $\alpha-\frac{1}{20}>1/2$.
This is sufficient to conclude.

\end{proof}

\subsection{Proof of Proposition \ref{thednproba}}\label{propnineone}

Again in this proof simply write $H$ for $H^{m,\hat \phi}_N$
The proof simply relies on computing the expectation and variance of  $$\sum_{x\in \tilde \gL_N} (\phi(x)+H(x))^2.$$
We have 
\begin{equation}\label{espuin}
\bE^m_N \left[ \sum_{x\in \tilde \gL_N} [\phi(x)+H(x)]^2\right]= \bE^m_N \left[ \sum_{x\in \gL_N} \phi(x)^2\right]+\sum_{x\in \tilde \gL_N} H(x)^2.
\end{equation}
From \eqref{eq:stimagreen}, for an appropriate choice of $C$, the following holds
\begin{equation}\label{espuout}
 \frac{1}{N^2} \bE^m_N \left[ \sum_{x\in \gL_N} \phi(x)^2\right] \ge \frac{1}{2\pi} | \log m | -C\ge \frac{2f(m)}{m^2}-C.
\end{equation}
Now let us estimate the variance.
With the cancellation of odd moments of Gaussians, the expansion of the products gives 
\begin{multline}\label{eq:decompp}
\bE^m_N \left[ \left( \sum_{x\in \tilde \gL_N} (\phi(x)+H(x))^2 \right)^2\right]-\left(\bE^m_N \left[ \sum_{x\in \tilde \gL_N} (\phi(x)+H(x))^2\right]\right)^2\\
=
\bE^m_N \left[ \left( \sum_{x\in \gL_N} \phi(x)^2 \right)^2\right] - \left(\bE^m_N \left[ \sum_{x\in \gL_N} \phi(x)^2\right]\right)^2\\
+
 4\bE^m_N\left[ \sum_{x,y \in \tilde \gL_N}  \phi(x)\phi(y)H(x)H(y) \right].
\end{multline}
We treat the last term separately and first concentrate on the two firsts which correspond to the zero boundary condition case.
We have 
\begin{equation}
 \bE^m_N\left[\phi(x)^2\phi(y)^2\right]- \bE^m_N\left[\phi(x)^2\right]\bE^m_N\left[\phi(y)^2\right]= 2 \left[ G^{m,*}(x,y)\right]^2,
\end{equation}
and hence from \eqref{greensum} we can deduce that  
 \begin{multline}
 \bE^m_N \left[ \left( \sum_{x\in \gL_N} \phi(x)^2 \right)^2\right]-\left(\bE^m_N \left[ \sum_{x\in \gL_N} \phi(x)^2\right]\right)^2
 \\ = 2 \sum_{x, y\in \gL_N} G^{m,*}(x,y)
  \le C N^2 m^{-2}.
 \end{multline}
Concerning the last term in \eqref{eq:decompp}, we bound it as follows
\begin{multline}
  \bE^m_N\left[ \sum_{x,y\in \tilde \gL_N}  \phi(x)\phi(y)H(x)H(y) \right]=  \sum_{x,y\in \tilde \gL_N} G^{m,*}(x,y)H(x)H(y)\\
  \le \sum_{x\in\tilde \gL_N} H(x)^2 \sum_{y\in \gL_N} G^{m,*}(x,y)\le  C m^{-2}\sum_{x\in \tilde \gL_N} H(x)^2,
\end{multline}
where in the last inequality we used \eqref{greensum}.
This gives 
\begin{equation}\label{lavarr}
\Var_{\bP^m_N}\left( \sum_{x\in \tilde \gL_N} (\phi(x)+H(x))^2  \right) \le C  m^{-2}\left( N^2+ \sum_{x\in \tilde \gL_N} H(x)^2\right).
\end{equation}
Hence, as long as $K$ is chosen sufficiently large, using \eqref{lavarr} and \eqref{espuin}-\eqref{espuout} we obtain
\begin{multline}
 \bP^m_N\left[\frac{1}{N^2}\sum_{x\in \tilde \gL_N} (\phi(x)+H(x))^2 \le   \frac{2 f(m)}{m^2}-K \right]\\ 
 \le \frac{\Var_{\bP^m_N}\left( \sum_{x\in \tilde \gL_N} (\phi(x)+H(x))^2  \right)}{\left(\bE^m_N\left[\sum_{x\in \tilde \gL_N} (\phi(x)+H(x))^2 \right]-N^2 \left[\frac{2 f(m)}{m^2}-K\right] \right)^2}
 \\
 \le  \frac{ C m^{-2}\left( \sum_{x\in \tilde \gL_N} H(x)^2+N^{2} \right)}{\left((K-C)N^2+\sum_{x\in \tilde \gL_N} H(x)^2 \right)^2}
 \le C m^{-2} N^{-2}.
 \end{multline}
 The result thus follows for our choice for the range of $N$.
 
\qed

\section{Preliminary work for the proofs of Propositions \ref{prop:boundary} and \ref{prop:inside}}\label{liminouze}

Both proofs require a detailed knowledge on the distribution of the number of contact in $\gL_N\setminus \gL'_N$ and in $\gL'_N$.
The highly correlated structure of the field makes this kind of information difficult to obtain.

\medskip

We have chosen $u$ quite high in order to obtain a very low empirical density of contact. For this reason our problem is quite related to that of the 
study of the maximum and of the extremal process of the $2$-dimensional free field, which has been the object of numerous studies in the past 
\cite{cf:BDG, cf:BDZ, cf:Dav, cf:Mad} together with the related subject of Branching Random Walk \cite{cf:A, cf:AS, cf:HS} or Brownian Motion \cite{cf:ABK}.
We borrow two key ideas from this literature:
\begin{itemize}
 \item [(1)] The Gaussian Free Field can be written as a sum of independent fields whose correlation spread on different scales.
 This makes the process very similar to the branching random walk.
 \item [(2)] The number of point present at a height close to the expected the maximum of the field is typically much smaller than its expectation
 (that is: by a factor $\log N$) but this $\log$ factor disapears if one conditions to a typical event.
\end{itemize}

These two points are respectively developed in Section \ref{galefash} and \ref{typic}.

  \subsection{Decomposing the free field in a martingale fashion}\label{galefash}

Let us decompose the massive free field into independent fields in order to separate the different scales 
in the correlation structure. The idea of decomposing the GFF is not new was used a lot to study the extremum and there are several possible choices 
(see \cite{cf:BDG} where a coarser decomposition is introduced or 
more recently \cite{cf:BDZ}). Our choice of decomposition is made in order to have a structure similar to that present in \cite{cf:Mad}. 

\medskip

There are several possible choices for the decomposition. The advantage of the one we present below is that the kernel of 
all the fields are expressed in terms of the 
heat-kernel, for which we have good estimates (cf. Section \ref{hkernel}).
Set (recall \ref{greenff})
\begin{equation}
k:=\lfloor G^m(x,x) \rfloor
\end{equation}
(it does not depend on $x$ as $G$ is translation invariant).
We perform the decomposition of $\phi$ into a sum $k$ subfield, each of which having (roughly) unit variance. With this construction,
$\phi(x)$ is the final step of a 
centered Gaussian random walk with $k$ steps.
With this in mind we define a decreasing sequence of times $t_i$, $i\in \lint 0,k \rint$ as follows
\begin{equation}\label{defti}\begin{cases}
t_0:=\infty,\\
\int_{t_1}^{\infty}e^{-m^2t} P_t(x,x) \dd t:=1,\\
\int_{t_{i+1}}^{t_{i}}e^{-m^2t} P_t(x,x) \dd t:=1, \quad  i\in \lint 1,k-2 \rint, \\
t_k:=0.
\end{cases}
\end{equation}
This definition implies that 
\begin{equation}
\int_{0}^{t_{k-1}}e^{-m^2t} P_t(x,x) \dd t \in [1,2).
\end{equation}
From the Local Central Limit Theorem \eqref{lclt}
we can deduce that
there exists a constant $C>0$ such that 
\begin{equation}\begin{split}\label{madga}
 \sup_{i\in \lint 1,k-1 \rint} |\log t_i-  4\pi(k-i)|\le C,\\
\left |k+  \frac{1}{2\pi}\log m\right|\le C,
\end{split}\end{equation}
We define $(\xi_i)_{i\in\lint 1,k \rint}$ to be a sequence of centered Gaussian fields (we use $\bP$ to denote their joint law) indexed by $\gL_N$,
each with covariance functions given by
 \begin{equation}
Q^{*}_i(x,y):=\int_{t_{i}}^{t_{i-1}}e^{-m^2t} P^{*}_t(x,y) \dd t,
\end{equation}
and set 
\begin{equation}
 \phi_i:=\sum_{j=1}^i \xi_i.
\end{equation}
Note that the  covariance of $\phi_k$ is given by $G^{m,*}_N$ and for this reason we simply set $\phi:=\phi_k$ and work from now on this extended probability space.
For this reason we use simply $\bP$ instead of $\bP^m_N$ (this should bring no confusion as $m$ and $N$ a now fixed by \eqref{parameters}).

\medskip

Note that the distribution of the field $\xi_i$  in the bulk of $\gL_N$ is ``almost'' translation invariant and its variance is very close to one. 
When $x$ is close to the boundary $Q^*_i(x,x)$ becomes smaller, and this effect starts at distance 
$\exp(2\pi (k-i))$ from the boundary. The distance  $\exp(2\pi(k-1))$ is also the scale on which
 covariance function $Q^*_i(x,y)$  varies in the bulk.
For this reason it is useful to set
\begin{equation}\label{defjx}
j(x):= \left ( k- \left\lceil \frac{1}{2\pi}\log d(x,\partial\gL_N) \right\rceil \right)_+.
\end{equation}
As a consequence of \eqref{eq:stimagreen}, of the definition of $k$ and that $j(x)$, we have
\begin{equation}\label{kratz1}
|\bE[\phi^2(x)]-(k-j(x))|\le C.
\end{equation}
We can deduce from this an estimate of the variance of $\phi_i(x)$,  up to a $O(1)$ correction:
there exists a constant $C$ such that
\begin{equation}\label{eq:varj}
\forall x \in \gL'_N, \, \forall i\in \{0,\dots,k\}, \quad   \left| \bE[\phi^2_i(x)]- (i-j(x))_+ \right | \le  C
\end{equation}
Indeed from Lemma \ref{lem:kerestimate} $(iii)$, we have  
\begin{equation}\label{kratz2}
\int_{\infty}^{t_{j(x)}} P^*_t(x,x)\dd t\le C,
\end{equation}
As the variance of $\xi_i(x)$ is bounded by $1$ (or $2$ when $i=k$) this implies 
\begin{equation}
 \bE[\phi^2_i(x)]\le C+(i-j(x))_+. 
\end{equation}
Finally we obtain the other bound using the fact that, as the increments have variance smaller than one (ore two for the last one) we have
\begin{equation}
  \bE[\phi^2(x)]-\bE[\phi^2_i(x)]  \le k-i+1
\end{equation}
and we conclude using \eqref{kratz1}.

\subsection{The conditional expectation for the number of contact}\label{typic}

Now we are going to use the decomposition in order to obtain finer results on the structure of the field $\phi$.
The idea is to show that with high probability the trajectory of $(\phi_i(x))_{i\in \lint 0,k \rint}$ tend to stay below a given line, for all $x\in \gL_N$, 
and thus 
if $\phi(x)$ reaches a value close to the maximum of the field, then conditioned to its final point, $(\phi_i(x))_{i=0}^k$ look 
more like a Brownian excursion than like a Brownian bridge, as it ``feels'' a constraint from above. 
If one restricts to the typical event described above, this constraint yields a loss of a factor $k$ 
(hence $\log N$) in the probability of contact.

\medskip

Note that for technical reasons, points near the boundary are a bit delicate to handle and thus we choose to prove a property in a sub-box $\gL''_N$ which excludes only 
a few points of $\gL_N$. We set
\begin{equation}
\gL''_N:= \bbZ^2 \cap \left[N(\log N)^{-2}, N(1-(\log N)^{-2})\right],
\end{equation}
and
$$\gamma:= 2\sqrt{2\pi},$$
and define the event
\begin{equation}\label{defan}
\mathcal A_N= \left\{ \forall x \in \gL''_N,\ \forall i\ge j(x),\, \phi_i(x)\le \gamma(i-j(x))+   100\gamma\log \log N \right\}.
\end{equation}
We show that this event is very typical.
This is a crucial step to define the event $\cC_N$ and to estimate the probability of  $\cB_N$.

\medskip

\begin{proposition} \label{th:condfirstmom}
We have 
\begin{equation}
 \bP\left[ \mathcal A_N \right]\ge 1-(\log N)^{-99},
\end{equation}
\end{proposition}

\medskip

\begin{proof}
 We define for $i=0,\dots,k$
 \begin{equation}
 M_i:=\frac{1}{|\gL''_N|}\sum_{x\in \gL''_N} \exp\left( \gamma \phi_i(x)-\frac{\gamma^2}{2} \bE\left[\phi^2_i(x)\right] \right).
 \end{equation}
 It is trivial to check that it is a martingale for the filtration 
 \begin{equation}
  \cF_i:=\sigma(\phi_j(x), j\le i, x\in \gL''_N). 
 \end{equation}
Integrating the second inequality in \eqref{gradientas} on the interval $[t_i,\infty)$, we have for all $x,y \in \gL_N$
 \begin{equation}
|x-y|\le e^{2\pi(k-i)} \Rightarrow  \bE\left[  \left( \phi_i(x)-\phi_i(y) \right)^2\right]\le C |x-y|^2 e^{-4\pi(k-i)}.
\end{equation}
Using a union bound, this implies that for $N$ sufficiently large 
\begin{equation}\label{eq:localvar}
\bP\left[ \max_{i\in\lint 0,k-1 \rint}\max_{\left\{ (x,y)\in (\gL_N)^2 \ : \ |x-y| \le e^{2\pi(k- i)}(\log N)^{-1} \right\}}  |\phi_i(x)-\phi_i(y)| > 1 \right]
\le \frac{1}{N}.
\end{equation}
On the complement of this event, if for a fixed $x\in \gL''_N$ we have$$\phi_i(x)\ge \gamma (i-j(x)) + 100\log \log N,$$
then 
\begin{equation}\label{ineqcro}
M_j\ge \frac{1}{|\gL''_N|}\sum_{\{ y \ :  \ |y-x|\le  e^{2\pi(k- i)}(\log N)^{-1} \} }  
e^{ \gamma^2 (i-j(x))+ 100\gamma\log \log N - \frac{\gamma^2}{2} \bE[\phi^2_i(y)] }.
\end{equation}
 Now as $i\ge j(x)$, we realize that in the range of $y$ which is considered $j(y)\ge j(x)-1$ and hence from \eqref{eq:varj} we have
 $$\bE[\phi^2_i(y)] \le i-j(x)+C+1.$$  
 For this reason, if $N$ is sufficiently large, \eqref{ineqcro} implies that.
 \begin{multline}
  M_i\ge  \frac{1}{\gL''_N} \exp\left( 4\pi(k-i)+ \frac{\gamma^2}{2} (i-j(x))+100\gamma (\log \log N) \right)\\
  \le 
  C e^{-4 \pi j(x)} (\log N)^{100\gamma}\le (\log N)^{-100}.
 \end{multline}
 The last inequality is valid for $N$ sufficiently large, it is obtained by
 using the definition \eqref{defjx} and the fact that $x\in \gL''_N$ (which implies that $j(x)\le \frac{1}{\pi} \log \log N + C$).
 Using \eqref{eq:localvar} and the fact that $M$ is a martingale, we conclude that
\begin{equation}
 \bP[\cA_N]\le \frac{1}{N}+ \bP \left[ \exists i, \, M_i\ge (\log N)^{100} \right] \le  \frac{1}{N}+ (\log N)^{-100}.
\end{equation}
\end{proof}
To conclude this section, we note that conditioning on the event $\cA_N$ the probability of having a contact drops almost 
by a factor $(\log N)$, in the bulk of the box. More precisely 
\begin{lemma}\label{probacont}
There exists a constant $C$ such that 
\begin{itemize}
 \item For all $x\in \gL_N$ we have 
 \begin{equation}
   \frac{1}{C}N^{-2}(\log N)^{1+\alpha} \le \hat \bE\bE\left[ \delta^{\hat \phi,u}_x\right]\le CN^{-2}(\log N)^{1+\alpha}.
 \end{equation}
 \item For all $x\in \gL''_N$, we have
 \begin{equation}\label{greluche}
 \bE\left[ \delta^{\hat \phi,u}_x \ind_{\cA_N} \right]
 \le  C N^{-2} (\log N)^{\alpha} \left[ H(x)^2+(\log \log N)^2\right] \exp\left(\gamma H(x)- \frac{\gamma^2}{2}j(x)  \right).
\end{equation}
In particular 
\begin{equation}\label{grelot}
 \hat\bE^m  \bE\left[ \delta^{\hat \phi,u}_x \ind_{\cA_N} \right]\le  C N^{-2} (\log N)^{\alpha} (\log \log N)^2.
\end{equation}
\end{itemize}
\end{lemma}

\begin{proof}

For the first point we notice that under law $\hat \bP^m\otimes \bP$, $\phi_x+H(x)$ is distributed like an infinite volume free field and hence has covariance 
$G^{m}(x,x)\in[k,k+1)$.
For this reason if $u\ge 1$ we have 
\begin{equation}
   \hat \bE\bE\left[ \delta^{\hat \phi,u}_x\right]= \int_{u-1}^{u+1}\frac{\dd t}{\sqrt{2\pi G^m(x,x)}} e^{-\frac{-t^2}{2G^{m}(x,x)}} \le \frac{2}{\sqrt{2\pi G^m(x,x)}} e^{-\frac{-(u-1)^2}{2G^{m}(x,x)}},
\end{equation}
and the result (the upper bound, but the lower bound is proved similarly)
follows by replacing $u$ by its value, and $G^{m}(x,x)$ by the asymptotic estimate $\frac{1}{2\pi} \log N+O(1)$.

\medskip

Let us now focus on the second point.
First we note that the result is completely obvious is $H(x)\ge 4u/5$ (the l.h.s. of \eqref{greluche} is larger than one).
Hence we assume $H(x)\le 4u/5$.
Then note that
\begin{multline}
\bE\left[ \delta^{\hat \phi,u}_x \ind_{\cA_N} \right]\le 
\bP\big[ \forall i\in \lint j(x),k \rint,\ \phi_i(x)\le \gamma (i-j(x))+100(\log \log N)\ ; \\ 
\phi(x)+H(x)\in [u-1,u+1]  \big]
\end{multline}
A first step is to show that 
\begin{equation}\label{asdas}
    \bP\big[  \phi_k(x)+H(x)\in [u-1,u+1]  \big]\le C N^{-2}(\log N)^{\alpha+1} \exp\left(\gamma H(x)- \frac{\gamma^2}{2}i(x)  \right).
\end{equation}
Using the Gaussian tail estimate \eqref{gtail} and \eqref{kratz1} we have
\begin{equation}
  \bP\big[  \phi_k(x)+H(x)\in [u-1,u+1]  \big]\le  \frac{C\sqrt{k}}{u-H(x)}\exp\left(-\frac{\left(u-1-H(x)\right)^2}{2(k-j(x)+C)} \right).
\end{equation}
Note that the factor in front of the exponential is smaller than $C(\log N)^{-1/2}$ when $H(x)\le 4u/5$.
Concerning the exponential term, notice that
\begin{multline}
\frac{\left(u-1-H(x)\right)^2}{2(k-j(x)+C)}=\frac{u^2}{2k}+ \frac{u^2(j(x)-C)}{2k(k-j(x)+C)}- \frac{(1+H(x))u}{k-j(x)+C}+  \frac{(1+H(x))^2}{2(k-j(x)+C)}
\\ 
\ge  2\log N-(\alpha+3/2)(\log \log N)+\frac{\gamma^2}{2}j(x)-\gamma H(x)-C'.
\end{multline}
This yields \eqref{asdas}.
To conclude the proof we need to show that for all $t\in [u-H(x)-1,u-H(x)+1]$
\begin{multline}
  \bP\big[ \forall i\in \{ 0,\dots, k\}, \phi_i(x)\le \gamma (i-j(x))_+ +100(\log \log N)\ | \ \phi(x)=t  \big]\\
\le C(\log N)^{-1}\left( H(x)^2+(\log \log N)^2\right).
\end{multline}
We use Lemma \ref{lem:bridge}, for the re-centered walk $$\phi_i(x)-\bE[\phi_i(x) \ | \ \phi(x)=t].$$
Let $V_i=V_i(x)$ denote the variance of $\phi_i(x)$ and $V=V(x)$ that of $\phi(x)$. We have by standard properties of Gaussian variables 
$$\bE[\phi_i(x) \ | \ \phi(x)=t]=(V_i/V) t.$$
Using the bound \eqref{eq:varj}, for all the considered values of $t$ we have 
\begin{equation}
  \gamma (i-j(x))_+ +100(\log \log N)-(V_i/V)t\le 200 (\log \log N)+ |H(x)|.
\end{equation}
Hence we have 
\begin{multline}
 \bP\left[ \forall i\in \lint 0, k\rint, \phi_i(x)\le \gamma (i-j(x))_+ +100(\log \log N)\ | \ \phi(x)=t  \right]\\
 \le  \bP\big[ \ \forall i\in \{ 0,\dots, k\}, \phi_i(x)\le 200 (\log \log N)+ |H(x)| \ | \ \phi(x)=0   \ \big],
\end{multline}
and we conclude using Lemma \ref{lem:bridge}.

\end{proof}

\subsection{Proof of Proposition \ref{prop:boundary}} \label{boundarry}
We are now ready to define the event $\cC_N$.
We set 
\begin{equation}
\cC_N:=\cD_N\cap \cC'_N, 
\end{equation}
where
\begin{equation}
 \cC'_N:=\left\{  \left(
\sum_{x\in \tilde \gL_N\setminus \gL'_N}\delta^{\hat\phi,u}_x\right)\le (\log N)^{1/16}\bE \left[\sum_{x\in \tilde \gL_N\setminus \gL'_N}\delta^{\hat\phi,u}_x \ | 
\ \cA_N \right] \right\}.
\end{equation}
From Markov's inequality, it is obvious that 
\begin{equation}
 \bP[\cC^{\cc}_N \ | \ \cA_N] \le (\log N)^{-1/16},
\end{equation}
and we can conclude (provided that $N$ is large enough) by using Proposition \ref{th:condfirstmom}, that \eqref{lecnepetit} holds.
\medskip

Let us turn to the proof of \eqref{primieq}.
We want to get rid of the environment outside $\gL'_N$.
The reader can check (by computing the second derivative that can be expressed as a variance)
\begin{multline}
\gb_2 \mapsto\bbE \left[ \log  \bE\left[ \exp\left( \sum_{x\in \gL'_N} (\gb \go_x+h-\gl(\gb)) \delta^{\hat\phi,u}_x\right.\right.\right. \\+  
\left.\left.\left. \sum_{x\in \tilde \gL_N\setminus \gL'_N}(\gb_2 \go_x+h-\gl(\gb))\delta^{\hat\phi,u}_x \right) \ind_{\cD_N}  \right] \right]
\end{multline}
is convex in $\gb_2$ and has zero derivative at $0$. Hence reaches its minimum when $\gb_2$ equals zero, and
\begin{multline}
\bbE\left[ \log \bE\left[ \exp\left(\sum_{x\in \tilde \gL_N} \left( \gb\go_x-\gl(\gb)+h \right)\delta^{\hat \phi,u}_x\right) \ind_{\cD_N}  \right] \right]\\
\ge  \bbE\left[ \log  \bE\left[ \exp\left( \sum_{x\in \gL'_N} (\gb \go_x+h-\gl(\gb)) \delta^{\hat\phi,u}_x -\gl(\gb) 
 \sum_{x\in \tilde \gL_N\setminus \gL'_N}\delta^{\hat\phi,u}_x\right)\ind_{\cD_N}  \right] \right]\\
 \ge 
 \bbE\left[ \log  \bE\left[ \exp\left( \sum_{x\in \gL'_N} (\gb \go_x+h-\gl(\gb)) \delta^{\hat\phi,u}_x\right)\ind_{\cC_N}  \right] \right] \\
 - (\log N)^{1/16}\gl(\gb)\bE \left[\sum_{x\in \tilde \gL_N\setminus \gL'_N}\delta^{\hat\phi,u}_x \ | \ \cA_N \right],
\end{multline}
where the last line is obtained by restricting the expectation to $\cC_N$ in order to bound 
$(\sum_{x\in \tilde \gL_N\setminus \gL'_N}\delta^{\hat\phi,u}_x)$ from below.
Finally, using Lemma \ref{probacont} and the definition of $\gL'_N$ \eqref{defglprim} we obtain that
\begin{equation}
 \hat \bE^m\bE \left[\sum_{x\in \tilde \gL_N\setminus \gL'_N}\delta^{\hat\phi,u}_x \ | \ \cA_N\right]
 \le C(\log \log N)^4(\log N)^{\alpha-1/8},
\end{equation}
which is sufficient to conclude.
\qed

\section{Proof of Proposition \ref{prop:inside}}\label{intelinside}

\subsection{Control of bad boundary conditions: Proof of \eqref{bcondition} and \eqref{binfluence}}

We start with the easy part of the proposition: showing that the probability of a bad boundary condition is scarce \eqref{bcondition}, and that
for this reason, a quite rough bound \eqref{binfluence} is sufficient to bound their contribution to the total expectation.

\medskip

To prove \eqref{bcondition}, we use Lemma \ref{lem:kerestimate}. For a fixed $x\in \gL'_N$, we set in the next equation $d:= d(x,\partial \gL_N)$. We have 
\begin{multline}\label{smass}
 \hat \bE^m [ (H^{m,\hat \phi}_N(x))^2 ]= \int^{\infty}_0 e^{-m^2t}[ P_t(x,x)-P^*_t(x,x) ]\dd t\\ 
 \le 
 \int^{\infty}_0 \frac{C}{t} e^{-m^2t} \exp \left( -C^{-1}\min\left(\frac{d^2}{t},d\log[(d/t)+1] \right)\right)\dd t\\
 \le e^{-c'dm} \le \exp\left(-c' (\log N)^{1/8} \right).
\end{multline}
We have used in the last inequality that  $d(x,\gL_N)\ge N(\log N)^{-1/8}$ for $x\in \gL'_N$.
Hence we have for any $x\in \gL'_N$
\begin{equation}
\hat \bP^m\left[|H^{m,\hat \phi}_N(x)|\ge 1 \right]\le \exp\left(-e^{c(\log N)^{1/8}}\right),
\end{equation}
and we can conclude using a union bound.

\medskip

To prove \eqref{binfluence}, we use Jensen's inequality and obtain
 \begin{multline}
    \bbE \log \bE\left[ \exp\left(\sum_{x\in \gL'_N} \left( \gb\go_x-\gl(\gb)+h \right)\delta^{\hat \phi,u}_x\right) \ \big| \ \cC_N \right]\\
    \ge     \bbE \bE \left[\sum_{x\in \gL'_N} \left( \gb\go_x-\gl(\gb)+h \right)\delta^{\hat \phi,u}_x \ \big| \ \cC_N \right]
    \ge -\gl(\gb)N^2.
 \end{multline}
Hence the conclusion follows from $\bP[\cC_N]\ge 1/2$.

\subsection{Decomposing the proof of \eqref{gcondition}}
Proving that good boundary conditions give a good contribution to the expected $\log$ partition function \eqref{gcondition}, 
is the most delicate point.
We divide the proof in several steps.
First we want to show that conditioned on the event $\cB_N$, the expected $\log$ partition function is close to the corresponding 
annealed bound (obtained by moving the expectation w.r.t.\ $\go$ 
inside the $\log$). This result is obtained by a control of the second moment of the restricted partition function.

\medskip

\begin{lemma}\label{consad}
 For any $\hat \phi\in \hat \cA_N$ we have 
\begin{equation}\label{consaf1}
  \log \bE\left[ \exp\left(\sum_{x\in \gL'_N} \left( \gb\go_x-\gl(\gb)+h \right)\delta^{\hat \phi,u}_x\right) \ |  \ \cB_N \right]
  \ge h  \bE \left[L_{N} \ | \ \cB_N \right]- 1.
  \end{equation}
 \end{lemma}

 \medskip

The second point is to show that $\bE \left[L_{N} \ | \ \cB_N \right]$ is large. What makes this difficult is that 
$L_N$ typically does not behave like its expectation $\bE_N[L_N]$ (cf. Lemma \ref{probacont})
We are going to prove that conditioned to 
$\cA_N$, $L_N$ almost behaves like its expectation. 
To prove such a statement, we will impose a restriction to the trajectories which is slightly stronger than $\cA_N$, as this makes computation easier.

\medskip

\begin{lemma}\label{sdasda}
We have for any $\hat \phi \in \hat \cA_N$
\begin{equation}\label{condexp}
 \bE \left[L_{N} \ | \ \cB_N \right]\ge c(\log N)^\alpha.
 \end{equation}
\end{lemma}

\medskip

\begin{proof}[Proof of \eqref{gcondition}]
 
 Combining \eqref{consaf1} and\eqref{condexp} We have for $\hat \phi\in \hat\cA_N$,
 \begin{multline}
   \bbE \log \bE\left[ \exp\left(\sum_{x\in \gL'_N} \left( \gb\go_x-\gl(\gb)+h \right)\delta^{\hat \phi,u}_x\right)\ind_{\cC_N}\right]\\
   \ge   \bbE \log \bE\left[ \exp\left(\sum_{x\in \gL'_N} \left( \gb\go_x-\gl(\gb)+h \right)\delta^{\hat \phi,u}_x\right) \ | \ \cB_N\right]
    + \log \bP\left[ \cB_N \right]\\
    \ge h\bE\left[ L_N \ | \ \cB_N \right] -1
    \ge c h (\log N)^{\alpha}-1.
 \end{multline}

 \end{proof}

\subsection{Proof of Lemma \ref{consad}}

\begin{proof}

First let us get a rough estimate on the probability of $\cB_N$, valid for $N$ sufficiently large
\begin{equation}\label{kkk}
 \bP[\cB^\cc_N]\le  C(\log N)^{-\frac{1-\alpha}{4}}.
\end{equation}
According to \eqref{greluche}, for all $\hat \phi\in \hat \cA_N$ 
\begin{equation}
 \bE\left[ L_N \ind_{\cA_N}\right]\le C(\log N)^{\alpha}(\log \log N)^2.
\end{equation}
Hence using the Markov inequality and the definition of $\cB_N$ \eqref{defbn}, we have 
\begin{equation}
 \bP[ \cB^\cc_N \ |  \ \cA_N] \ge  \bP[ L_N\ge (\log N)^{\frac{1+\alpha}{2}} \ |  \ \cA_N] \le C(\log N)^{-\frac{1-\alpha}{2}}(\log \log N)^2.
\end{equation}
We deduce \eqref{kkk} from the above equation, using the fact that $\bP[\cA_N]$ tends to one very fast (Proposition \ref{th:condfirstmom}).

\medskip

\noindent We continue the proof by setting,
\begin{equation}
Y_N:=\bE\left[ \exp\left(\sum_{x\in \gL'_N} \left( \gb\go_x-\gl(\gb)+h \right)\delta^{\hat \phi,u}_x\right) \ind_{\cB_N} \right],
\end{equation}
and $\zeta:=Y_N/\bE[Y_N]$. 
We can bound the first term from below using Jensen's inequality as follows 
\begin{equation}
 \bbE \left[ \log Y_N  \right]=   \log \bbE \left[  Y_N  \right]+ \bbE \log[\zeta].
\end{equation}
We have 
\begin{equation}
 \log \bbE[Y_N]=\log \bE\left[ \exp\left( h L_{N}\right)\ind_{\cB_N} \right]
 \ge h  \bE \left[L_{\alpha} \ | \ \cB_N \right]+ \log \bE[\cB_N].
\end{equation}
By \eqref{kkk}, the second term is larger than $-\log 2$.
To estimate $\bbE \log[\zeta]$ we simply compute the second moment of $\zeta$.
We have 
\begin{equation}
\bbE[\zeta^2]=
\tilde \bE_h^{\otimes 2}\left[\exp\left(\sum_{x\in \gL'_N}\chi(\gb)\delta^{(1)}_x \delta^{(2)}_x \right) \right],
\end{equation}
where $\chi(\gb):= \gl(2\gb)-2\gl(\gb)$ and 
\begin{equation}
  \frac{\dd \tilde\bP_h}{\dd \bP}(\phi):= \frac{1}{\bbE[Y_N]}\exp\left( hL_{N} \right) \ind_{\cB_N}.
\end{equation}
Note that as a consequence of the definition of $\cB_N$ for $N$ sufficiently large, the density is bounded from above as follows
$$\frac{\dd \tilde\bP_h}{\dd \bP}(\phi)\le  \frac{1}{\bP[\cB_N]}\exp\left( h (\log N)^{\frac{1+\alpha}{2}} \right)\le N^{1/4}.$$
Using the inequality 
\begin{equation}
\exp\left(\chi X\right)\le 1+[e^{\chi K}-1]X   
\end{equation}
valid for $X\in [0,K]$, we obtain
\begin{multline}\label{hophop}
 \bbE[\zeta^2]\le 1+e^{\chi(\gb)(\log N)^{\frac{1+\alpha}{2}}}\tilde \bE_h^{\otimes 2}[\delta^{(1)}_x \delta^{(2)}_x]\\
 \le 1+N^{1/2}e^{\chi(\gb)(\log N)^{\frac{1+\alpha}{2}}}\sum_{x\in \gL'_N} (\bE[\delta^{\hat\phi}_x])^2\le 1+N^{3/4}\sum_{x\in \gL'_N} (\bE[\delta^{\hat\phi}_x])^2.
\end{multline}
Note that from \eqref{eq:stimagreen} and our choice for $m$ \eqref{parameters}, the variance of $\phi$  satisfies 
\begin{equation}\label{danslaboit}
\forall x \in \gL'_N, \quad \left| G^{*,m}(x,x)+  \frac{1}{2\pi} \log m \right| \le C.
\end{equation}
Thus using our assumption on $|H(x)|\le 1$, and replacing $u$ by its value \eqref{parameters} we obtain that for all $x\in \gL'_N$ 
\begin{equation}
\bE[\delta^{\hat\phi,u}_x]^2\le \left[ \frac{2}{\sqrt{2\pi G^{*,m}(x,x)}} \exp\left(- \frac{(u-1-H(x))^2}{2G^{*,m}(x,x)} \right) \right]^2\le C N^{-4}(\log N)^{2(1+\alpha)}.
\end{equation}
Thus we deduce from \eqref{hophop} that
\begin{equation}\label{varixx}
 \bbE[\zeta^2]-1\le N^{-1}.
\end{equation}
This ensures that $\zeta$ is close to one with a large probability.
However to estimate $\bbE[\log \zeta]$, we also need some estimate on the right-tail distribution of $\log \zeta$.
We use a rather rough one
\begin{equation}\label{aaaa}
 |\log \zeta|\le \max_{x\in \gL'_N}|\gb\go_x-\gl(\gb)|.
\end{equation}
To conclude we note that for $\zeta\ge 1/2$ we have 
\begin{equation}
 \log (\zeta)+1-\zeta \ge -(\zeta-1)^2,
\end{equation}
and hence that 
\begin{equation}
  \bbE[\log \zeta]=\bbE[\log(\zeta)+1-\zeta]\ge -\bbE[(\zeta-1)^2]+\bbE\left[ (\log (\zeta) +1-\zeta)\ind_{\{ \zeta\le 1/2\}} \right].
\end{equation}
The first term in the r.h.s.\ can be controlled using  \eqref{varixx}.
By Cauchy-Schwartz, the second term is smaller in absolute value than 
\begin{equation}
 (\bbP[\zeta\le 1/2])^{1/2}\left(\bbE\left[ (\log \zeta+1-\zeta)^2 \ind_{\{ \zeta\le 1/2\}} \right] \right)^{1/2}\le (\bbP[\zeta\le 1/2])^{1/2}\left(\bbE\left[ (\log \zeta)^2\right] \right)^{1/2}  .
\end{equation}
Using Chebychev inequality together with \eqref{varixx}, we get that
$$\bbP[\zeta\le 1/2]\le 4N^{-1}.$$
Using \eqref{aaaa} and the fact that $\go$ have exponential tails (cf.\ assumption \eqref{eq:assume-gl}), we have 
\begin{equation}
 \bbE\left[ (\log \zeta)^2 \right]\le C (\log N)^2.
\end{equation}
Altogether we obtain that 
\begin{equation}
  \log \bbE[Y_N]\ge h  \bE \left[L_{N} \ | \ \cB_N \right]+ \log \bE[\cB_N]- CN^{-1/2}(\log N),
\end{equation}
and we can conclude using \eqref{kkk}.

\end{proof}

\subsection{Proof of Lemma \ref{sdasda}}

Instead of counting all the contacts, we decide to consider only a subset of them: those for which the trajectory 
$(\phi_i(x))_{i\in\lint 0,k \rint}$ stays below a given line. 
We choose the restriction to be a bit stronger than the one used in the definition of the event $\cA_N$ \eqref{defan}.
We set
\begin{equation}
\begin{split}
 \delta'_x&:=\ind_{\big\{ \left(\phi(x)-u+H(x)\right) \in [-1,1], \, \forall i\in \lint 1,k\rint, \,  \phi_i(x)\le  \frac{u i}{k}+10\big \}},\\
 L'_N&:= \sum_{x\in \gL'_N} \delta'_x.
 \end{split}
\end{equation}
Let us first show how to reduce the proof of Lemma \ref{sdasda} to a control on the two first moment of $L'_N$.
We have 
\begin{multline}\label{dsaddsadcz}
 \bE \left[L_{N} \ind_{\cB_N} \right]\ge \bE \left[L'_{N} \ind_{\cB_N}\right]
 = \bE \left[L'_{N}\right]- \bE \left[ L'_{N} \ind_{\cB^{\complement}_N}\right]\\
 \ge \bE \left[L'_{N}\right]- \sqrt{ \bE \left[ (L'_{N})^2 \right]} \sqrt{ \bP \left[ \cB^{\complement}_N \right] }. 
\end{multline}
Thus we can conclude provided that one can prove the two following bounds on the expectation and variance of $L'_N$
\begin{equation}\label{lnprime}\begin{split}
 \bE[L'_N] &\ge c(\log N)^\alpha,\\
 \bE[(L'_N)^2]&\le C(\log N)^{2\alpha}(\log \log N)^{8}.
\end{split}\end{equation}
It is then sufficient to combine these results with \eqref{dsaddsadcz} and \eqref{kkk}.
Hence we need to prove the two following results.

\medskip

\begin{lemma}\label{lesperance}
 For all $x\in \gL'_N$ and $\hat\phi\in \hat \cA_N$, we have
\begin{equation}\label{grominet}
 \bE[\delta'_x]\ge c N^{-2} (\log N)^{\alpha}.
\end{equation}
 \end{lemma}
\medskip

\begin{lemma}\label{covariancee}
 We have for all $x, y\in \gL'_N$ and $\hat\phi\in \hat \cA_N$,
 \begin{equation}\label{secondmoment}
 \bE[\delta'_{x}\delta'_{y}]\le \frac{CN^{-4} (\log N)^{2\alpha+3}(\log \log N)^8}{(j(x,y)+1)^{3/2}(k-j(x,y)+1)^3}e^{\frac{j(x,y) u^2}{2k}}.
 \end{equation}
 where 
 \begin{equation}
  j(x,y):= \left\lceil  k-\frac{1}{2\pi}\log|x-y|_+ \right\rceil.
\end{equation}
\end{lemma}

\medskip

\noindent The quantity $j(x,y)$ can be interpreted as the step around which the increments of $(\phi_i(x))_{i=1}^k$ and $(\phi_j(x))_{i=1}^k$ decorrelate.

\medskip

Before giving the details of these lemmas, let us prove \eqref{lnprime}.
The bound on the expectation follows immediately from \eqref{grominet}.
Concerning the bound on the variance, as for a fixed $l=1$, we have 
 \begin{equation}
  \#\{ \ (x,y)\in (\gL'_N)^2 \ : \ j(x,y)=l \} \le  C  N^2e^{4\pi(k-l)}=C N^4 (\log N)^{-1/2} e^{-4\pi l},
 \end{equation}
and a trivial bound of $N^4$ for the case $l=0$. Hence we have 
 \begin{multline}
 \bE[(L'_N)^2]= \sum_{x,y\in \gL'_N}  \bE[\delta'_{x}\delta'_{y}]\\ 
 \le C (\log N)^{2\alpha+3}(\log \log N)^8 \left[  (\log N)^{-3}+
  (\log N)^{-1/2}  \sum_{l=1}^k \frac{e^{-l \left(4\pi-\frac{u^2}{2k}\right)}}{(l+1)^{3/2}(k-l+1)^3} \right].
 \end{multline}
We must then control the above sum.
Note that from \eqref{parameters} and \eqref{madga} we deduce that
\begin{equation}
 \frac{u^2}{2k}-4\pi=2\pi\left(1+\alpha\right)\frac{\log \log N}{\log N}+O((\log N)^{-1}),
\end{equation}
and hence that
\begin{equation}
 \sum_{j=1}^k \frac{e^{-j \left(4\pi-\frac{u^2}{2k}\right)}}{(j+1)^{3/2}(k-j+1)^3}\le C  ( \log N)^{-\min(3, \frac{5}{2}+\alpha)}. 
\end{equation}
This implies \eqref{lnprime}. 

\subsection{Proof of Lemma \ref{lesperance}}

If $(u-H(x)-1)\ge 0$ (which is satisfied if $h$ is small enough because as $\hat \phi\in \mathcal A_N$ we have $|H(x)|\le 1$), we obtain from the expression of 
the Gaussian density
\begin{equation}
 \bP\big[\phi(x)\in [-1,1]+u-H(x) \big]\ge \frac{2}{\sqrt{2\pi G^{*,m}(x,x)}}\exp\left( -\frac{(u-H(x)+1)^2}{2G^{*,m}(x,x)} \right).
\end{equation}
Using again that $H(x)\in[-1,1]$, we obtain, using \eqref{danslaboit}
\begin{equation}
  \bP\left[\phi(x)\in [-1,1]+u-H(x)\right]\ge cN^{-2}(\log N)^{1+\alpha}.
\end{equation}
Now we can conclude provided we show that for all $t$ in the interval $[u-1-H(x),u+1-H(x)]$, we have
\begin{equation}\label{condprob}
 \bP\left[  \forall i\in \lint 1,k\rint, \,  \phi_i(x)\le  \frac{u i}{k}+10 \ | \ \phi(x)=t \right]\ge \frac{c}{\log N}. 
\end{equation}
Let us recall the notation of Section \ref{typic}: $V_i=V_i(x)$ denotes the variance of $\phi_i(x)$.
For $i\le k-1$, we have
\begin{equation}
 V_i(x)=\int_{t_i}^{\infty} e^{-m^2t}P_t(x,x)\dd t= i- \int_{t_i}^{\infty} e^{-m^2t}\left[ P_t(x,x)-P^*_t(x,x)\right]\dd t.
\end{equation}
Hence from \eqref{smass} we have 
\begin{equation}\label{bvar}
\forall x \in \gL'_N, \forall i \in \lint 1,k \rint, \quad V_i(x)\in[i-1,i+1].
\end{equation}
We can check that \eqref{bvar} and $t\in[u-2,u+2]$ implies
\begin{equation}
\frac{u i}{k}+10-(V_i/V) t\ge 1,
\end{equation}
To prove \eqref{condprob},
we use simply Lemma \ref{lem:bridge} $(ii)$ for the re-centered process
$\phi_i(x)-(V_i/V)t$. We have
\begin{multline}
 \bP\left[  \forall i\in \lint 1,k\rint, \,  \phi_i(x)\le \frac{u i}{k}+10  \ | \ \phi(x)=t \right]\\
 \ge  \bP\left[  \forall i\in \lint 1,k\rint, \,  \phi_i(x)\le 1  \ | \ \phi(x)=t \right] \ge \frac{C}{k}. 
\end{multline}
\qed

\subsection{A simplified version of Lemma \ref{covariancee}}

We replace $(\phi_i(x))_{i=1}^k$ and $(\phi_i(y))_{i=1}^k$ and their intricate correlation structure by a simplified picture. 
Let $(X^{(1)}_i)_{i=1}^k $,  $(X^{(2)}_i)_{i=1}^k$ be two walks, with IID standard Gaussian increments which are totally correlated 
until step $j$ and independent afterwards.
More formally the covariance structure is given by
\begin{equation}\label{struct}
\begin{split}
 &\bE[X^{(1)}_{i_1} X^{(2)}_{i_2}]:= \min(i_1,i_2,j),\\
 \bE[X^{(1)}_{i_1} & X^{(1)}_{i_2}]=\bE[X^{(2)}_{i_1} X^{(2)}_{i_2}]:= \min(i_1,i_2).
 \end{split}
\end{equation}
For $i\le j$ we set $X_i=X^{(1)}_{i}=X^{(2)}_{i}$.
The simplified version of \eqref{secondmoment} we are going to prove is the following
\begin{multline}
\bP\left[ \forall l\in \{1,2\}, \forall i \in \lint1 ,,k\rint, \ X^{(l)}_i \le \left( \frac{iu}{k}+10\right), \  X^{(l)}_k \in [u-2,u+2] \right]
\\ \le \frac{C(\log \log N)^8}{(j+1)^{3/2}(k-j+1)^3}\exp\left( -\frac{(k+j)u^2}{2k^2} \right).
\end{multline}
Note that we replaced  the interval $[u-H(x)-1,u-H(x)+1]$ and $[u-H(y)-1,u-H(y)+1]$ by $[u-2,u+2]$, and we also do so in the true proof of Lemma \ref{covariancee}.
This is ok since we are looking for an upper bound 
and as $\hat \phi \in \hat \cA_N$, the latter inverval includes the other two.

\medskip

\noindent The strategy is to first evaluate the probability
 $$ \bP\left[      X_j \in \dd t   \ ; \    \forall l\in \{1,2\},\  X^{(l)}_k \in [u-2,u+2] \right],$$
and then compute the cost of the constraint $X^{(l)}_i \le \frac{iu}{k}+10$ using Lemma \ref{lem:bridge} 
and the fact that conditioned to $X_j$, $X^{(1)}_k$ and $X^{(2)}_k$, the processes
$(X_i)_{i=1}^j$, $(X^{(1)}_i)_{i=j}^k$ and $(X^{(2)}_i)_{i=j}^k$ are three independent Brownian bridges.
For the first step, notice that we have
\begin{multline}
 \bP\left[      X_j \in \dd t,\  X^{(1)}_k \in \dd s_1,\  X^{(2)}_k \in \dd s_2   \right] \\
 = \frac{1}{(2\pi)^{3/2}(k-j)\sqrt{j}}
 \exp\left(-\frac{t^2}{2j}-\frac{(s_1-t)^2+(s_2-t)^2}{2(k-j)}\right) \dd t \dd s_1 \dd s_2.
\end{multline}
With the constraint $s_1,s_2\in  [u-2,u+2]$ and $t\le \left( \frac{ju}{k}+10\right)$, at the cost of loosing a constant factor we can replace $s_1$ and $s_2$ by $u-2$.
We obtain, after integrating over $s_1$ and $s_2$,
\begin{multline}\label{struct1}
 \bP\left[      X_j \in \dd t,\ X^{(1)}_k, X^{(2)}_k \in [u-2,u+2]  \right]\le \frac{C}{(k-j)\sqrt{j}}\exp\left(-\frac{t^2}{2j}-\frac{(u-2-t)^2}{k-j}\right) \dd t \\
 \le   \frac{C}{(k-j)\sqrt{j}} \exp\left( -\frac{(2k-j)u^2}{2k^2}-\left(\frac{u}{k}-\frac{2}{k-j}\right) \left(\frac{uj}{k}-t\right) \right) \dd t.
\end{multline}
Note that due to our choice for $u$ \eqref{parameters} and value of $k$ we have $\left(\frac{u}{k}-\frac{2}{k-j}\right)\in [\gamma/2,\gamma]$ 
provided that $h$ is sufficiently small
and $k-j$ is sufficiently large (and hence the term can be replaced by $\gamma/2$ at the cost of changing the value of $C$).

\medskip

Now using Lemma \ref{lem:bridge} (after re-centering the process), we obtain that
\begin{multline}\label{bridjun}
 \bP\left[ X_i \le \left(\frac{ui}{k}+10\right), \  \forall i\in\lint0,j\rint \ | \ X_j=t    \right]\\
 = \bP\left[ X_i \le \left(\frac{ui}{k}+10\right)- \frac{i t}{j},\  \forall i\in\lint 0,j\rint \  | \ X_j=0    \right]
\\ \le C j^{-1}\left( \left(\frac{uj}{k}-t\right)^2+ (\log j)^2 \right),
\end{multline}
where we have used that for  $t\le \left(\frac{uj}{k}+10\right)$ and $i\le j$ 
$$\left(\frac{ui}{k}+10\right)- \frac{i t}{j}= \frac{i}{j}\left(\frac{uj}{k}-t \right)+10\le  \left(\frac{uj}{k}-t \right)+20.$$
In the same manner we obtain that for $l\in \{1,2 \}$
\begin{multline}\label{bridj2}
 \bP\left[ X^{(l)}_i \le \frac{ui}{k}+10, \  \forall i\in\lint j,k \rint  \ | \ X_j=t, X^{(l)}_k\in[u-2,u+2] \right]\\
 \le C(k-j)^{-1}\left( \left(\frac{uj}{k}-t\right)^2+ (\log (k-j))^2\right).
\end{multline}
Hence using \eqref{struct1}-\eqref{bridjun}-\eqref{bridj2} and conditional independence we obtain that 
\begin{multline}
\bP\left[ \forall l\in \{1,2\}, \  \forall i \in \lint 1,k\rint, \ X^{(l)}_i \le \frac{iu}{k}+10\ ; \ X^{(l)}_k \in [u-2,u+2]\ ; \ X_j\in \dd t \right] \\
 \le   C (k-j)^{-3} j^{-3/2} (\log k)^6 \exp\left( -\frac{(2k-j)u^2}{2k^2}- (\gamma/2)\left(\frac{uj}{k}-t\right) \right) \dd t,
\end{multline}
which after integration over $t\le (\frac{uj}{k}+10)$ gives
\begin{multline}
 \bP\left[ \forall l\in \{1,2\} \forall i \in \lint 1,k\rint, \ X^{(l)}_i \le \frac{iu}{k}+10, X^{(l)}_k \in [u-2,u+2] \right]\\ \le 
  C (k-j)^{-3} j^{-3/2} (\log k)^6\exp\left( -\frac{ (2k-j)u^2}{2k^2} \right).
\end{multline}

\subsection{Proof of Lemma \ref{covariancee}}

Now, we are ready to handle the case were $X^{(1)}_i$ and $X^{(2)}_i$ are replaced by $\phi_i(x)$ and $\phi_i(y)$.
Some adaptation are needed since the increments of $\phi_i(x)$ and $\phi_i(y)$ have a less simple correlation structure
but the  method presented above is hopefully robust enough to endure such mild modifications.
Given $x$ and $y$ set 
\begin{equation}
 Z_i(x,y)=Z_i:=\frac{\phi_i(x)+\phi_i(y)}{2} \quad \text{and} \quad U_i:= \bbE \left[ Z_i \right]^2.
\end{equation}
Let us prove that there exists a constant $C$ such that
\begin{equation}\label{tborne}\begin{cases}
 \left|U_i-i \right|\le C, \quad \forall i \in \lint 0,j\rint,\\
 \left|U_i-\frac{i+j}{2} \right|\le C, \quad \forall i\in  \lint j,k\rint.
\end{cases}\end{equation}
To see this it is sufficient to remark that 
\begin{multline}
 U_i:= \frac{1}{4}\int_{t_i}^{\infty} e^{-m^2t}\left[P^*_t(x,x)+P^*_t(y,y)+2P^*_t(x,y)\right] \dd t\\
 =
 \frac{1}{2}\int_{t_i}^{\infty} e^{-m^2t} \left[ P_t(x,x)+P_t(x,y)\right] \dd t- r_i(x,y).
\end{multline}
where 
\begin{equation}
 r_i(x,y):= \frac{1}{4}\int_{t_i}^{\infty} e^{-m^2t}\left[(P_t-P^*_t)(x,x)+(P_t-P^*_t)(y,y)+2(P_t-P^*_t)(x,y)\right] \dd t.
\end{equation}
Using \eqref{smass}, we see that $P^*_t$ can be replaced by $P_t$ at the cost of a small correction i.e. \ that $r_i$ is small.
Using the definition of $t_i$ \eqref{defti} We have for $i\in  \lint 0,j\rint$
\begin{equation}\label{truit}
 \frac{1}{2}\int_{t_i}^{\infty} e^{-m^2t}\left[ P_t(x,x)+P_t(x,y) \right] \dd t= 
                                                    i- \int_{t_i}^{\infty} e^{-m^2t}\left[ P_t(x,x)-P_t(x,y)\right] \dd t,
 \end{equation}                             
    while for  $i\in \lint j,k\rint$   we have          
    \begin{multline}\label{struit}
             \frac{1}{2}\int_{t_i}^{\infty} e^{-m^2t} \left[P_t(x,x)+P_t(x,y)\right] \dd t\\
             = \frac{i+j}{2}- \int_{t_j}^{\infty} 
                                                   e^{-m^2t} \left[ P_t(x,x)-P_t(x,y) \right] \dd t+ \int^{t_j}_{t_i} e^{-m^2t} P_t(x,y) \dd t.
                                                    \end{multline}
The kernel estimates \eqref{croco} and \eqref{gradientas} then allow to conclude that the integrals in the r.h.s of \eqref{truit} and \eqref{struit} are 
bounded by a constant and thus that \eqref{tborne} hold.
Similarly to \eqref{struct1}, we are first going to show that we have, for all $t\le (\frac{uj}{k}+10)$,
\begin{multline}\label{crig}
 \bP\big[      Z_j \in \dd t \; \  \phi(x), \phi(y) \in [u-2,u+2]  \big]\\
 \le 
 \frac{C}{(k-j)\sqrt{j}} \exp\left( -\frac{(2k-j)u^2}{2k^2}- (\gamma/2) \left(\frac{uj}{k}-t\right) \right) \dd t.
\end{multline}
Using the independence of  $Z_j$ and $Z_k-Z_j$ and the fact that, up to correction of a constant order  their respective variance are respectively equal to 
$j$ and $(k-j)/2$ (cf \eqref{truit}-\eqref{struit}), we can obtain (provided that $(k-j)$ is large enough), similary to \eqref{struct1} that 
\begin{multline}\label{creg}
 \bP\left[      Z_j \in \dd t\ ;\ Z_k \in [u-2,u+2]  \right]\\
 \le 
 \frac{C}{\sqrt{j(k-j)}} \exp\left( -\frac{(2k-j)u^2}{2k^2}- (\gamma/2) \left(\frac{uj}{k}-t\right) \right) \dd t.
\end{multline}
Now, on top of that, we want to show that 
\begin{equation}\label{crog}
  \bP\big[ \   (\phi(x)-\phi(y))\in[-4,4]  \ | \  Z_j \in \dd t, Z_k \in [u-2,u+2] \  \big]\le C(k-j)^{-1/2}.
\end{equation}
As $(\phi(x)-\phi(y))$ is a Gaussian we can prove \eqref{crog} by showing that 
\begin{equation}\label{varparlba}
 \Var_{\bE[\cdot \ | \  Z_j, Z_k ]} \left[ \phi(x)-\phi(y) \right]\ge c(k-j),
\end{equation}
at least when $(k-j)$ is large: it implies that conditional density is bounded by $(2\pi c(k-j))^{-1/2}$ and thus that \eqref{crog} holds.
In fact we prove this bound for the variance conditioned to $\phi_j(x), \phi_j(y)$ and $Z_k$ (which is smaller as the conditioning is stronger) as it is easier to compute.

\medskip

If one sets $$Z'_i=\phi_i(x)-\phi_i(y),$$ one can remark, first using the fact that the increments of $(Z,Z')$ are independent and then using the usual 
formula for the conditional variance of Gaussian variable, that
\begin{equation}\label{fromul}
  \Var_{\bE[\cdot \ | \ \phi_j(x), \phi_j(y), Z_k ]} = \bE[ (Z'_k-Z'_j)^2]-\frac{\left(\bE\left[(Z'_k-Z'_j)(Z_k-Z_j)\right]\right)^2}{\bE[(Z_k-Z_j)^2]}.
\end{equation}
Using \eqref{smass} (to replace $P^*_t$  by $P_t$) and \eqref{croco} (to control the term $P^*_t(x,y)$) we have
\begin{multline}
  \bE[ (Z'_k-Z'_j)^2]= \int^{t_j}_0 e^{-m^2t} \left[ P^*_{t}(x,x)+P^*_t(y,y)-2P^*_t(x,y)\right]\dd t \\ 
 \ge  \int^{t_j}_0 e^{-m^2t} \left[P_{t}(x,x)+P_t(y,y)\right]\dd t - C\ge 2(k-j)-C.
\end{multline}
Obviously $\bE\left[(Z_k-Z_j)^2\right]$ is of the same order, and from \eqref{smass} again.
\begin{equation}
 |\bE\left[(Z'_k-Z'_j)(Z_k-Z_j)\right]|= \left| \frac{1}{2}\int^{t_j}_0 e^{-m^2t} (P^*_{t}(x,x)-P^*_t(y,y)) \dd t \right|\le 1.
\end{equation}
Hence combining these inequalities in \eqref{fromul} we obtain that \eqref{varparlba} holds.
To conclude the proof we need to show that 
\begin{equation}\label{pstnfut}
\bP\left[ \forall i\in [0,j], Z_i\le  \frac{uj}{k}+10 \ | \ Z_{j}=t  \right]  \le C \left[ \left(\frac{uj}{k}-t\right)^2+(\log j)^2 \right] j^{-1}  
\end{equation}
and
\begin{multline}\label{pstnfut2}
 \bP\left[ \ \forall i\in [j,k], \phi_i(x), \phi_i(y)\le  \frac{uj}{k}+10 \ | \ Z_{j}=t, \  \phi(x),\phi(y)\in[u-2,u+2] \ \right] \\
 \le \left[ \left(\frac{uj}{k}-t\right)^2+(\log j)^2 \right]^2 (k-j)^{-2}.
\end{multline}
Indeed using conditional independence we can multiply the inequalities \eqref{pstnfut} and \eqref{pstnfut2} with \eqref{crig} to obtain
\begin{multline}
  \bP[ \delta'_x\delta'_y, \  Z_{j}\in \dd t] \\ 
  \le \frac{C\left[ \left(\frac{uj}{k}-t\right)^2+(\log k)^2 \right]^3}{(k-j)^{3}j^{3/2}} 
  \exp\left( -\frac{(2k-j)u^2}{2k^2}- (\gamma/2) \left(\frac{uj}{k}-t\right) \right) \dd t,
\end{multline}
and conclude by integrating over $t$.
The proof of \eqref{pstnfut} is quite similar to that of \eqref{bridjun}.
\begin{multline}
 \bP\left[ \forall i\le j,\  Z_i\le  \frac{u i}{k}+10 \ | \  
 Z_j=t \right]\\ =\bP\left[ \forall i\le j,\  Z_i\le  \frac{u i}{k}+10-  (U_i/U_j)t  \ | \  
 Z_j=0 \right].
\end{multline}
We use \eqref{truit} to obtain for all $i\in \lint 0, j \rint$,
\begin{equation}
 \frac{u i}{k} - \frac{ U_i t}{U_j}\le \frac{U_i}{U_j}\left(\frac{u j}{k}-t  \right)+C\le
 \left(t- \frac{u i}{k}  \right)+C' 
 \end{equation}
and apply Lemma \ref{lem:bridge}, we obtain 
\begin{equation}\label{croco2}
  \bP\left[ \forall i\le j,\  Z_i\le  \frac{u i}{k}+10 \ | \  
 Z_j=t \right]\le  C j^{-1}\left( \left(\frac{uj}{k}-t\right)^2+ (\log j)^2 \right).
\end{equation}
To prove \eqref{pstnfut2} we have to be more careful as the increments of $\phi(x)$ and $\phi(y)$ are correlated.
It is more practical in the computation to condition to the constraint $(\phi_j(x),\phi_j(y))=(t_1,t_2)$ than to $Z_j=t$.
To obtain a bound we then take the maximum over the constraint $(t_1+t_2)=2t$. We consider only the case
$\phi(x)=\phi(y)=u-2$ in the conditioning as the others can be deduced by monotonicity (which follows from positive correlations in the Gaussian processes that are considered).

We can consider without loss of generality that 
\begin{equation}\label{restrict}
 u-C(k-j)\le  t_1 , t_2 \le \frac{ju}{k}+10,
\end{equation}
the upper bound is due to the conditioning, and if the lower-bound is violated, $t$ is so small that the r.h.s. of \eqref{pstnfut2} is larger than one.
Similarly to \eqref{croco2}, using \eqref{bvar} to control the value of $V_i$ we can prove
\begin{multline}\label{rick}
 \bP\left[ \forall i\in[j,k],\  \phi_i(x)\le  \frac{u i}{k}+10 \ | \  \phi_j(x)=t_1,\ ; \ \phi(x)\in[u-2,u+2] \right] \\ \le 
 C (k-j)^{-1}\left( \left(\frac{uj}{k}-t_1\right)^2+ (\log (k- j))^2 \right).
\end{multline}
Now the challenge lies in estimating the cost of the constraint $\phi_i(y)\le  (\frac{u i}{k}+10)$, on the segment $\lint j,k \rint$, knowing $\phi(y)$, $\phi_j(y)$ 
and  $\phi_i(x)$, $i\in\lint1,k \rint$. After conditioning to $\phi_j(y)$  and $(\phi_i(x))_{i\in \lint 1,k\rint}$, note that $\left(\phi_i(y)\right)_{i\in \lint j,k\rint}$ is still
a process with independent increments. Hence we can apply Lemma \ref{lem:bridge} provided we get to know the expectation and variance of these increments.
Let $V_i$ denote the conditional variance of $\phi_i(y)$ knowing $(\phi_r(x))_{r\in\lint0,k \rint}$. 
For a sequence $f_i$ (random or deterministic) indexed by the integers, we set 
\begin{equation}\label{defnabla}
\nabla f_i: = f_i-f_{i-1}.
\end{equation}
Let $T_i$ measure the correlation between $\nabla \phi_i(x)$ and 
$\nabla\phi_i(y)$.
We have 
\begin{equation}\begin{split}
\nabla V_i&= \bE[ (\nabla \phi_i(y))^2]- \bE[\nabla \phi_i(x)\nabla\phi_i(y)],\\
T_i&:= \frac{\bE[\nabla \phi_i(x)\nabla\phi_i(y)]}{\bE[ (\nabla \phi_i(y))^2]}.
\end{split}\end{equation}
Note that from \eqref{smass} we have $\bE[ (\nabla \phi_i(y))^2]\ge 1/2$,
and thus we deduce from \eqref{croco} that
\begin{equation}
 \sum_{i=j+1}^k T_i \le 2 \int^{t_j}_{t_i} P_t(x,y)\dd t\le C.
\end{equation}
Also using \eqref{croco} we obtain that for all $i\in\lint j,k \rint$
\begin{equation}\label{approx}
 \left| V_i- V_j -(i-j) \right|\le C.
\end{equation}
The conditional expectation of $\phi_i(y)$, $i\ge j$ given $\phi_j(y)$ and $(\phi_r(x))_{r\in\lint0,k\rint}$ is given by
\begin{equation}
 \bE\big[ \phi_i(y)-\phi_j(y) \ | \ (\phi_r(x))_{r\in\lint 0,k\rint}\big]= \sum_{r=j+1}^k T_r \nabla \phi_r(x).
\end{equation}
In particular this is smaller (in absolute value) than $C\log (k-j)$ on the event
\begin{equation}
 \cH(j,N,x)=\cH:=\left\{ |\nabla \phi_i(x)|\le \log (k-j), \ \forall i\in\lint j+1,k \rint \ \right\}.
\end{equation}
Note that $\cH$ is a very likely event.
We have, uniformly in $t_1$ satisfying \eqref{restrict}
\begin{equation}
 \bP\left[  \cA^{\cc} \ | \ \phi_j(x)=t_1\ ; \ \phi(x)=u-2 \right] \le \exp\left(-c (\log (k-j))^2 \right).
\end{equation}
Indeed, after conditioning, the increments $\nabla \phi_i(x)$ are Gaussian variables of variance smaller than $1$ (or $2$ for $i=k$)
and their mean, equal to $(V_i-V_j)(u-2-t_1)/V_i$,
is bounded by a uniform constant, due to the restriction \eqref{restrict}.
If one add the conditioning to $\phi(y)$ and $\phi_j(y)$ one obtains, for all $(\phi_i(x))_{i\in \lint0,k \rint} \in \cH$
\begin{multline}
\bE\left[ \phi_i(y) \ | \ (\phi_r(x))_{r\in\lint 0,k\rint}, \phi_j(y)=t_2 \ ; \  \phi(y)=u-2 \right] \\ \ge   
t_2+ \left(\frac{V_i-V_j}{V_k-V_j}\right)(u-2-t_2) +  \sum_{r=j+1}^k T_r \nabla \phi_r(x)
 \\ \ge t_2 + \left(\frac{V_i-V_j}{V_k-V_j}\right) \frac{(k-j)u}{k}- C ( \log(k-j)+1)   
\end{multline}
where to obtain the last inequality we used \eqref{restrict} and \eqref{approx} and the definition of $\cH$.
We have 
\begin{multline}
 \frac{iu}{k}- t_2 - \left(\frac{V_i-V_j}{V_k-V_j}\right) \frac{(k-j)u}{k}\\
 = \left(\frac{ju}{k}- t_2 \right)+ \frac{u}{k} \left(  \frac{(V_i-V_j(k-j)}{V_k-V_j}- (i-j) \right)
 \ge \left(\frac{ju}{k}- t_2 \right)- C.
\end{multline}
Hence, using Lemma \ref{lem:bridge}, after the necessary re-centering for the bridge conditioned to $(\phi_r(x))_{r\in\lint 0,k\rint}$
we obtain that if  $(\phi_r(x))_{r\in\lint 0,k\rint}\in \cH$ and \eqref{restrict} is satisfied we have 
\begin{multline}\label{rock}
 \bE[ \forall i \in \lint j,k \rint, \ \phi_i(y)\le u  \ | \ (\phi_r(x))_{r\in\lint 0,k\rint} \ ; \ \phi_j(y)=t_2 \ ; \ \phi(y)=u-2 ] \\
 \le C (k-j)^{-1}\left[ \left(\frac{ju}{k}- t_2 \right)^2+ C  \log(k-j) \right]^2.
\end{multline}
Using \eqref{rick} and \eqref{rock}, we obtain that 
\begin{multline}
 \bP\Big[ \forall i\in \lint j,k \rint, \phi_i(x), \phi_i(y)\le  \frac{uj}{k}+10  \\
  \big| \ \phi_j(x)=t_1,\ ; \  \phi_j(y)=t_2 \ ; \  \phi(x),\phi(y)\in[u-2,u+2]\Big]\\
\\ \le C (k-j)^{-2}\left[ \left(\frac{ju}{k}- t_1 \right)^2+ C  \log(k-j) \right]\left[ \left(\frac{ju}{k}- t_2 \right)^2+ C  \log(k-j)^2 \right]\\ 
+ \bP\left[  \cH^{\cc} \ | \ \phi_j(x)=t_1, \phi(x)=u-2 \right].
 \end{multline}
The last term is negligible when compared to the first and taking the maximum over $t_1+t_2=2t$ satisfying \eqref{restrict},
this concludes the proof of \eqref{pstnfut2}.
\qed

\medskip

\noindent {\bf Acknowldedgements:}  
The author would like to express his gratitude to Jian Ding, Giambattista Giacomin  and Thomas Madaule for various enlightening discussions.

\appendix

\section{Estimates on heat-kernels and random walks} \label{appendix}

\subsection{Proof of Lemma \ref{Greenesteem}}

To estimate the Green Function of the massive field we use a bit of potential theory.
We let $a$ denote the potential Kernel of $\gD$ in $\bbZ^2$ i.e.
\begin{equation}
 a(x):=\lim_{T\to \infty} \int_{0}^T \left( P_t(0,0)-P_t(x,0)\right)\dd t. 
\end{equation}
From \cite[Theorem 4.4.4]{cf:LL} we have 
\begin{equation}\label{potentas}
 a(x):=\frac{1}{2\pi}\log |x| + O(1).
\end{equation}
Set $a(x,y):=a(x-y)$.
Now recall that  $X$ is a continuous time random-walk on $\bbZ^2$ with generator $\gD$ and
that $P^x$ denote is law when the initial condition is $x\in \bbZ^2$, 
and $\tau_{A}$ denote the hitting time of $A$. Let $T_m$ be a Poisson variable of mean $m^{-2}$ which is independent of $X$.

\medskip

\noindent By adapting the proof of \cite[Proposition 4.6.2(b)]{cf:LL} we obtain that 
\begin{equation}\label{sdasdad}\begin{split}
 G^{m,*}(x,y)=E^x\left[ a\left(X_{\tau_{\partial\gL_N}\wedge T_m},y\right)\right]-a(x,y),\\ 
 G^{m}(x,y)=E^x\left[ a\left(X_{T_m},y\right)\right]-a(x,y). 
 \end{split}\end{equation}
Considering the case $y=x$ and when there is no boundary, it is not difficult to see that 
\begin{equation}\label{stima}
  G^{m}(x,x)=E^x\left[ a\left(X_{T_m},x\right)\right]:= -\frac{1}{2\pi}\log m+O(1).
\end{equation}
In the case $x=y$ with boundary, this is more delicate.
On one side it is easy to deduce from \eqref{sdasdad} that for some appropriate $C>0$,
\begin{equation}
 \frac{1}{2\pi}\log \left( \min\left( d(x,\partial \gL_N),m^{-1}\right)\right)-C  \le G^{m,*}(x,x)\le  -\frac{1}{2\pi}\log m+C. 
\end{equation}
What remain to prove is that  $\frac{1}{2\pi}\log d(x,\partial \gL_N)$ is an upper-bound (which is a concern only if $d(x,\gL_N)\le m^{-1}$).

\medskip

Note that the Green Function with Dirichlet boundary condition is an increasing function of the domain and a decreasing function of $m$.
Hence to obtain an upper-bound on $G^{m,*}$, we can compare it with the the variance of the massless free field 
in the half plane $\bbZ_+\times \bbZ$ at the point 
\begin{equation}\label{ddxx}
 x_d:= (d(x,\partial \gL_N),0)
\end{equation}
 that is given by
\begin{equation}
  E^{x_d}\left[ a\left(X_{\tau_{\{0\}\times \bbZ}},x_d\right)\right]\ge G^{m}(x,x) .
  \end{equation}
Now note that $\tau_{\{0\}\times \bbZ}$ is simply the hitting time of zero by one dimensional simple random walk starting from $d(x,\partial \gL_N)$.
Hence
$$P^{x_d}[\tau_{\{0\}\times \bbZ}\ge t]\le C  d(x,\partial \gL_N) t^{-1/2}.$$
As the second coordinate of $X_{\tau_{\{0\}\times \bbZ}}$ is simply the value of an independent random walk evaluated at $\tau$ we get 
that for some constant $C'$ all $u>0$
\begin{equation}
 P^{x_d}[ |X_{\tau_{\{0\}\times \bbZ}}-x_d|\ge u ]\le C (u/d). 
\end{equation}
This tail estimate, together with \eqref{potentas} is sufficient to conclude that 
\begin{equation}
 E^{x_d}\left[ a\left(X_{\tau_{\{0\}\times \bbZ}},x_d\right)\right]\le \frac{1}{2\pi}\log d(x,\partial \gL_N)+C.
\end{equation}

 \qed

\subsection{Proof of Lemma \ref{lem:kerestimate}}

 \begin{proof}
 
 Let us start with $(i)$
 The first inequality in \eqref{gradientas} 
 can be deduced from \cite[Theorem 2.3.6]{cf:LL} which is a fine estimate for $P_t(x,x)- P_t(x,y)$ in discrete time.
 
 \medskip
 
 For the second one, we notice that we can reduce the problem to proving that for any $u,v\in[0,N]$
 \begin{equation}
  \left( p^*_t(u,u)+p^*_t(v,v)-2p^*_t(u,v)\right)\le \frac{ C |u-v|^2 }{t^{3/2}},
 \end{equation}
where $p^*_t$ is the heat-kernel associated with the simple random-walk on $\lint 0,N \rint$ with Dirichlet boundary condition.
Indeed if $x$ and $y$ differ by only one coordinate, say $x_1=y_1$ we can factorize the l.h.s\ of \eqref{gradientas} by the common coordinate and obtain 
\begin{multline}
  p^*_t(x_1,x_1)\left[p^*_t(x_2,x_2)+p^*_t(y_2,y_2)-2p^*_t(x_2,y_2)\right ] \\
  \le \frac{C}{\sqrt{t}}\left[p^*_t(x_2,x_2)+p^*_t(y_2,y_2)-2p^*_t(x_2,y_2)\right ].
\end{multline}
If the two coordinates of $x$ and $y$ differ, then if we let $\varphi$ be a field with covariance function $P^*_t$,  the l.h.s\ of \eqref{gradientas} can be rewritten
as 
\begin{equation}
 \bbE[\left(\varphi_x-\varphi_y\right)^2]\le  2\left(\bbE[\left(\varphi_x-\varphi_z\right)^2]+\bbE[\left(\varphi_y-\varphi_z\right)^2]\right)
\end{equation}
and we reduce to the first case by choosing $z=(x_1,y_2)$.

\medskip

Now, by Fourier decomposition of the kernel, we have 
 \begin{equation}
  \left( p^*_t(u,u)+p^*_t(v,v)-2p^*_t(u,v)\right)=\frac{2}{N}
  \sum_{i=1}^{N-1} e^{-\gl_i t}\left [\sin\left( \frac{i\pi u}{N}\right)- \sin\left( \frac{i\pi v}{N}\right) \right]^2.
\end{equation}
where $\gl_i:= 2\left(1-\cos\left(\frac{i\pi}{N}\right)\right)$.
The sum can obviously be bounded by
\begin{equation}
 \frac{C|v-u|^2}{N^3} \sum_{i=1}^{N-1} e^{-\gl_i t} i^2,
 \end{equation}
It is a simple exercise to show that this sum is of order $k^2 t^{-3/2}$.

\medskip

For $(ii)$ we can just use large deviations estimates for $|x-y|\ge C\sqrt{t \log t}$ with $C$ chosen sufficiently large, and use the 
local central limit Theorem \cite[Theorem 2.1.1]{cf:LL} to cover the case  $|x-y|\le C\sqrt{t \log t}$.
For $(iii)$ we can compare to the half-plane case where $x=x_d$ (recall that from the argument presented before \eqref{ddxx} this gives an upper bound).
In that case we have 

\begin{equation}
P^*_t(x,x)=P[ X_t=0; \forall s\in[0,t],\ X_s\le d].
\end{equation}
where $X$ is the simple random walk on $\bbZ^2$ starting from zero. By a reflexion argument we have
\begin{equation}
 P[ X_t=0; \forall s\in[0,t],\ X_s< d]= P[ X_t=0]-P[X_t=2d]=P_t(0,0)-P_t(0,2d\be_1).
\end{equation}
The later quantity can be estimated with \cite[Theorem 2.3.6]{cf:LL}, and shown to be smaller than $2d^2/t$.
For $(iv)$ we have
 
 \begin{multline}
  P_t(x,x)- P^*_t(x,x)= P[ X_t=0; \exists s \in[0,t],\ X_s+x \in \partial\gL_N ]\\
  \le P\left[ X_t=0; \max_{s \in[0,t]}|X_s|\ge d \right].
 \end{multline}
The right-hand side is smaller than 
\begin{equation}
 4 P\left[ X_t=0; \max X^{(1)}_s \ge d \right]\le 4 P_t(2d\be_1)
\end{equation}
The later quantity can be estimated with the LCLT for large $t$  \cite[Theorem 2.3.6]{cf:LL}, or with large deviation estimates for small $t$.

 \end{proof}

 \subsection{Proof of Lemma \ref{lem:bridge}}

Let $V_i$ denote the variance of $X_i$ (without conditioning), $\nabla V_i=(V_i-V_{i-1})$ and set $V:=V_k$ ($V\in[k/2.k]$).
After conditioning to $X_k:=0$, the process $(X_i)^k_{i=1}$ remains Gaussian and centered but the covariance structure is given by  
\begin{equation}
  \bP\left[ X_i X_j  \ | \ X_k=0 \right]= \frac{V_i(V-V_j)}{V} \quad 
  \quad 0\le i\le j\le k.
\end{equation}
We denote by $\tilde \bP$ the law of the conditioned process.
We can couple this process with a Brownian Motion conditioned to $B_V=0$: a centered Brownian bridge $(B_t)_{t\in[0,V]}$, by setting $X_i:=B_{V_i}$.
Note that we have (by applying standard reflexion argument at the first hitting time of $x$)
\begin{equation}\label{reflex}
 \tilde \bP\left[ \max_{t\in [0,V]} B_t \ge x \right]= 1 - e^{-\frac{x^2}{2V}}.
\end{equation}
As the max of $B$ is larger than that of $X$ this gives the lower bound. 
To prove $(i)$, by monotonicity, we can restrict the proof to the case $x\ge (\log k)$. Two estimate the difference between \eqref{reflex} and the probability we have to estimate, we let
let $B^i$ denote the brownian bridges formed by $B$ between the $X_i$, 
$$(B^i_s)_{s\in [V_{i-1},V_{i}]}:= B_s- \frac{(s-V_{i-1})B_{V_{i-1}}+(V_{i}-s)B_{V_{i}}}{V}.$$
We have
\begin{multline}
  \tilde\bP\left[\max_{i\in\lint1,k-1\rint} X_i\le x\right]\le \tilde\bP\left[\max_{t\in[0,V]} B_t\le 2x\right]+
\sum_{i=1}^k \tilde\bP\left[\min_{s\in [V_{i},V_{i+1}]}B^i_s\le -x\right]\\
=\left( 1-e^{\frac{-2x^2}{V}}\right)+\sum_{i=1}^k \exp\left(-\frac{x^2}{2 \nabla V_i} \right)
\end{multline}
where in the last line we used \eqref{reflex} for $B$ and $B^{i}$.
This is smaller than $C x^2/k$ for some well chosen $C$.

\qed

 \section{Proof of Proposition \ref{propure}}\label{secpropure}

We use Proposition \ref{massivecompa}  to prove the lower-bound in the  asymptotic, and then briefly explain how to obtain a matching upper-bound.
First note that using \eqref{massivecompa} for $\gb=0$ and $u=0$ we obtain

\begin{equation}\label{JKJK}
 \tf(h)\ge \lim_{N\to \infty} \frac{1}{N^2} \log \bE^{m}_N \left[ e^{\sum_{x\in \tilde \gL_N} h\delta_x}\right] -f(m).
\end{equation}
Now, using Jensen's inequality we have 
\begin{equation}
 \frac{1}{N^2} \log \bE^{m}_N \left[ e^{\sum_{x\in \tilde \gL_N} h\delta_x}\right]\ge \frac{h}{N^2} \bE^{m}_N \left[ \sum_{x\in \tilde \gL_N} \delta_x \right]
 \ge h P[  \cN(\sigma_m) \in  [-1,1] ]
\end{equation}
where $\sigma_m:= \sqrt{G^m(x,x)}$ denote the standard deviation of the infinite volume massive free field and $\cN(\sigma_m)$ 
is a centered normal variable with standard deviation $\sigma_m$.
As the variance grows when $m$ tends to zero we obtain that for arbitrary $\gep>0$ for $m\le m_{\gep}$ we have 
\begin{equation}
  \tf(h)\ge (1-\gep)\frac{h}{\sqrt{2\pi}\sigma_m}-f(m).
\end{equation}
Using the above inequality for  $m=\frac{\sqrt{h}}{|\log h|}$, using \eqref{eq:variance} to estimate $\sigma_m$ and  \eqref{asymf} for $f(m)$ 
we obtain that for any $\gep$, for $h\le h_{\gep}$ sufficiently small we have 
\begin{equation}
 \tf(h)\ge h P[ \sigma_m \cN \in  [-1,1] ] - f(m)\ge \frac{h}{\sqrt{(1/2)\log h}}(1-\gep).
\end{equation}

\medskip

Concerning the upper-bound, we can show as in \cite[Equation (2.20)]{cf:GL} that $\sup_{\hat \phi} Z_{N,h}^{\hat \phi}$ is a sub-multiplicative 
function and thus that we have for every $N\ge 1$ we have
\begin{equation}
\tf(h)\le \sup_{\hat \phi} \frac{1}{N^2} \log Z^{\hat \phi}_{N,h}.
\end{equation}
We use this inequality for $$N=h^{-1/2} |\log h|^{-1}.$$
In that case, the Taylor expansion of the exponential in the partition function gives
\begin{equation}
Z^{\hat \phi}_{N,h}\le 1+ e^{(\log h)^{-2}} \bE^{\hat \phi}_N\left[ \sum_{x\in \tilde \gL_N}  \delta_x\right]\le  e^{(\log h)^{-2}} \bE_N\left[ 
\sum_{x\in \tilde \gL_N}  \delta_x\right],
\end{equation}
where in the last inequality we used that the probability for a Gaussian of a given variance  to be in $[-1,1]$ is maximized if its mean is equal to zero.
Using \eqref{eq:stimagreen} then to estimate the probability, it is a simple exercise to check that for any $\gep>0$ and $N$ large enough, we have 
\begin{equation}
\bE_N[ \sum_{x\in \tilde \gL_N}  \delta_x]\le \frac{(1+\gep) 2N^2}{ \sqrt{\log N}}.
\end{equation}
Combining all these inequality, we obtain that for $h$ sufficiently small
\begin{equation}
\tf(h) \le \frac{(1+2\gep) \sqrt{2} h}{ \sqrt{\log h}}.
\end{equation}
\qed


\begin{thebibliography}{99}

\bibitem{cf:A}
E.\ Aidekon, \emph{Convergence in law of the minimum of a branching random walk,} 
Ann. Probab. {\bf 41} (2013) 1115-1766.


\bibitem{cf:AS}
E.\ Aidekon and Z. Shi, \emph{Weak convergence for the minimal position in a branching random walk: a simple proof}  
 Period. Math. Hungar. (special issue in the honour of E. Cs\'aki and P.R\'ev\'esz) (2010)  {\bf 61} 43-54.


 \bibitem{cf:Ken} K.~S.~Alexander, \emph{The effect of disorder on
    polymer depinning transitions}, Commun. Math. Phys. {\bf 279} (2008),
     117-146. 
     
       \bibitem{cf:KZ}
   K.~S.~Alexander and N.~Zygouras, \emph{Quenched and annealed critical points in polymer pinning models}, Comm. Math. Phys. {\bf 291} (2009), 659-689.  
   
 \bibitem{cf:AZnew}
 K.~S.~Alexander and N.~Zygouras, \emph{Path properties of the disordered pinning model in the delocalized regime},
Ann. Appl. Probab. {\bf 24} (2014),  599-615.  
   
   
 \bibitem{cf:ABK} L.P.\ Arguin, A.\ Bovier, N. Kistler, \emph{The extremal process of branching Brownian motion}, Probab. Theory Related Fields
{\bf 157} (2013) 535-574.


\bibitem{cf:BT}
Q.~Berger and F.~L.~Toninelli \emph{On the Critical Point of the Random Walk Pinning Model in Dimension d=3},
          Elec. J. Probab.
    {\bf 15} (2010), 654, 683.


 \bibitem{cf:QH2}
 Q.~Berger
 and H.~Lacoin, \emph{Pinning on a defect line: characterization of marginal disorder relevance and sharp asymptotics for the critical point shift}, 
 to appear in J. Inst. Math. Jussieu.
 
 \bibitem{cf:BS1} M.~Birkner and R.~Sun, \emph{Annealed vs quenched critical points for a random walk pinning model} Ann. Inst. H. Poincar\'e Probab. Statist. {\bf 46} (2010), 414-441.  
    
\bibitem{cf:BS2} M.~Birkner and R.~Sun, \emph{Disorder relevance for the random walk pinning model in dimension 3} Ann. Inst. H. Poincar\'e Probab. Statist. {\bf 47} (2011), 259-293.

 
\bibitem{cf:BL} M. Biskup, O. Louidor,
  \emph{Extreme local extrema of two-dimensional discrete Gaussian free field,} to appear in Communications in Mathematical Physics.

   \bibitem{cf:Bcoprev} 
 E. Bolthausen, \emph{Random copolymers}, in {\sl Five lectures in correlated random systems}, CIRM Jean-Morlet Chair 2013, V. Gayrard et N. Kistler Eds.,  Springer Lecture Notes in Mathematics, to appear.
     
 
  \bibitem{cf:BB}
 E. Bolthausen and D. Brydges, \emph{Localization and decay of correlations for a pinned lattice free field in dimension two},  in: State of the Art in Probability and Statistics, Festschrift for Willem R. van Zwet, IMS Lecture Notes Vol. 36 (2001),  134-149.    
     
    \bibitem{cf:BDZ1}      E.~Bolthausen, J.-D.~Deuschel and O.~Zeitouni, \emph{Absence of a wetting transition for a pinned harmonic crystal in dimensions three and larger},   J. Math. Phys. {\bf 41} (2000),  1211-1223.
 
     
\bibitem{cf:BDG} E.\ Bolthausen J-D.\ Deuschel and G.\ Giacomin \emph{Entropic Repulsion and the Maximum of the two-dimensional harmonic cristal}
Ann. Probab. {\bf 29} (2001), 1670-1692.

\bibitem{cf:BDZ} E.\ Bolthausen J-D.\ Deuschel
 and O.\ Zeitouni \emph{Recursions and tightness for the maximum of the
discrete, two dimensional Gaussian Free Field}, Elec. Comm. Probab. {\bf 16} (2011) 114-119.

\bibitem{cf:BV}
E.~Bolthausen and Y.~Velenik, 
\emph{Critical behavior of the massless free field at the depinning transition}, 
Comm. Math. Phys. {\bf 223} (2001), 161-203. 

\bibitem{cf:brams}
M.\ Bramson, \textit{Maximal displacement of branching Brownian motion} Comm. Pure App. Math.
{\bf 31} (1978) 531-581.



   \bibitem{cf:CV}
P. Caputo and Y. Velenik, 
\emph{A note on wetting transition for gradient field},  Stoch. Proc. Appl. {\bf 87} (2000), 107-113.

\bibitem{cf:coprev}
F. Caravenna, G. Giacomin and F. L. Toninelli
\emph{Copolymers at selective interfaces: settled issues and open problems},
in Probability in Complex Physical Systems, volume 11 of Springer Proceedings in Mathematics, pages 289-311. Springer Berlin Heidelberg, 2012.     


\bibitem{cf:CCH}
 A.~Chiarini, A.~Cipriani and R.~S.~Hazra, \emph{A note on the extremal process of the supercritical Gaussian Free Field},
 arXiv:1505.05324
     
     
\bibitem{cf:CM1} L. Coquille and
    P. Milos,  \emph{A note on the discrete Gaussian free field with disordered pinning on $\bbZ^d$, $d\ge 2$},
Stoch.  Proc. and  Appl. {\bf 123} (2013) 3542–3559.
    
\bibitem{cf:CM2} L. Coquille and
    P. Milos,  \emph{A second  note on the discrete Gaussian free field with disordered pinning on $\bbZ^d$, $d\ge 2$},
(preprint) arXiv:1303.6770.

\bibitem{cf:Dav} O.\ Daviaud, \emph{Extremes of the discrete two-dimensional Gaussian Free Field},
Ann. Probab. {\bf 34} 962-986.


\bibitem{cf:DZ}  A.\ Dembo and 
        O. Zeitouni, \emph{Large Deviations Techniques and Applications,}
Stochastic Modelling and Applied Probability
Series {\bf 38}  Springer-Verlag Berlin Heidelberg, 2010.

 \bibitem{cf:DHV} B. Derrida, V. Hakim and J. Vannimenus,
  \emph{Effect of disorder on two-dimensional wetting}, J. Statist.
  Phys. {\bf 66} (1992), 1189--1213.

   \bibitem{cf:DGLT}
B. Derrida,
G. Giacomin, H. Lacoin and F. L. Toninelli, \emph{
Fractional moment bounds and disorder relevance for pinning models},
Commun. Math. Phys. {\bf 287} (2009) 867--887.

  \bibitem{cf:DZ2}   J. Ding and O. Zeitouni, \emph{Extreme values for two-d
imensional discrete Gaussian free field}     Ann. Probab.
   {\bf 42} (2014), 1480-1515.
  \bibitem{cf:Fisher}
M.~E.~Fisher, \emph{Walks, walls, wetting, and melting},
J. Statist. Phys. {\bf 34} (1984), 667-729.


\bibitem{cf:FKG}
C. M. Fortuin, J. Ginibre, and P. W. Kasteleyn,
 {\em Correlation inequalities on some partially ordered sets},
Comm. Math. Phys. {\bf 22} (1971) 89-103. 


\bibitem{cf:Notes}  G. Giacomin, \emph{Aspects of statistical mechanics of random surfaces.} Notes of the lectures given at
IHP, 2001 (available at www.proba.jussieu.fr/pageperso/giacomin/pub/IHP.ps)


 \bibitem{cf:GB} G. Giacomin, {\sl Random polymer models}, 
Imperial College Press, World Scientific (2007). 

 
\bibitem{cf:G} G. Giacomin, \emph{Disorder and critical phenomena through basic probability models}, \'Ecole d'\'et\'e de probablit\'es de Saint-Flour XL-2010, Lecture Notes in Mathematics {\bf 2025}, Springer, 2011.



\bibitem{cf:GL}  G. Giacomin and H. Lacoin \emph{Pinning and disorder relevance
for the lattice Gaussian free field}  to appear in J. Eur. Math. Soc.  

\bibitem{cf:GLT} G. Giacomin, H. Lacoin and F. L. Toninelli, \emph{
    Hierarchical pinning models, quadratic maps and quenched
    disorder}, Probab. Theory Relat. Fields {\bf 147} (2010), 185-216.

 \bibitem{cf:GLT2} G. Giacomin, H. Lacoin and F. L. Toninelli, \emph{Marginal relevance of disorder
for pinning models}, Commun. Pure Appl. Math. {\bf 63} (2010) 233-265.

\bibitem{cf:GT_cmp} G.~Giacomin and  F.~L.~Toninelli, \emph{Smoothing effect of quenched disorder on polymer depinning transitions},
Commun. Math. Phys. {\bf 266} (2006), 1--16.

 \bibitem{cf:Hcrit} A.~B.~Harris, {\it Effect of random defects on the critical behaviour of Ising models} J. Phys. C {\bf 7} (1974),
1671–1692.

\bibitem{cf:Holley} R. Holley,  {\em Remarks on the FKG Inequalities},
Commun. math. Phys. {\bf 36} (1974) 227—231.


 \bibitem{cf:HS} Y.\ Hu and Z.\ Shi, \emph{
 Minimal position and critical martingale convergence in branching random walks, and directed polymers on disordered trees},
     Ann. Probab. {\bf 37} (2009), 742-789.


\bibitem{cf:L2} H.~Lacoin, {\it New bounds for the free energy of directed polymers in dimension $1+1$ and $1+2$}, 
Commun. Math. Phys.
{\bf 294} (2010) 471-503. 

     \bibitem{cf:L}
 H.~Lacoin, \emph{The martingale approach to disorder irrelevance for pinning models}, Electron. Commun. Probab. {\bf 15}
  (2010), 418-427. 
  
  \bibitem{cf:L3} H.~Lacoin, \emph{Non-coincidence of Quenched and Annealed Connective Constants on the supercritical planar percolation cluster} Probability Theory and Related Fields {\bf 159} (2014) 777-808.



  \bibitem{cf:LL} G.~F.~Lawler and V.~Limic,
{\it Random Walk: A Modern Introduction}, Cambridge Studies in Advanced Mathematics {\bf 123}, 2010.

\bibitem{cf:Mad} T.\ Madaule, \emph{Maximum of a log-correlated Gaussian field} (preprint) arXiv:1307.1365.

\bibitem{cf:Trep}
F.~L.~Toninelli, \emph{A replica-coupling approach to disordered pinning models},
Commun. Math. Phys. {\bf 280} (2008), 389-401.

\bibitem{cf:Tcopo} F. Toninelli, \emph{Coarse graining, fractional moments and the critical slope of random copolymers}
Electronic Journal of Probability {\bf 14} (2009) 531-547. 


 \bibitem{cf:Vel} Y.~Velenik, 
\emph{Localization and delocalization of random interfaces},  
Probab. Surv. {\bf 3} (2006), 112-169. 


\bibitem{cf:YZ} A.~Yilmaz and  O.~Zeitouni, \emph{Differing Averaged and Quenched Large Deviations for Random Walks in Random Environments in Dimensions Two and Three}, Commun. Math. Phys. {\bf 300} (2010) 243-271. 

\end{thebibliography}
\end{document}